\documentclass[a4paper,12pt,reqno]{amsart}
\pdfoutput=1
\usepackage{a4}
\usepackage{graphicx,color}
\usepackage{amsthm, amssymb}
\usepackage{tikz,pgf}
\usepackage{xcolor}
\usepackage{scalerel}
\usepackage{todonotes}
\usepackage{verbatim}

\usepackage{hyperref}
\usepackage{bbm}

\def\Xint#1{\mathchoice
	{\XXint\displaystyle\textstyle{#1}}%
	{\XXint\textstyle\scriptstyle{#1}}%
	{\XXint\scriptstyle\scriptscriptstyle{#1}}%
	{\XXint\scriptscriptstyle\scriptscriptstyle{#1}}%
	\!\int}
\def\XXint#1#2#3{{\setbox0=\hbox{$#1{#2#3}{\int}$}
		\vcenter{\hbox{$#2#3$}}\kern-.5\wd0}}

\def\dashint{\Xint-}

\graphicspath{{graphs/}}

\newcommand{\vtwo}[1]{#1}%{{\color{jblue} #1}}
\renewcommand{\[}{\begin{equation}}
\renewcommand{\]}{\end{equation}}

\newtheorem{definition}{Definition}
\newtheorem{example}[definition]{Example}
\newtheorem{theorem}[definition]{Theorem}
\newtheorem{lemma}[definition]{Lemma}
\newtheorem{remark}[definition]{Remark}

\newtheorem{corollary}[definition]{Corollary}
\newtheorem{proposition}[definition]{Proposition}

\makeatletter
\@addtoreset{definition}{section}
\@addtoreset{equation}{section}
\makeatother

\def\d{\mathrm{d}}

\DeclareMathOperator{\Res}{Res}

\newcommand{\Q}{\mathbb{Q}}

\renewcommand{\H}{\mathcal{H}}
\newcommand{\V}{\mathcal{V}}
\newcommand{\Hc}{\overline{\mathcal{H}}}
\newcommand{\Hint}{\mathcal{H_{\textrm{int}}}}
\newcommand{\Hext}{\mathcal{H}_{\textrm{ext}}}
\newcommand{\ei}{\alpha}
\newcommand{\eic}{\bar{\alpha}}
\newcommand{\sigmac}{\bar{\sigma}}
\newcommand{\ib}{\iota}
\newcommand{\ep}{\pi}

\newcommand{\sgr}{\subset}
\renewcommand{\Res}{\mathbf{R}}

\newcommand{\res}{\mathrm{res}}
\newcommand{\skl}{\mathrm{skl}}
\newcommand{\sdd}{\omega^{\textrm{sd}}}
\newcommand{\rg}{G}
\newcommand{\sg}{H}
\newcommand{\cg}{{\tilde{G}}}
\newcommand{\RG}[2]{\mathbf{G}_{#1}^{#2}}
\newcommand{\opi}{\textsc{1pi}}
\newcommand{\uca}{\mathcal{A}}
\newcommand{\btg}{\mathcal{G}}

\newcommand{\hfd}{\mathcal{H}_{\textsc{ck}}}
\newcommand{\graft}{B_+}
\newcommand{\maxf}{\mathrm{maxf}}
\newcommand{\bij}{\mathbf{bij}}
\newcommand{\cop}{\Delta}
\newcommand{\rcop}{\tilde\Delta}
\newcommand{\one}{\mathbbm{1}}
\newcommand{\id}{\mathrm{id}}
\newcommand{\cou}{\epsilon}
\newcommand{\conp}{*}
\newcommand{\anti}{S}

\newcommand{\aut}{\mathrm{Aut}}
\newcommand{\nv}{V}
\newcommand{\nei}{I}
\newcommand{\nf}{F}

\newcommand{\ce}{c^e}
\newcommand{\cv}{c^v}
\newcommand{\xe}{X^e}
\newcommand{\xv}{X^v}
\newcommand{\xr}{X^\bullet}
\newcommand{\pe}{P^e}
\newcommand{\pv}{P^v}
\newcommand{\qe}{Q^e}
\newcommand{\qv}{Q^v}

\allowdisplaybreaks[4]

\begin{document}

% Johannes Thurigen - tikz definitions for 2-graph diagrams

\definecolor{jred}{rgb}{0.8,0,0}
\definecolor{jgreen}{rgb}{0,0.7,0}
\definecolor{jblue}{rgb}{0,0,0.8}

\usetikzlibrary{arrows,patterns,shapes,positioning,fit}
\tikzstyle{c} =	[coordinate]
\tikzstyle{vc} = [circle, draw=black, line width=1pt, inner sep=0pt, minimum size=1.5mm]
\tikzstyle{v} =  [circle, draw=black, line width=.2pt, fill=black, inner sep=0pt, minimum size=1mm]
\tikzstyle{m} =  [circle, draw=black, line width=.2pt, inner sep=0pt, minimum size=1.5mm]
\tikzstyle{vb} =[circle, draw=black, line width=.1pt, fill=jred, inner sep=0pt, minimum size=1mm]
\tikzstyle{vh} =[circle, draw=black, line width=.1pt, fill=jblue, inner sep=0pt, minimum size=1.5mm]
\tikzstyle{vs} =[coordinate]%[circle, draw=black, line width=.1pt, fill=jgreen, inner sep=0pt, minimum size=1mm]
\tikzstyle{e} =	[draw=jred,line width=1pt]
\tikzstyle{eb}= [draw=jgreen,line width=1pt]
\tikzstyle{es}= [dashed, line width=1pt]
\tikzstyle{ev}= [draw=jgreen,line width=1pt]
\tikzstyle{f} = [line width=0.01pt,dotted,fill=blue, fill opacity=.1]
\tikzstyle{cs}= [ellipse, draw=black, line width=.1pt, fill=jred, inner sep=1pt, minimum size=2mm]%[draw=none,fill=jred, fill opacity=.7]

% ------------ combinatorial maps (MFT diagrams) --------------

\newcommand{\edge}{{
\begin{tikzpicture}[x=2ex,y=2ex,baseline={([yshift=-.59ex]current bounding box.center)}]    
    \node [c]	(a)	at (0,0){};
    \node [c]	(b) at (1,0) {};
    \path
    (a) edge [e] (b);
\end{tikzpicture}}
}

\newcommand{\edgemap}{{
\begin{tikzpicture}[x=2ex,y=2ex,baseline={([yshift=-.59ex]current bounding box.center)}]    
\footnotesize{
    \node [vh]  (1)	at (0,-2){};
    \node [vb]  (a)	at (0,-1){};
    \node [vh]	(2) at (0,0){};
    \path
    (1) edge [ev, ->] node[c, label=left:$1'$]{} (a)
    (a) edge [ev] node[c, label=left:$2'$]{} (2);
}
\end{tikzpicture}}
}

\newcommand{\edgecompletion}{{
\begin{tikzpicture}[x=2ex,y=2ex,baseline={([yshift=-.59ex]current bounding box.center)}]    
\footnotesize{
    \node [m]   (1)	at (0,0){};
    \node [vb]	(2) at (0,-2) {};
    \path
    (1) edge [e, ->, bend right=90] node[c, label=left:$1'$]{} (2)
    (1) edge [e,  bend left=90] node[c, label=left:$2'$]{} (2);
}
\end{tikzpicture}}
}

\newcommand{\inftymap}{
\begin{tikzpicture}[x=3ex,y=3ex,baseline={([yshift=-.59ex]current bounding box.center)}] 
\footnotesize{
    \node [c]	(a) at (-1,0) {};
    \node [vh]	(b) at (0,0) {};
    \node [c]	(c) at (1,0) {};
    \path 
    (a) edge [ev, <-, bend left=45] (b)
    (b) edge [ev, bend left=45] (a)
    (b) edge [ev, bend left=45] (c)
    (c) edge [ev, <-, bend left=45] (b);
}
\end{tikzpicture}
}

\newcommand{\edgevertex}{{
\begin{tikzpicture}[x=2ex,y=2ex,baseline={([yshift=-.59ex]current bounding box.center)}]    
    \node [c]	(a)	at (0,0){};
    \node [c]	(b) at (1,0) {};
    \path
    (a) edge [e] node[v]{} (b);
\end{tikzpicture}}
}

\newcommand{\vertex}{{
\begin{tikzpicture}[x=2ex,y=2ex,baseline={([yshift=-.59ex]current bounding box.center)}]    
    \node [v]	(1)	at (.5,.5) {};
    \node [c]	(a)	at (0,1){};
    \node [c]	(b) 	at (0,0) {};
    \node [c]	(c)	at (1,0){};
    \node [c]	(d) 	at (1,1) {};
    \path
    (1) edge [e] (a)
    (1) edge [e] (b)
    (1) edge [e] (c)
    (1) edge [e] (d);
\end{tikzpicture}}
}

\newcommand{\vertexl}{{
\begin{tikzpicture}[x=4ex,y=4ex,baseline={([yshift=-.59ex]current bounding box.center)}]  
\footnotesize{
    \node [v]	(1)	at (.5,.5) {};
    \node [vb, label = right:1]	(a)	at (0,1){};
    \node [vb, label = above:2]	(b) 	at (0,0) {};
    \node [vb, label = left:6]	(c)	at (1,0){};
    \node [vb, label = below:7]	(d) 	at (1,1) {};
    \path
    (1) edge [e] (a)
    (1) edge [e] (b)
    (1) edge [e] (c)
    (1) edge [e] (d);
    }
\end{tikzpicture}}
}

\newcommand{\vertexmultil}{{
\begin{tikzpicture}[x=4ex,y=4ex,baseline={([yshift=-.59ex]current bounding box.center)}]  
\footnotesize{
    \node [v]	(1)	at (.5,.5) {};
    \node [c]	(2)	at (.5,.65) {};
    \node [c]	(3)	at (.5,.35) {};
    \node [vb, label = right:1]	(a)	at (0,1){};
    \node [vb, label = above:3]	(b) 	at (0,0) {};
    \node [vb, label = left:8]	(c)	at (1,0){};
    \node [vb, label = below:6]	(d) 	at (1,1) {};
    \path
    (2) edge [e] (a)
    (3) edge [e] (b)
    (3) edge [e] (c)
    (2) edge [e] (d);
    }
\end{tikzpicture}}
}

\newcommand{\multivertexl}{{
\begin{tikzpicture}[x=4ex,y=4ex,baseline={([yshift=-.59ex]current bounding box.center)}]  
\footnotesize{
%    \node [v]	(1)	at (.5,.5) {};
    \node [v]	(2)	at (.5,.65) {};
    \node [v]	(3)	at (.5,.35) {};
    \node [vb, label = right:1]	(a)	at (0,1){};
    \node [vb, label = above:3]	(b) 	at (0,0) {};
    \node [vb, label = left:8]	(c)	at (1,0){};
    \node [vb, label = below:6]	(d) 	at (1,1) {};
    \path
    (2) edge [e] (a)
    (3) edge [e] (b)
    (3) edge [e] (c)
    (2) edge [e] (d);
    }
\end{tikzpicture}}
}

% edge 1-loop
\newcommand{\tadpoleup}{
\begin{tikzpicture}[x=2ex,y=2ex,baseline={([yshift=-.59ex]current bounding box.center)}]    
    \node [v]	(1)	at (0,0) {};
    \node [c]	(a)	at (-1,0){};
    \node [c]	(b) at (1,0) {};
    \node [c]	(c1)at (0,1) {};
    \path
    (1) edge [e] (a)
    (1) edge [e] (b)
    (1) edge [e, bend right=90] (c1)
    (1) edge [e, bend left=90] (c1);
\end{tikzpicture}
}

\newcommand{\tadpoledown}{
\begin{tikzpicture}[x=2ex,y=2ex,baseline={([yshift=-.59ex]current bounding box.center)}]    
    \node [v]	(1)	at (0,0) {};
    \node [c]	(a)	at (-1,0){};
    \node [c]	(b) at (1,0) {};
    \node [c]	(c1)at (0,-1) {};
    \path
    (1) edge [e] (a)
    (1) edge [e] (b)
    (1) edge [e, bend right=90] (c1)
    (1) edge [e, bend left=90] (c1);
\end{tikzpicture}
}

\newcommand{\tadpoleume}{
\begin{tikzpicture}[x=2ex,y=2ex,baseline={([yshift=-.59ex]current bounding box.center)}]    
    \node [v]	(1)	at (0,0) {};
    \node [c]	(a)	at (-1,0){};
    \node [c]	(b) at (1,0) {};
    \node [m]	(c1)at (0,1) {};
    \path
    (1) edge [e] (a)
    (1) edge [e] (b)
    (1) edge [e, bend right=90] (c1)
    (1) edge [e, bend left=90] (c1);
\end{tikzpicture}
}

\newcommand{\tadpoledme}{
\begin{tikzpicture}[x=2ex,y=2ex,baseline={([yshift=-.59ex]current bounding box.center)}]    
    \node [v]	(1)	at (0,0) {};
    \node [c]	(a)	at (-1,0){};
    \node [c]	(b) at (1,0) {};
    \node [m]	(c1)at (0,-1) {};
    \path
    (1) edge [e] (a)
    (1) edge [e] (b)
    (1) edge [e, bend right=90] (c1)
    (1) edge [e, bend left=90] (c1);
\end{tikzpicture}
}

\newcommand{\tadpoleumv}{
\begin{tikzpicture}[x=2ex,y=2ex,baseline={([yshift=-.59ex]current bounding box.center)}]    
    \node [m]	(1)	at (0,0) {};
    \node [c]	(a)	at (-1,0){};
    \node [c]	(b) at (1,0) {};
    \node [c]	(c1)at (0,1) {};
    \path
    (1) edge [e] (a)
    (1) edge [e] (b)
    (1) edge [e, bend right=90] (c1)
    (1) edge [e, bend left=90] (c1);
\end{tikzpicture}
}

\newcommand{\tadpoledmv}{
\begin{tikzpicture}[x=2ex,y=2ex,baseline={([yshift=-.59ex]current bounding box.center)}]    
    \node [m]	(1)	at (0,0) {};
    \node [c]	(a)	at (-1,0){};
    \node [c]	(b) at (1,0) {};
    \node [c]	(c1)at (0,-1) {};
    \path
    (1) edge [e] (a)
    (1) edge [e] (b)
    (1) edge [e, bend right=90] (c1)
    (1) edge [e, bend left=90] (c1);
\end{tikzpicture}
}

% vertex 1-loop
\newcommand{\fishmap}{
\begin{tikzpicture}[x=5ex,y=5ex,baseline={([yshift=-.59ex]current bounding box.center)}] 
\footnotesize{
    \node [c]	(a) at (-1.5,0) {};
    \node [c]	(b) at (0,-1.5) {};
    \node [c]	(c) at (1.5,0) {};
    \node [c]	(d) at (0,1.5) {};
    \node [c]	(q) at (0,0) {};
    \path [f] (a) -- (b) -- (c) -- (d) -- cycle;
    \foreach \i/\j in {b/c,c/d,d/q,q/b,b/a,a/d}{
    \draw [ev] (\i) node[vh]{} -- (\j) node[vh]{};
    }
    \draw [ev, ->]  (a) node[vh]{} -- (-.8,.7);
}
\end{tikzpicture}
}

\newcommand{\fishmapdual}{
\begin{tikzpicture}[x=5ex,y=5ex,baseline={([yshift=-.59ex]current bounding box.center)}] 
\footnotesize{
    \node [c]	(a) at (-1.5,0) {};
    \node [c]	(b) at (0,-1.5) {};
    \node [c]	(c) at (1.5,0) {};
    \node [c]	(d) at (0,1.5) {};
    \node [c]	(q) at (0,0) {};
    \path [f] (a) -- (b) -- (c) -- (d) -- cycle;
    \foreach \i/\j in {b/c,c/d,d/q,q/b,b/a,a/d}{
    \draw [ev] (\i) node[vh]{} -- (\j) node[vh]{};
    }
    \node [v]	(A)	at (-.7,0) {};
    \node [v]	(B)	at (.7,0) {};
    \node [vc]	(C)	at (2.1,0) {};
    \node [vb, label=right:1, label=left:$1'$]	(1)	at (-1,.5){};
    \node [vb, label=above:2, label=below:$2'$]	(2) at (-1,-.5) {};
    \node [vb, label=left:3]	(3) at (0,-.7) {};
    \node [vb, label=right:8]	(4) at (0,.7) {};
    \node [vb, label=left:6, label=right:$6'$]	(6)	at (1,-.5){};
    \node [vb, label=below:7, label=above:$7'$]	(7) at (1,.5) {};
    \node [c, label=above:4]   (e) at (-.2,-.1) {};
    \node [c, label=below:5]   (f) at (.2,.1) {};
    \node [c]	(up) at (0,1.7) {};
    \node [c]	(down) at (0,-1.7) {};
    \draw [e, <-]  (1) .. controls (-1.3,1) and (-1,1.7) .. (up);
    \draw [e]  (up) .. controls (1,1.7) and (2.5,.5) .. (C);
    \draw [e]  (2) .. controls (-1.6,-1.3) and (-1,-1.7) .. (down);
    \draw [e]  (down) .. controls (1,-1.7) and (2.5,-.5) .. (C);
    \path
%    (A) edge [e, bend right=60] node [vb] {} (B)
%    (A) edge [e, bend left=60]   node [vb] {} (B)	
    (7) edge [e, bend left=80] (C)
    (6) edge [e, bend right=80] (C)
    \foreach \i/\j in {A/1,A/2,A/3,A/4,B/3,B/4,B/6,B/7}{
    (\i) edge [e] (\j)
    };
    \draw [ev, ->] (a) node[vh]{} -- (1);
}
\end{tikzpicture}
}

\newcommand{\fishdual}{
\begin{tikzpicture}[x=5ex,y=5ex,baseline={([yshift=-.59ex]current bounding box.center)}] 
\footnotesize{
    \node [v]	(A)	at (-.7,0) {};
    \node [v]	(B)	at (.7,0) {};
    \node [vc]	(C)	at (2.1,0) {};
    \node [vb, label=60:$1'$, label=240:1]	(up) at (.7,1.5) {};
    \node [vb, label=60:$2'$, label=240:2]	(down) at (.7,-1.5) {};
    \draw [e]  (A) .. controls (-1.1,1) and (-.3,1.5) .. (up);
    \draw [e, ->]  (C) .. controls (2.5,1) and (1.7,1.5) .. (up);
    \draw [e]  (A) .. controls (-1.1,-1) and (-.3,-1.5) .. (down);
    \draw [e]  (C) .. controls (2.5,-1) and (1.7,-1.5) .. (down);
    \path
    (A) edge [e, bend right=60] node [vb, label=60:5, label=240:3] {} (B)
    (A) edge [e, bend left=60]   node [vb, label=60:8, label=240:4] {} (B)
    (B) edge [e, bend right=60] node [vb, label=60:$6'$, label=240:6] {} (C)
    (B) edge [e, bend left=60]   node [vb, label=60:$7'$, label=240:7] {} (C);
}
\end{tikzpicture}
}

\newcommand{\fishl}{ % labbeled example
\begin{tikzpicture}[x=5ex,y=5ex,baseline={([yshift=-.59ex]current bounding box.center)}] 
\footnotesize{
    \node [v]	(A)	at (-.6,0) {};
    \node [v]	(B)	at (.6,0) {};
    \node [vb, label=right:1]	(1)	at (-1,.5){};
    \node [vb, label=above:2]	(2) at (-1,-.5) {};
    \node [vb, label=left:6]	(6)	at (1,-.5){};
    \node [vb, label=below:7]	(7) at (1,.5) {};
    \path
    (A) edge [e, bend right=60] node [vb, label=60:5, label=240:3] {} (B)
    (A) edge [e, bend left=60]  node [vb, label=60:8, label=240:4] {} (B)
    (A) edge [e] (1)
    (A) edge [e] (2)
    (B) edge [e] (6)
    (B) edge [e] (7);
}
\end{tikzpicture}
}

\newcommand{\fishtwob}{ % labbeled example
\begin{tikzpicture}[x=2ex,y=2ex,baseline={([yshift=-.59ex]current bounding box.center)}] 
\footnotesize{
    \node [v]	(1)	at (0,0) {};
    \node [v]	(2)	at (1.5,0) {};
    \node [vb]	(a)	at (-.5,0){};
    \node [vb]	(b) at (.5,0) {};
    \node [vb]	(c)	at (2,0){};
    \node [vb]	(d) at (1,0) {};
    \path
    (1) edge [e] (a)
    (1) edge [e] (b)
    (1) edge [e, bend right=80] (2)
    (1) edge [e, bend left=80]  (2)	
    (2) edge [e] (c)
    (2) edge [e] (d);
}
\end{tikzpicture}
}

\newcommand{\fishtwobl}{ % labbeled example
\begin{tikzpicture}[x=5ex,y=5ex,baseline={([yshift=-.59ex]current bounding box.center)}] 
\footnotesize{
    \node [v, label=-90:2]	(1)	at (0,0) {};
    \node [v, label=90:7]	(2)	at (1.5,0) {};
    \node [vb, label=above:1]	(a)	at (-.5,0){};
    \node [vb, label=above:4]	(b) at (.5,0) {};
    \node [vb, label=below:6]	(c)	at (2,0){};
    \node [vb, label=below:5]	(d) at (1,0) {};
    \path
    (1) edge [e] (a)
    (1) edge [e] (b)
    (1) edge [e, bend right=80] node [vb, label=120:3] {} (2)
    (1) edge [e, bend left=80]  node [vb, label=-60:8] {} (2)	
    (2) edge [e] (c)
    (2) edge [e] (d);
}
\end{tikzpicture}
}

\newcommand{\fishright}{ % fish right = fish left
\begin{tikzpicture}[x=2ex,y=2ex,baseline={([yshift=-.59ex]current bounding box.center)}]    
    \node [v]	(1)	at (.5,.5) {};
    \node [v]	(2)	at (1.5,.5) {};
    \node [c]	(a)	at (0,1){};
    \node [c]	(b) 	at (0,0) {};
    \node [c]	(c)	at (2,0){};
    \node [c]	(d) 	at (2,1) {};
    \path
    (1) edge [e] (a)
    (1) edge [e] (b)
    (1) edge [e, bend right=60] (2)
    (1) edge [e, bend left=60] (2)	
    (2) edge [e] (c)
    (2) edge [e] (d);
\end{tikzpicture}
}

\newcommand{\fishup}{ % fish up = fish down
\begin{tikzpicture}[x=2ex,y=2ex,baseline={([yshift=-.59ex]current bounding box.center)}]
\begin{scope}[rotate=90]    
    \node [v]	(1)	at (.5,.5) {};
    \node [v]	(2)	at (1.5,.5) {};
    \node [c]	(a)	at (0,1){};
    \node [c]	(b) 	at (0,0) {};
    \node [c]	(c)	at (2,0){};
    \node [c]	(d) 	at (2,1) {};
    \path
    (1) edge [e] (a)
    (1) edge [e] (b)
    (1) edge [e, bend right=60] (2)
    (1) edge [e, bend left=60] (2)	
    (2) edge [e] (c)
    (2) edge [e] (d);
\end{scope}
\end{tikzpicture}
}

\newcommand{\fishtwobb}{ % fish up = fish down
\begin{tikzpicture}[x=2ex,y=2ex,baseline={([yshift=-.59ex]current bounding box.center)}]
\begin{scope}[rotate=90]    
    \node [v]	(1)	at (.5,.5) {};
    \node [v]	(2)	at (1.5,.5) {};
    \node [c]	(a)	at (0,1){};
    \node [c]	(b) 	at (0,0) {};
    \node [c]	(c)	at (1,.3){};
    \node [c]	(d) 	at (1,.7) {};
    \path
    (1) edge [e] (a)
    (1) edge [e] (b)
    (1) edge [e, bend right=90] (2)
    (1) edge [e, bend left=90] (2)	
    (2) edge [e] (c)
    (2) edge [e] (d);
\end{scope}
\end{tikzpicture}
}

\newcommand{\fisht}{ % t channel
\begin{tikzpicture}[x=2ex,y=2ex,baseline={([yshift=-.59ex]current bounding box.center)}]
\begin{scope}[rotate=-90]    
    \node [v]	(1)	at (.5,.5) {};
    \node [v]	(2)	at (1.5,.5) {};
    \node [c]	(a)	at (2,1.5){};
    \node [c]	(b) 	at (0,0) {};
    \node [c]	(c)	at (2,0){};
    \node [c]	(d) 	at (0,1.5) {};
    \path
    (1) edge [e] (a)
    (1) edge [e] (b)
    (1) edge [e, bend right=60] (2)
    (1) edge [e, bend left=15] (2)	
    (2) edge [e] (c)
    (2) edge [e] (d);
\end{scope}
\end{tikzpicture}
}

\newcommand{\fishmvr}{
\begin{tikzpicture}[x=2ex,y=2ex,baseline={([yshift=-.59ex]current bounding box.center)}]    
    \node [v]	(1)	at (.5,.5) {};
    \node [m]	(2)	at (1.5,.5) {};
    \node [c]	(a)	at (0,1){};
    \node [c]	(b) 	at (0,0) {};
    \node [c]	(c)	at (2,0){};
    \node [c]	(d) 	at (2,1) {};
    \path
    (1) edge [e] (a)
    (1) edge [e] (b)
    (1) edge [e, bend right=60] (2)
    (1) edge [e, bend left=60] (2)	
    (2) edge [e] (c)
    (2) edge [e] (d);
\end{tikzpicture}
}

\newcommand{\fishmvl}{
\begin{tikzpicture}[x=2ex,y=2ex,baseline={([yshift=-.59ex]current bounding box.center)}]    
    \node [m]	(1)	at (.5,.5) {};
    \node [v]	(2)	at (1.5,.5) {};
    \node [c]	(a)	at (0,1){};
    \node [c]	(b) 	at (0,0) {};
    \node [c]	(c)	at (2,0){};
    \node [c]	(d) 	at (2,1) {};
    \path
    (1) edge [e] (a)
    (1) edge [e] (b)
    (1) edge [e, bend right=60] (2)
    (1) edge [e, bend left=60] (2)	
    (2) edge [e] (c)
    (2) edge [e] (d);
\end{tikzpicture}
}

\newcommand{\fishmvu}{
\begin{tikzpicture}[x=2ex,y=2ex,baseline={([yshift=-.59ex]current bounding box.center)}]
\begin{scope}[rotate=90]    
    \node [v]	(1)	at (.5,.5) {};
    \node [m]	(2)	at (1.5,.5) {};
    \node [c]	(a)	at (0,1){};
    \node [c]	(b) 	at (0,0) {};
    \node [c]	(c)	at (2,0){};
    \node [c]	(d) 	at (2,1) {};
    \path
    (1) edge [e] (a)
    (1) edge [e] (b)
    (1) edge [e, bend right=60] (2)
    (1) edge [e, bend left=60] (2)	
    (2) edge [e] (c)
    (2) edge [e] (d);
\end{scope}
\end{tikzpicture}
}

\newcommand{\fishmvd}{
\begin{tikzpicture}[x=2ex,y=2ex,baseline={([yshift=-.59ex]current bounding box.center)}]
\begin{scope}[rotate=90]    
    \node [m]	(1)	at (.5,.5) {};
    \node [v]	(2)	at (1.5,.5) {};
    \node [c]	(a)	at (0,1){};
    \node [c]	(b) 	at (0,0) {};
    \node [c]	(c)	at (2,0){};
    \node [c]	(d) 	at (2,1) {};
    \path
    (1) edge [e] (a)
    (1) edge [e] (b)
    (1) edge [e, bend right=60] (2)
    (1) edge [e, bend left=60] (2)	
    (2) edge [e] (c)
    (2) edge [e] (d);
\end{scope}
\end{tikzpicture}
}

% edge 2-loop
\newcommand{\sunrise}{
\begin{tikzpicture}[x=2ex,y=2ex,baseline={([yshift=-.59ex]current bounding box.center)}]    
\node [v]	(1)	at (0,0) {};
\node [v]	(2)	at (1.5,0) {};
\node [c]	(a)	at (-1,0){};
\node [c]	(b) at (2.5,0) {};
\path
(1) edge [e] (a)
(2) edge [e] (b)
(1) edge [e] (2)
(1) edge [e, bend right=90] (2)
(1) edge [e, bend left=90] (2);
\end{tikzpicture}
}

\newcommand{\tadpoleuu}{
\begin{tikzpicture}[x=2ex,y=2ex,baseline={([yshift=-.59ex]current bounding box.center)}]    
    \node [v]	(1)	at (0,0) {};
    \node [v]	(2)	at (0,1) {};
    \node [c]	(a)	at (-1,0){};
    \node [c]	(b) 	at (1,0) {};
    \node [c]	(c)	at (0,2) {};
    \path
    (1) edge [e] (a)
    (1) edge [e] (b)
    (1) edge [e, bend right=90] (2)
    (1) edge [e, bend left=90] (2)
    (2) edge [e, bend right=90] (c)
    (2) edge [e, bend left=90] (c);
\end{tikzpicture}
}

\newcommand{\tadpoledd}{
\begin{tikzpicture}[x=2ex,y=2ex,baseline={([yshift=-.59ex]current bounding box.center)}]    
    \node [v]	(1)	at (0,0) {};
    \node [v]	(2)	at (0,-1) {};
    \node [c]	(a)	at (-1,0){};
    \node [c]	(b) 	at (1,0) {};
    \node [c]	(c)	at (0,-2) {};
    \path
    (1) edge [e] (a)
    (1) edge [e] (b)
    (1) edge [e, bend right=90] (2)
    (1) edge [e, bend left=90] (2)
    (2) edge [e, bend right=90] (c)
    (2) edge [e, bend left=90] (c);
\end{tikzpicture}
}

\newcommand{\tadpoleud}{
\begin{tikzpicture}[x=2ex,y=2ex,baseline={([yshift=-.59ex]current bounding box.center)}]    
    \node [v]	(1)	at (0,0) {};
    \node [v]	(2)	at (0,1.5) {};
    \node [c]	(a)	at (-1,0){};
    \node [c]	(b) 	at (1,0) {};
    \node [c]	(c)	at (0,.5) {};
    \path
    (1) edge [e] (a)
    (1) edge [e] (b)
    (1) edge [e, bend right=90] (2)
    (1) edge [e, bend left=90] (2)
    (2) edge [e, bend right=90] (c)
    (2) edge [e, bend left=90] (c);
\end{tikzpicture}
}

\newcommand{\tadpoledu}{
\begin{tikzpicture}[x=2ex,y=2ex,baseline={([yshift=-.59ex]current bounding box.center)}]    
    \node [v]	(1)	at (0,0) {};
    \node [v]	(2)	at (0,-1.5) {};
    \node [c]	(a)	at (-1,0){};
    \node [c]	(b) 	at (1,0) {};
    \node [c]	(c)	at (0,-.5) {};
    \path
    (1) edge [e] (a)
    (1) edge [e] (b)
    (1) edge [e, bend right=90] (2)
    (1) edge [e, bend left=90] (2)
    (2) edge [e, bend right=90] (c)
    (2) edge [e, bend left=90] (c);
\end{tikzpicture}
}

% vertex 2-loop
\newcommand{\fishrr}{
\begin{tikzpicture}[x=2ex,y=2ex,baseline={([yshift=-.59ex]current bounding box.center)}]    
    \node [v]	(1)	at (.5,.5) {};
    \node [v]	(2)	at (1.5,.5) {};
    \node [v]	(3)	at (2.5,.5) {};
    \node [c]	(a)	at (0,1){};
    \node [c]	(b) 	at (0,0) {};
    \node [c]	(c)	at (3,0){};
    \node [c]	(d) 	at (3,1) {};
    \path
    (1) edge [e] (a)
    (1) edge [e] (b)
    (1) edge [e, bend right=60] (2)
    (1) edge [e, bend left=60] (2)	
    (2) edge [e, bend right=60] (3)
    (2) edge [e, bend left=60] (3)	
    (3) edge [e] (c)
    (3) edge [e] (d);
\end{tikzpicture}
}

\newcommand{\fishuu}{
\begin{tikzpicture}[x=2ex,y=2ex,baseline={([yshift=-.59ex]current bounding box.center)}]
\begin{scope}[rotate=90]    
    \node [v]	(1)	at (.5,.5) {};
    \node [v]	(2)	at (1.5,.5) {};
    \node [v]	(3)	at (2.5,.5) {};
    \node [c]	(a)	at (0,1){};
    \node [c]	(b) 	at (0,0) {};
    \node [c]	(c)	at (3,0){};
    \node [c]	(d) 	at (3,1) {};
    \path
    (1) edge [e] (a)
    (1) edge [e] (b)
    (1) edge [e, bend right=60] (2)
    (1) edge [e, bend left=60] (2)	
    (2) edge [e, bend right=60] (3)
    (2) edge [e, bend left=60] (3)	
    (3) edge [e] (c)
    (3) edge [e] (d);

\end{scope}
\end{tikzpicture}
}

\newcommand{\fishru}{
\begin{tikzpicture}[x=2ex,y=2ex,baseline={([yshift=-.59ex]current bounding box.center)}]    
    \node [v]	(1)	at (.5,.5) {};
    \node [v]	(2)	at (1.5,0) {};
    \node [v]	(3)	at (1.5,1) {};
    \node [c]	(a)	at (0,1){};
    \node [c]	(b) 	at (0,0) {};
    \node [c]	(c)	at (2,-.5){};
    \node [c]	(d) 	at (2,1.5) {};
    \path
    (1) edge [e] (a)
    (1) edge [e] (b)
    (1) edge [e] (2)
    (1) edge [e] (3)	
    (2) edge [e, bend right=45] (3)
    (2) edge [e, bend left=45] (3)	
    (2) edge [e] (c)
    (3) edge [e] (d);
\end{tikzpicture}
}

\newcommand{\fishur}{
\begin{tikzpicture}[x=2ex,y=2ex,baseline={([yshift=-.59ex]current bounding box.center)}]
\begin{scope}[rotate=90]    
    \node [v]	(1)	at (.5,.5) {};
    \node [v]	(2)	at (1.5,0) {};
    \node [v]	(3)	at (1.5,1) {};
    \node [c]	(a)	at (0,1){};
    \node [c]	(b) 	at (0,0) {};
    \node [c]	(c)	at (2,-.5){};
    \node [c]	(d) 	at (2,1.5) {};
    \path
    (1) edge [e] (a)
    (1) edge [e] (b)
    (1) edge [e] (2)
    (1) edge [e] (3)	
    (2) edge [e, bend right=45] (3)
    (2) edge [e, bend left=45] (3)	
    (2) edge [e] (c)
    (3) edge [e] (d);
\end{scope}
\end{tikzpicture}
}

\newcommand{\fishlu}{
\begin{tikzpicture}[x=2ex,y=2ex,baseline={([yshift=-.59ex]current bounding box.center)}]
\begin{scope}[rotate=180]    
    \node [v]	(1)	at (.5,.5) {};
    \node [v]	(2)	at (1.5,0) {};
    \node [v]	(3)	at (1.5,1) {};
    \node [c]	(a)	at (0,1){};
    \node [c]	(b) 	at (0,0) {};
    \node [c]	(c)	at (2,-.5){};
    \node [c]	(d) 	at (2,1.5) {};
    \path
    (1) edge [e] (a)
    (1) edge [e] (b)
    (1) edge [e] (2)
    (1) edge [e] (3)	
    (2) edge [e, bend right=45] (3)
    (2) edge [e, bend left=45] (3)	
    (2) edge [e] (c)
    (3) edge [e] (d);
\end{scope}
\end{tikzpicture}
}

\newcommand{\fishdr}{
\begin{tikzpicture}[x=2ex,y=2ex,baseline={([yshift=-.59ex]current bounding box.center)}]
\begin{scope}[rotate=270]    
    \node [v]	(1)	at (.5,.5) {};
    \node [v]	(2)	at (1.5,0) {};
    \node [v]	(3)	at (1.5,1) {};
    \node [c]	(a)	at (0,1){};
    \node [c]	(b) 	at (0,0) {};
    \node [c]	(c)	at (2,-.5){};
    \node [c]	(d) 	at (2,1.5) {};
    \path
    (1) edge [e] (a)
    (1) edge [e] (b)
    (1) edge [e] (2)
    (1) edge [e] (3)	
    (2) edge [e, bend right=45] (3)
    (2) edge [e, bend left=45] (3)	
    (2) edge [e] (c)
    (3) edge [e] (d);
\end{scope}
\end{tikzpicture}
}

\newcommand{\boxgraph}{
\begin{tikzpicture}[x=4ex,y=4ex,baseline={([yshift=-.59ex]current bounding box.center)}]
\begin{scope} 
    \node [v]	(0)	at (0,0)  {};
    \node [v]	(1)	at (-.6,.6) {};
    \node [v]	(2)	at (-.6,-.6) {};
    \node [v]	(3)	at (.6,-.6)  {};
    \node [v]	(4)	at (.6,.6)  {};
    \node [c]	(a)	at (-1,1)  {};
    \node [c]	(b) at (-1,-1) {};
    \node [c]	(c)	at (1,-1 ) {};
    \node [c]	(d) at (1,1)   {};
    \path
    \foreach \i  in {1,2,3,4}{
    (0) edge [e] (\i)
    }
    \foreach \i/\j  in {1/2,2/3,3/4,4/1}{
    (\i) edge [e] (\j)
    }
    \foreach \i/\j  in {a/1,b/2,c/3,d/4}{
    (\i) edge [e] (\j)
    };
\end{scope}
\end{tikzpicture}
}

\newcommand{\boxxgraph}{
\begin{tikzpicture}[x=4ex,y=4ex,baseline={([yshift=-.59ex]current bounding box.center)}]
\begin{scope}  
    \node [v]	(0)	at (0,0)     {};
    \node [v]	(1)	at (-.6,.6)  {};
    \node [v]	(2)	at (-.6,-.6) {};
    \node [v]	(3)	at (.6,-.6)  {};
    \node [v]	(4)	at (.6,.6)   {};
    \node [v]	(5)	at (1.8,-.6) {};
    \node [v]	(6)	at (1.8,.6)  {};
    \node [v]	(7)	at (1.2,0)   {};
    \node [c]	(a)	at (-1,1)    {};
    \node [c]	(b) at (-1,-1)   {};
    \node [c]	(c)	at (2.2,-1 ) {};
    \node [c]	(d) at (2.2,1)   {};
    \path
    \foreach \i  in {1,2,3,4}{
    (0) edge [e] (\i)
    }
    \foreach \i  in {3,4,5,6}{
    (7) edge [e] (\i)
    }
    \foreach \i/\j  in {1/2,2/3,3/5,5/6,6/4,4/1}{
    (\i) edge [e] (\j)
    }
    \foreach \i/\j  in {a/1,b/2,c/5,d/6}{
    (\i) edge [e] (\j)
    };
\end{scope}
\end{tikzpicture}
}

\newcommand{\boxxxxgraph}{
\begin{tikzpicture}[x=4ex,y=4ex,baseline={([yshift=-.59ex]current bounding box.center)}]
    \node [v]	(0)	at (0,0)  {};
    \node [v]	(1)	at (-.6,.6) {};
    \node [v]	(2)	at (-.6,-.6) {};
    \node [v]	(3)	at (.6,-.6)  {};
    \node [v]	(4)	at (.6,.6)  {};
    \node [c]	(a)	at (-1,1)  {};
    \node [c]	(b) at (-1,-1) {};
%    \node [c]	(c)	at (1,-1 ) {};
%    \node [c]	(d) at (1,1)   {};
\begin{scope}[xshift=10ex]
    \node [v]	(00)at (0,0)  {};
    \node [v]	(5)	at (-.6,.6) {};
    \node [v]	(6)	at (-.6,-.6) {};
    \node [v]	(7)	at (.6,-.6)  {};
    \node [v]	(8)	at (.6,.6)  {};
%    \node [c]	(e)	at (-1,1)  {};
%    \node [c]	(f) at (-1,-1) {};
    \node [c]	(g)	at (1,-1 ) {};
    \node [c]	(h) at (1,1)   {};
\end{scope}
    \path
    \foreach \i  in {1,2,3,4}{
    (0) edge [e] (\i)
    }
    \foreach \i  in {5,6,7,8}{
    (00) edge [e] (\i)
    }
    \foreach \i/\j  in {1/2,2/3,4/1,6/7,7/8,8/5}{
    (\i) edge [e] (\j)
    }
    \foreach \i/\j  in {a/1,b/2,g/7,h/8}{
    (\i) edge [e] (\j)
    }
    (3) edge [e, dotted] (6)
    (4) edge [e, dotted] (5);
\end{tikzpicture}
}

% ------------ the old example for illustration --------------

\newcommand{\mapex}{
\begin{tikzpicture}[x=4ex,y=4ex,baseline={([yshift=-.59ex]current bounding box.center)}]    
    \node [v, label=30:1]	(1)	at (-1.5,0) {};
    \node [v, label=150:2, label=-90:3, label=90:4]	(234)	at (0,0) {};
    \node [v, label=-90:5, label=90:6, label=30:7]	(789)at (1.5,0) {};
    \node [c]	(9)	at (2.25,0) {};
    \path
    (1) edge [e] (234)
    (234) edge [e, bend right=80] (789)
    (234) edge [e, bend left=80] (789)
    (789) edge [e] (9);
\end{tikzpicture}
}

\newcommand{\mapexpq}{
\begin{tikzpicture}[x=3ex,y=3ex,baseline={([yshift=-.59ex]current bounding box.center)}]    
    \node [v]	(1)	at (-1.5,0) {};
    \node [v, label=0:$q$]	(234)	at (0,0) {};
    \node [v, label=30:$p$]	(789)   at (1.5,0) {};
    \node [c]	(9)	at (2.25,0) {};
    \path
    (1) edge [e] (234)
    (234) edge [e, bend right=80] (789)
    (234) edge [e, bend left=80] (789)
    (789) edge [e] (9);
\end{tikzpicture}
}

\newcommand{\mapexstranded}{
\begin{tikzpicture}[x=4ex,y=4ex,baseline={([yshift=-.59ex]current bounding box.center)}]
    \node [vs,label=80:1]	(1a)	at (-2,.05)	{};
    \node [vs]				(1b)	at (-2,-.05)	{};
    \node [c]				(1c)	at (-3,0)	{};
    \node [c]				(1d)	at (-2.5,0)	{};	
    \node [vc, fit=(1a) (1c)] 	{};
    \node [cs,fit=(1a) (1b)] 	{};    
    \node [vs, label=30:2]	(24)	at (-1,.05)	{};
    \node [vs]				(23)	at (-1,-.05)	{};
    \node [c]				(2c)	at (1,-.05)	{};
    \begin{scope}[rotate=120]  
    \node [vs, label=below:3]	(32)	at (-1,.05)	{};
    \node [vs]				(34)	at (-1,-.05)	{};
    \end{scope}
    \begin{scope}[rotate=-120]  
    \node [vs]				(43)	at (-1,.05)	{};
    \node [vs,label=above:4]	(42)	at (-1,-.05)	{};
    \end{scope}
    \node [vc, fit=(23) (2c)] {};
    \node [cs,fit=(24) (23)] {};
    \node [cs,fit=(32) (34)] {};
    \node [cs,fit=(42) (43)] {};
    \path
    \foreach \i/\j in {24/42,32/23,43/34}{
    (\i) edge [ev,bend right=30]  (\j)
    };
    \begin{scope}[xshift=2cm]
    \node [vs, label=0:7]   (89)	at (1,.05)	{};
    \node [vs]				(87)	at (1,-.05)	{};
    \node [c]				(8c)	at (-1,-.05){};
    \begin{scope}[rotate=120]  
    \node [vs, label=80:6]  (97)	at (1,.05)	{};
    \node [vs]				(98)	at (1,-.05)	{};
    \end{scope}
    \begin{scope}[rotate=-120]  
    \node [vs]				(78)	at (1,.05)	{};
    \node [vs,label=-80:5]  (79)	at (1,-.05)	{};
    \end{scope}
    \node [vc, fit=(89) (8c)] {};
    \node [cs,fit=(98) (97)] {};
    \node [cs,fit=(79) (78)] {};
    \path
    \foreach \i/\j in {89/98,97/79,78/87}{
    (\i) edge [ev,bend left=30]  (\j)
    };
    \node [cs,fit=(89) (87)] {};
    \end{scope}    
    \path
    (1a) edge [ev, bend right=70] (1d)
    (1b) edge [ev, bend left=75]	 (1d)
    (1a) edge [ev] (24)
    (1b) edge [ev] (23)
    (42) edge [ev, bend left=60] (98)
    (43) edge [ev, bend left=60] (97)
    (34) edge [ev, bend right=60] (79)
    (32) edge [ev, bend right=60] (78);
\end{tikzpicture}
}

\newcommand{\mapexvg}{
\begin{tikzpicture}[x=4ex,y=4ex,baseline={([yshift=-.59ex]current bounding box.center)}]        
    \node [v]	(v1)	at (-2.5,0) {};
    \node [c]	(c1)	at (-3,0) 	{};
    \node [v]	(234)	at (0,0) {};
    \node [v]	(789)at (2.6,0) {};
    \node [vb, label=90:1]	(1)	at (-2,0)	{};
    \node [vb, label=90:2]	(2)	at (-1,0)	{};
    \node [vb, label=-90:3]	(3)	at (.6,-1)	{};
    \node [vb, label=90:4]	(4)	at (.6,1)	{};
    \node [vb, label=90:7]	(9)	at (3.6,0)	{};
    \node [vb, label=90:6]	(8)	at (2,1)	{};
    \node [vb, label=-90:5]	(7)	at (2,-1)	{};
    \path
    (v1) edge [e] (1)
    \foreach \i in {2,3,4}{
    (\i) edge [e] (234)
    }
    \foreach \i in {7,8,9}{
    (\i) edge [e] (789)
    }
    (1) edge [ev, bend right=90] (c1)
    (1) edge [ev, bend left=90] (c1)
    \foreach \i/\j in {2/3,3/4,4/2,7/8,8/9,9/7}{
    (\i) edge [ev] (\j)
    }
    \foreach \i/\j in {1/2,3/7,4/8}{
    (\i) edge [es] (\j)
    };
\end{tikzpicture}
}

\newcommand{\vertexexstranded}{
\begin{tikzpicture}[x=6ex,y=6ex,baseline={([yshift=-.59ex]current bounding box.center)}] 
    \node [vs]	(12)	at (1,.1)	{};
    \node [vs]	(13)	at (1,0)	{};
    \node [vs]	(14)	at (1,-.1)	{};    
    \node [vs]	(21)	at (.1,1)	{};
    \node [vs]	(24)	at (0,1)	{};
    \node [vs]	(23)	at (-.1,1)	{};    
    \node [vs]	(32)	at (-1,.1)	{};
    \node [vs]	(31)	at (-1,0)	{};
    \node [vs]	(34)	at (-1,-.1)	{};    
    \node [vs]	(41)	at (.1,-1)	{};
    \node [vs]	(42)	at (0,-1)	{};
    \node [vs]	(43)	at (-.1,-1)	{};
    \node [vc, fit=(13) (31)] {};
    \path
    \foreach \i/\j in {14/41,21/12,32/23,43/34}{
    (\i) edge [ev,bend right=40] node [vh] {} (\j)
    }
    \foreach \i/\j in {13/42, 31/24}{
    (\i) edge [ev,bend right=60] node [vh] {} (\j)
    };
    \node [cs,fit=(12) (14)] {};
    \node [cs,fit=(21) (23)] {};
    \node [cs,fit=(32) (34)] {};
    \node [cs,fit=(41) (43)] {};
\end{tikzpicture}
}

\newcommand{\vertexexgraph}{
\begin{tikzpicture}[x=6ex,y=6ex,baseline={([yshift=-.59ex]current bounding box.center)}]    
    \node [vb]	(1)	at (1,0)	{};
    \node [vb]	(2)	at (0,1)	{};
    \node [vb]	(3)	at (-1,0)	{};
    \node [vb]	(4)	at (0,-1)	{};
    \node [v, fill=black, minimum size =3mm]   (0) at (0,0)    {};
    \path
    (1) edge [ev] (2)
    (3) edge [ev] (4)
    \foreach \i in {1,2,3,4}{
    (\i) edge [e] (0)
    }
    \foreach \i/\j in {1/4,4/1,2/3,3/2}{
    (\i) edge [ev,bend right=20] node [c] {} (\j)
    };
\end{tikzpicture}
}

% ------------ 2-coloured graphs --------------

\newcommand{\rtedge}{
\begin{tikzpicture}[x=2ex,y=2ex,baseline={([yshift=-.59ex]current bounding box.center)}] 
\node [v]	(v)	at (0,0) 	{};
\node [vb]	(1)	at (-1,0)	{};
\node [vb]	(2)	at (1,0)	{};
\path
\foreach \i in {1,2}{
    (\i) edge [e] (v)
    }
(1)	edge [ev,bend left=35]	(2)
(1)	edge [ev,bend right=35]	(2);
\end{tikzpicture}
}

\newcommand{\edgevgb}{
\begin{tikzpicture}[x=2ex,y=2ex,baseline={([yshift=-.59ex]current bounding box.center)}]
\node [vb]	(1)	at (0,1)	{};
\node [vb]	(2)	at (0,-1)	{};
\path
(1)	edge [ev,bend left=30]	(2)
(1)	edge [ev,bend right=30]	(2);
\end{tikzpicture}
}

\newcommand{\edgevg}{
\begin{tikzpicture}[x=2ex,y=2ex,baseline={([yshift=-.59ex]current bounding box.center)}]
\node [vb]	(1)	at (1,0)	{};
\node [vb]	(2)	at (-1,0)	{};
\path
(1)	edge [ev,bend left=30]	(2)
(1)	edge [ev,bend right=30]	(2);
\end{tikzpicture}
}

\newcommand{\edgevgr}{
\begin{tikzpicture}[x=2ex,y=2ex,baseline={([yshift=-.59ex]current bounding box.center)}]
\node [vb]	(1)	at (1,0)	{};
\node [vb]	(2)	at (-1,0)	{};
\node at (0,0) {...};
\path
(1)	edge [ev,bend left=50] 	(2)
(1)	edge [ev,bend left=30] 	(2)
(1)	edge [ev,bend right=30] (2)
(1)	edge [ev,bend right=50]	(2);
\end{tikzpicture}
}

\newcommand{\mapthreevertex}{
\begin{tikzpicture}[x=2ex,y=2ex,baseline={([yshift=-.59ex]current bounding box.center)}] 
\node [v]	(v) at (0,0) 	{};
\node [vb]	(1)	at (-1,0)	{};
\node [vb]	(2)	at (.6,-1)	{};
\node [vb]	(3)	at (.6,1)	{};
\path
\foreach \i in {1,2,3}{
    (\i) edge [e] (v)
    }
\foreach \i/\j in {1/2,2/3,3/1}{
    (\i) edge [ev] (\j)
    };
\end{tikzpicture}
}

\newcommand{\onevg}{
\begin{tikzpicture}[x=1.2ex,y=1.2ex,baseline={([yshift=-.59ex]current bounding box.center)}] 
    \node [c]	(c1)	at (-3,0) 	{};
    \node [vb]	(1)	at (-2,0)	{};    
    \path
    (1) edge [ev, bend right=90] (c1)
    (1) edge [ev, bend left=90] (c1);
\end{tikzpicture}
}

\newcommand{\threevg}{
\begin{tikzpicture}[x=.7ex,y=.7ex,baseline={([yshift=-.59ex]current bounding box.center)}] 
\node [vb]	(1)	at (-1,0)	{};
\node [vb]	(2)	at (.6,-1)	{};
\node [vb]	(3)	at (.6,1)	{};
\path
\foreach \i/\j in {1/2,2/3,3/1}{
    (\i) edge [ev] (\j)
    };
\end{tikzpicture}
}

\newcommand{\rtvf}{
\begin{tikzpicture}[x=2ex,y=2ex,baseline={([yshift=-.59ex]current bounding box.center)}]
\scriptsize{
\node [v]	(v)	at (0,0) 	{};
\node [vb]	(1)	at (-1,1)	{};
\node [vb]	(2)	at (-1,-1)	{};
\node [vb]	(3)	at (1,-1)	{};
\node [vb]	(4)	at (1,1)	{};
\path
\foreach \i in {1,2,3,4}{
    (\i) edge [e] (v)
    }
(1) edge [ev]  node [c, label=left:$c$] {} (2)
\foreach \i/\j in {2/3,3/4,4/1}{
    (\i) edge [ev] (\j)
    };
}
\end{tikzpicture}
}

\newcommand{\rtvfg}{
\begin{tikzpicture}[x=.7ex,y=.7ex,baseline={([yshift=-.59ex]current bounding box.center)}]
\node [vb]	(1)	at (-1,1)	{};
\node [vb]	(2)	at (-1,-1)	{};
\node [vb]	(3)	at (1,-1)	{};
\node [vb]	(4)	at (1,1)	{};
\path
\foreach \i/\j in {1/2,2/3,3/4,4/1}{
    (\i) edge [ev] (\j)
    };
\end{tikzpicture}
}

\newcommand{\rttadpole}{
\begin{tikzpicture}[x=2ex,y=2ex,baseline={([yshift=-1.2ex]current bounding box.center)}]
\node [v]	(v)	at (0,0) 	{};
\node [vb]	(1)	at (-1,1)	{};
\node [vb]	(2)	at (-1,-1)	{};
\node [vb]	(3)	at (1,-1)	{};
\node [vb]	(4)	at (1,1)	{};
\path
\foreach \i in {1,2,3,4}{
    (\i) edge [e] (v)
    }
\foreach \i/\j in {1/2,2/3,3/4,4/1}{
    (\i) edge [ev] (\j)
    }
(1)	edge [es,bend left=60]	(4);
\end{tikzpicture}
}

\newcommand{\rttadpolec}{
\begin{tikzpicture}[x=2ex,y=2ex,baseline={([yshift=-1.2ex]current bounding box.center)}]
\scriptsize{
\node [v]	(v)	at (0,0) 	{};
\node [vb]	(1)	at (-1,1)	{};
\node [vb]	(2)	at (-1,-1)	{};
\node [vb]	(3)	at (1,-1)	{};
\node [vb]	(4)	at (1,1)	{};
\path
\foreach \i in {1,2,3,4}{
    (\i) edge [e] (v)
    }
(1) edge [ev]  node [c, label=left:$c$] {} (2)
\foreach \i/\j in {2/3,3/4,4/1}{
    (\i) edge [ev] (\j)
    }
(1)	edge [es,bend left=60]	(4);
}
\end{tikzpicture}
}

\newcommand{\rttadpolepq}{
\begin{tikzpicture}[x=2ex,y=2ex,baseline={([yshift=-1.2ex]current bounding box.center)}]
\scriptsize{
\node [v]	(v)	at (0,0) 	{};
\node [vb]	(1)	at (-1,1)	{};
\node [vb]	(2)	at (-1,-1)	{};
\node [vb]	(3)	at (1,-1)	{};
\node [vb]	(4)	at (1,1)	{};
\path
\foreach \i in {1,2,3,4}{
    (\i) edge [e] (v)
    }
(1) edge [ev]  node [c, label=left:$p_1$] {} (2)
(1) edge [ev]  node [c, label=above:$q_2$]{} (4)
\foreach \i/\j in {2/3,3/4,4/1}{
    (\i) edge [ev] (\j)
    };
%(1)	edge [es,bend left=60]	(4);
\draw [es] (1) arc [start angle=180, end angle=0, radius=1];
}
\end{tikzpicture}
}

\newcommand{\rttadpoleb}{
\begin{tikzpicture}[x=2ex,y=2ex,baseline={([yshift=-.59ex]current bounding box.center)}]
\scriptsize{
\node [v]	(v)	at (0,0) 	{};
\node [vb]	(1)	at (-1,1)	{};
\node [vb]	(2)	at (-1,-1)	{};
\node [vb]	(3)	at (1,-1)	{};
\node [vb]	(4)	at (1,1)	{};
\path
\foreach \i in {1,2,3,4}{
    (\i) edge [e] (v)
    }
(1) edge [ev]  node [c, label=above:$c$] {} (4)
\foreach \i/\j in {1/2,2/3,3/4}{
    (\i) edge [ev] (\j)
    }
(2)	edge [es,bend right=60]	(3);
}
\end{tikzpicture}
}

\newcommand{\rttadpolebpq}{
\begin{tikzpicture}[x=2ex,y=2ex,baseline={([yshift=.49ex]current bounding box.center)}]
\scriptsize{
\node [v]	(v)	at (0,0) 	{};
\node [vb]	(1)	at (-1,1)	{};
\node [vb]	(2)	at (-1,-1)	{};
\node [vb]	(3)	at (1,-1)	{};
\node [vb]	(4)	at (1,1)	{};
\path
\foreach \i in {1,2,3,4}{
    (\i) edge [e] (v)
    }
(1) edge [ev]  node [c, label=left:$p_2$] {} (2)
(2) edge [ev]  node [c, label=below:$q_1$] {} (3)
\foreach \i/\j in {3/4,4/1}{
    (\i) edge [ev] (\j)
    };
%(2)	edge [es,bend right=60]	(3);
\draw [es] (2) arc [start angle=180, end angle=360, radius=1];
}
\end{tikzpicture}
}

\newcommand{\rtsunrisec}{
\begin{tikzpicture}[x=2ex,y=2ex,baseline={([yshift=-.6ex]current bounding box.center)}]
\scriptsize{
\begin{scope}[rotate=-45]
\node [v]	(v1)at (0,0) 	{};
\node [vb]	(1)	at (-1,1)	{};
\node [vb]	(2)	at (-1,-1)	{};
\node [vb]	(3)	at (1,-1)	{};
\node [vb]	(4)	at (1,1)	{};
\end{scope}
\begin{scope}[xshift=1.5cm,rotate=-45]
\node [v]	(v2)at (0,0) 	{};
\node [vb]	(5)	at (-1,1)	{};
\node [vb]	(6)	at (-1,-1)	{};
\node [vb]	(7)	at (1,-1)	{};
\node [vb]	(8)	at (1,1)	{};
\end{scope}
\path
\foreach \i in {1,2,3,4}{
    (\i) edge [e] (v1)
    }
\foreach \i in {5,6,7,8}{
    (\i) edge [e] (v2)
    }
(1) edge [ev]  node [c, label=135:$c$] {} (2)
(8) edge [ev]  node [c, label=45:$c$] {} (5)
\foreach \i/\j in {2/3,3/4,4/1,5/6,6/7,7/8}{
    (\i) edge [ev] (\j)
    }
\foreach \i/\j in {1/5,3/7,4/6}{
    (\i) edge [es] (\j)
    };
}
\end{tikzpicture}
}

\newcommand{\rtsunrisepq}{
\begin{tikzpicture}[x=2ex,y=2ex,baseline={([yshift=-.6ex]current bounding box.center)}]
\scriptsize{
\begin{scope}[rotate=-45]
\node [v]	(v1)at (0,0) 	{};
\node [vb]	(1)	at (-1,1)	{};
\node [vb]	(2)	at (-1,-1)	{};
\node [vb]	(3)	at (1,-1)	{};
\node [vb]	(4)	at (1,1)	{};
\end{scope}
\begin{scope}[xshift=1.5cm,rotate=-45]
\node [v]	(v2)at (0,0) 	{};
\node [vb]	(5)	at (-1,1)	{};
\node [vb]	(6)	at (-1,-1)	{};
\node [vb]	(7)	at (1,-1)	{};
\node [vb]	(8)	at (1,1)	{};
\end{scope}
\path
\foreach \i in {1,2,3,4}{
    (\i) edge [e] (v1)
    }
\foreach \i in {5,6,7,8}{
    (\i) edge [e] (v2)
    }
(1) edge [ev]  node [c, label=135:$p_1$] {} (2)
(2) edge [ev]  node [c, label=225:$p_2$] {} (3)
(3) edge [ev]  node [c, label=right:$q_1$] {} (4)
(4) edge [ev]  node [c, label=right:$q_2$] {} (1)
\foreach \i/\j in {5/6,6/7,7/8,8/5}{
    (\i) edge [ev] (\j)
    }
\foreach \i/\j in {1/5,3/7,4/6}{
    (\i) edge [es] (\j)
    };
}
\end{tikzpicture}
}

\newcommand{\rtsunrisesub}[1]{
\begin{tikzpicture}[x=2ex,y=2ex,baseline={([yshift=-.6ex]current bounding box.center)}]
\scriptsize{
\begin{scope}[rotate=-45]
\node [v]	(v1)at (0,0) 	{};
\node [vb]	(1)	at (-1,1)	{};
\node [vb]	(2)	at (-1,-1)	{};
\node [vb]	(3)	at (1,-1)	{};
\node [vb]	(4)	at (1,1)	{};
\end{scope}
\begin{scope}[xshift=1.5cm,rotate=-45]
\node [v]	(v2)at (0,0) 	{};
\node [vb]	(5)	at (-1,1)	{};
\node [vb]	(6)	at (-1,-1)	{};
\node [vb]	(7)	at (1,-1)	{};
\node [vb]	(8)	at (1,1)	{};
\end{scope}
\path
\foreach \i in {1,2,3,4}{
    (\i) edge [e] (v1)
    }
\foreach \i in {5,6,7,8}{
    (\i) edge [e] (v2)
    }
(1) edge [ev]  node [c, label=135:$c$] {} (2)
(8) edge [ev]  node [c, label=45:$c$] {} (5)
\foreach \i/\j in {2/3,3/4,4/1,5/6,6/7,7/8}{
    (\i) edge [ev] (\j)
    }
\foreach \i/\j in {#1}{
    (\i) edge [es] (\j)
    };
}
\end{tikzpicture}
}

\newcommand{\rtfisha}{
\begin{tikzpicture}[x=2ex,y=2ex,baseline={([yshift=-1.2ex]current bounding box.center)}]
\scriptsize{
\begin{scope}[rotate=-90]
\node [v]	(v1)at (0,0) 	{};
\node [vb]	(1)	at (-1,1)	{};
\node [vb]	(2)	at (-1,-1)	{};
\node [vb]	(3)	at (1,-1)	{};
\node [vb]	(4)	at (1,1)	{};
\end{scope}
\begin{scope}[xshift=1.3cm]
\node [v]	(v2)at (0,0) 	{};
\node [vb]	(5)	at (-1,1)	{};
\node [vb]	(6)	at (-1,-1)	{};
\node [vb]	(7)	at (1,-1)	{};
\node [vb]	(8)	at (1,1)	{};
\end{scope}
\path
\foreach \i in {1,2,3,4}{
    (\i) edge [e] (v1)
    }
\foreach \i in {5,6,7,8}{
    (\i) edge [e] (v2)
    }
(1) edge [ev]  node [c, label=above:$c$] {} (2)
(8) edge [ev]  node [c, label=above:$c$] {} (5)
\foreach \i/\j in {2/3,3/4,4/1,5/6,6/7,7/8}{
    (\i) edge [ev] (\j)
    }
\foreach \i/\j in {1/5,4/6}{
    (\i) edge [es] (\j)
    };
}
\end{tikzpicture}
}

\newcommand{\rtfishpq}[3]{
\begin{tikzpicture}[x=2ex,y=2ex,baseline={([yshift=-.59ex]current bounding box.center)}]
\scriptsize{
\begin{scope}[rotate=-90]
\node [v]	(v1)at (0,0) 	{};
\node [vb]	(1)	at (-1,1)	{};
\node [vb]	(2)	at (-1,-1)	{};
\node [vb]	(3)	at (1,-1)	{};
\node [vb]	(4)	at (1,1)	{};
\end{scope}
\begin{scope}[xshift=1.3cm]
\node [v]	(v2)at (0,0) 	{};
\node [vb]	(5)	at (-1,1)	{};
\node [vb]	(6)	at (-1,-1)	{};
\node [vb]	(7)	at (1,-1)	{};
\node [vb]	(8)	at (1,1)	{};
\end{scope}
\path
\foreach \i in {1,2,3,4}{
    (\i) edge [e] (v1)
    }
\foreach \i in {5,6,7,8}{
    (\i) edge [e] (v2)
    }
(1) edge [ev]  node [c, label=above:$#1$] {} (2)
(3) edge [ev]  node [c, label=below:$#3$] {} (4)
(1) edge [ev]  node [c, label=right:$#2$] {} (4)
\foreach \i/\j in {2/3,3/4,4/1,5/6,6/7,7/8,8/5}{
    (\i) edge [ev] (\j)
    }
\foreach \i/\j in {1/5,4/6}{
    (\i) edge [es] (\j)
    };
}
\end{tikzpicture}
}

\newcommand{\rtfishb}{
\begin{tikzpicture}[x=2ex,y=2ex,baseline={([yshift=-.59ex]current bounding box.center)}]
\scriptsize{
\node [v]	(v1)at (0,0) 	{};
\node [vb]	(1)	at (-1,1)	{};
\node [vb]	(2)	at (-1,-1)	{};
\node [vb]	(3)	at (1,-1)	{};
\node [vb]	(4)	at (1,1)	{};
\begin{scope}[xshift=1.3cm]
\node [v]	(v2)at (0,0) 	{};
\node [vb]	(5)	at (-1,1)	{};
\node [vb]	(6)	at (-1,-1)	{};
\node [vb]	(7)	at (1,-1)	{};
\node [vb]	(8)	at (1,1)	{};
\end{scope}
\path
\foreach \i in {1,2,3,4}{
    (\i) edge [e] (v1)
    }
\foreach \i in {5,6,7,8}{
    (\i) edge [e] (v2)
    }
(1) edge [ev]  node [c, label=left:$c$] {} (2)
(5) edge [ev]  node [c, label=left:$c$] {} (6)
\foreach \i/\j in {2/3,3/4,4/1,6/7,7/8,8/5}{
    (\i) edge [ev] (\j)
    }
\foreach \i/\j in {4/5,3/6}{
    (\i) edge [es] (\j)
    };
}
\end{tikzpicture}
}

\newcommand{\mapa}{
\begin{tikzpicture}[x=2ex,y=2ex,baseline={([yshift=-.59ex]current bounding box.center)}] 
\node [v]	(v)	at (0,0) 	{};
\node [vb]	(1)	at (-1,1)	{};
\node [vb]	(2)	at (-1,-1)	{};
\node [vb]	(3)	at (1,-1)	{};
\node [vb]	(4)	at (1,1)	{};
\path
\foreach \i in {1,2,3,4}{
    (\i) edge [e] (v)
    }
\foreach \i/\j in {1/2,3/4}{
   (\i)	edge [ev,bend left=30]	(\j)
   (\i)	edge [ev,bend right=30]	(\j)
   };
\end{tikzpicture}
}

\newcommand{\mapb}{
\begin{tikzpicture}[x=2ex,y=2ex,baseline={([yshift=-.59ex]current bounding box.center)}] 
\node [v]	(v1)at (0,0) 	{};
\node [v]	(v2)at (2.6,0) 	{};
\node [vb]	(1)	at (-1,0)	{};
\node [vb]	(2)	at (.6,-1)	{};
\node [vb]	(3)	at (.6,1)	{};
\node [vb]	(9)	at (3.6,0)	{};
\node [vb]	(8)	at (2,1)	{};
\node [vb]	(7)	at (2,-1)	{};
\path
\foreach \i in {1,2,3}{
    (\i) edge [e] (v1)
    }
\foreach \i in {7,8,9}{
    (\i) edge [e] (v2)
    }
\foreach \i/\j in {1/2,2/3,3/1,7/8,8/9,9/7}{
    (\i) edge [ev] (\j)
    }
\foreach \i/\j in {2/7,3/8}{
    (\i) edge [es] (\j)
    };
\end{tikzpicture}
}

\newcommand{\mapc}{
\begin{tikzpicture}[x=2ex,y=2ex,baseline={([yshift=-.59ex]current bounding box.center)}]
\node [v]	(v1)	at (0,0) 	{};
\node [vb]	(1)	at (-1,1)	{};
\node [vb]	(2)	at (-1,-1)	{};
\node [vb]	(3)	at (1,-1)	{};
\node [vb]	(4)	at (1,1)	{};
\begin{scope}[xshift=7ex]
\node [v]	(v2)	at (0,0) 	{};
\node [vb]	(5)	at (-1,1)	{};
\node [vb]	(6)	at (-1,-1)	{};
\node [vb]	(7)	at (1,-1)	{};
\node [vb]	(8)	at (1,1)	{};
\end{scope}
\path
\foreach \i in {1,2,3,4}{
    (\i) edge [e] (v1)
    }
\foreach \i in {5,6,7,8}{
    (\i) edge [e] (v2)
    }
\foreach \i/\j in {1/2,2/3,3/4,4/1,5/6,6/7,7/8,8/5}{
    (\i) edge [ev] (\j)
    }
\foreach \i/\j in {4/8,7/3}{
    (\i) edge [es, bend left=30] (\j)
    };
\end{tikzpicture}
}

\newcommand{\mapd}{
\begin{tikzpicture}[x=2ex,y=2ex,baseline={([yshift=-.59ex]current bounding box.center)}]
\node [v]	(v1)at (0,0) 	{};
\node [v]	(v2)at (2.4,0) 	{};
\node [vb]	(1)	at (-1,0)	{};
\node [vb]	(2)	at (.6,-1)	{};
\node [vb]	(3)	at (.6,1)	{};
\node [vb]	(4)	at (1.4,0)	{};
\node [vb]	(5)	at (3,-1)	{};
\node [vb]	(6)	at (3,1)	{};
\path
\foreach \i in {1,2,3}{
    (\i) edge [e] (v1)
    }
\foreach \i in {4,5,6}{
    (\i) edge [e] (v2)
    }
\foreach \i/\j in {1/2,2/3,3/1,4/5,5/6,6/4}{
    (\i) edge [ev] (\j)
    }
\foreach \i/\j in {2/5,3/6}{
    (\i) edge [es] (\j)
    };
\end{tikzpicture}
}

\newcommand{\maptorus}{
\begin{tikzpicture}[x=2ex,y=2ex,baseline={([yshift=-.59ex]current bounding box.center)}]
\node [v]	(v3)	at (0,0) 	{};
\node [v]	(v4)	at (2,0) 	{};
\node [vb]	(11)	at (-1,0)	{};
\node [vb]	(12)	at (.6,-1)	{};
\node [vb]	(13)	at (.6,1)	{};
\node [vb]	(9)	at (3,0)	{};
\node [vb]	(8)	at (1.4,1)	{};
\node [vb]	(7)	at (1.4,-1)	{};
\begin{scope}[yshift=5ex]
\node [v]	(v1)	at (0,0) 	{};
\node [v]	(v2)	at (2.4,0) 	{};
\node [vb]	(1)	at (-1,0)	{};
\node [vb]	(2)	at (.6,-1)	{};
\node [vb]	(3)	at (.6,1)	{};
\node [vb]	(4)	at (1.4,0)	{};
\node [vb]	(5)	at (3,-1)	{};
\node [vb]	(6)	at (3,1)	{};
\end{scope}
\path
\foreach \i/\j in {1/v1,2/v1,3/v1,4/v2,5/v2,6/v2,11/v3,12/v3,13/v3,7/v4,8/v4,9/v4}{
    (\i) edge [e] (\j)
    }
\foreach \i/\j in {1/2,2/3,3/1,4/5,5/6,6/4,7/8,8/9,9/7,11/12,12/13,13/11}{
    (\i) edge [ev] (\j)
    }
\foreach \i/\j in {2/5,3/6,1/11,12/7,13/8}{
    (\i) edge [es] (\j)
    }
(4) edge [es,sloped] (9);
\end{tikzpicture}
}

\newcommand{\mapcon}{
\begin{tikzpicture}[x=2ex,y=2ex,baseline={([yshift=-.59ex]current bounding box.center)}]
\draw [ev] (0,1.8) circle (.5ex); 
\draw [ev] (2,1.8) circle (.5ex); 
\node [v]	(v3)at (0,0) 	{};
\node [v]	(v4)at (2,0) 	{};
\node [vb]	(1)	at (-1,0)	{};
\node [vb]	(2)	at (.6,-1)	{};
\node [vb]	(3)	at (.6,1)	{};
\node [vb]	(9)	at (3,0)	{};
\node [vb]	(8)	at (1.4,1)	{};
\node [vb]	(7)	at (1.4,-1)	{};
\node [v]	(v2)at (1,1.5) 	{};
\node [vb]	(4)	at (0,1.5)	{};
\node [vb]	(5)	at (2,1.5)	{};
\path
\foreach \i/\j in {4/v2,5/v2,1/v3,2/v3,3/v3,7/v4,8/v4,9/v4}{
    (\i) edge [e] (\j)
    }
\foreach \i/\j in {1/2,2/3,3/1,7/8,8/9,9/7}{
    (\i) edge [ev] (\j)
    }
\foreach \i/\j in {2/7,3/8,1/4,9/5}{
    (\i) edge [es] (\j)
    };
\end{tikzpicture}
}

%--------------- r=3 coloured graphs ------------------

\newcommand{\rhvfc}{
\begin{tikzpicture}[rotate=0, x=2ex,y=2ex,baseline={([yshift=-.59ex]current bounding box.center)}] 
\scriptsize{
\node [v]	(1)	at (0,0)	{};
\node [vb]	(b1)	at (-1,1)	{};
\node [vb]	(w1)	at (-1,-1)	{};
\node [vb]	(b2)	at (1,-1)	{};
\node [vb]	(w2)	at (1,1)	{};
\path
\foreach \i/\j in {1/2}{
   	(w\i) edge [ev]   (b\j)
	(b\i) edge [ev]  node [c, label=left:$$] {} (w\j)
	}
\foreach \i in {1,2}{
	(w\i)	edge [ev,bend left=30]	(b\i)
	(w\i)	edge [ev,bend right=30]	(b\i)
	}
\foreach \i in {w1,b1,w2,b2}{
   	(\i) edge [e] (1)
	};
}
\end{tikzpicture}
}

\newcommand{\rhvfg}[1]{
\begin{tikzpicture}[rotate=-90,x=.7ex,y=.7ex,baseline={([yshift=-.59ex]current bounding box.center)}] 
\scriptsize{
\node [vb]	(b1)	at (-1,1)	{};
\node [vb]	(w1)	at (-1,-1)	{};
\node [vb]	(b2)	at (1,-1)	{};
\node [vb]	(w2)	at (1,1)	{};
\path
\foreach \i/\j in {1/2}{
   	(w\i) edge [ev]   (b\j)
	(b\i) edge [ev]  node [c, label=right:$#1$] {} (w\j)
	}
\foreach \i in {1,2}{
	(w\i)	edge [ev,bend left=30]	(b\i)
	(w\i)	edge [ev,bend right=30]	(b\i)
	};
}
\end{tikzpicture}
}

\newcommand{\rhtadpolea}{
\begin{tikzpicture}[rotate=90, x=2ex,y=2ex,baseline={([yshift=-1.2ex]current bounding box.center)}] 
\scriptsize{
\node [v]	(1)	    at (0,0)	{};
\node [vb]	(b1)	at (-1,1)	{};
\node [vb]	(w1)	at (-1,-1)	{};
\node [vb]	(b2)	at (1,-1)	{};
\node [vb]	(w2)	at (1,1)	{};
\path
\foreach \i/\j in {1/2}{
   	(w\i) edge [ev]   (b\j)
	(b\i) edge [ev]  node [c, label=left:$c$] {} (w\j)
	}
\foreach \i in {1,2}{
	(w\i)	edge [ev,bend left=30]	(b\i)
	(w\i)	edge [ev,bend right=30]	(b\i)
	}
\foreach \i in {w1,b1,w2,b2}{
   	(\i) edge [e] (1)
	}
(w2)	edge [es,bend left=90]	(b2);
}
\end{tikzpicture}
}

\newcommand{\rhtadpoleapq}{
\begin{tikzpicture}[rotate=90, x=2ex,y=2ex,baseline={([yshift=-1.2ex]current bounding box.center)}] 
\scriptsize{
\node [v]	(1)	    at (0,0)	{};
\node [vb]	(b1)	at (-1,1)	{};
\node [vb]	(w1)	at (-1,-1)	{};
\node [vb]	(b2)	at (1,-1)	{};
\node [vb]	(w2)	at (1,1)	{};
\path
\foreach \i/\j in {1/2}{
   	(w\i) edge [ev]   (b\j)
	(b\i) edge [ev]  node [c, label=left:$p_1$] {} (w\j)
	}
\foreach \i in {1,2}{
	(w\i)	edge [ev,bend left=30]	(b\i)
	(w\i)	edge [ev,bend right=30]	(b\i)
	}
\foreach \i in {w1,b1,w2,b2}{
   	(\i) edge [e] (1)
	}
(w2)	edge [es,bend left=90]	(b2);
}
\end{tikzpicture}
}

\newcommand{\rhtadpoleb}{
\begin{tikzpicture}[x=2ex,y=2ex,baseline={([yshift=-.3ex]current bounding box.center)}] 
\scriptsize{
\node [v]	(1)	    at (0,0)	{};
\node [vb]	(b1)	at (-1,1)	{};
\node [vb]	(w1)	at (-1,-1)	{};
\node [vb]	(b2)	at (1,-1)	{};
\node [vb]	(w2)	at (1,1)	{};
\path
\foreach \i/\j in {1/2}{
   	(w\i) edge [ev]  node [c, label=below:$c$] {} (b\j)
	(b\i) edge [ev]   (w\j)
	}
\foreach \i in {1,2}{
	(w\i)	edge [ev,bend left=30]	(b\i)
	(w\i)	edge [ev,bend right=30]	(b\i)
	}
\foreach \i in {w1,b1,w2,b2}{
   	(\i) edge [e] (1)
	}
(b1)	edge [es,bend left=60]	(w2);
}
\end{tikzpicture}
}

\newcommand{\rhtadpolebpq}{
\begin{tikzpicture}[x=2ex,y=2ex,baseline={([yshift=-1.2ex]current bounding box.center)}] 
\scriptsize{
\node [v]	(1)	    at (0,0)	{};
\node [vb]	(b1)	at (-1,1)	{};
\node [vb]	(w1)	at (-1,-1)	{};
\node [vb]	(b2)	at (1,-1)	{};
\node [vb]	(w2)	at (1,1)	{};
\path
\foreach \i/\j in {1/2}{
   	(w\i) edge [ev]  (b\j)
	(b\i) edge [ev]  node [c, label=above:$q_1$] {} (w\j)
	}
(w1)	edge [ev,bend left=30] node [c, label=left:{$p_2,p_3$}] {}	(b1)
(w1)	edge [ev,bend right=30]	(b1)
\foreach \i in {2}{
	(w\i)	edge [ev,bend left=30]	(b\i)
	(w\i)	edge [ev,bend right=30]	(b\i)
	}
\foreach \i in {w1,b1,w2,b2}{
   	(\i) edge [e] (1)
	};
%(b1)	edge [es,bend left=60]	(w2);
\draw [es] (b1) arc [start angle=180, end angle=0, radius=1];
}
\end{tikzpicture}
}

\newcommand{\rhsunrise}[5]{
\begin{tikzpicture}[x=2ex,y=2ex,baseline={([yshift=-.6ex]current bounding box.center)}]
\scriptsize{
\begin{scope}[rotate=-45]
\node [v]	(v1)at (0,0) 	{};
\node [vb]	(1)	at (-1,1)	{};
\node [vb]	(2)	at (-1,-1)	{};
\node [vb]	(3)	at (1,-1)	{};
\node [vb]	(4)	at (1,1)	{};
\end{scope}
\begin{scope}[xshift=1.5cm,rotate=-45]
\node [v]	(v2)at (0,0) 	{};
\node [vb]	(5)	at (-1,1)	{};
\node [vb]	(6)	at (-1,-1)	{};
\node [vb]	(7)	at (1,-1)	{};
\node [vb]	(8)	at (1,1)	{};
\end{scope}
\path
\foreach \i in {1,2,3,4}{
    (\i) edge [e] (v1)
    }
\foreach \i in {5,6,7,8}{
    (\i) edge [e] (v2)
    }
(1) edge [ev]  node [c, label=135:$#1$] {} (2)
(2) edge [ev, bend left=30]  node [c, label=225:$#2$] {} (3)
(2) edge [ev, bend right=30] (3)
(4) edge [ev, bend right=30]  node [c, label=right:$#3$] {} (1)
(4) edge [ev, bend left=30]   (1)
(3) edge [ev]  node [c, label=right:$#4$] {} (4)
(5) edge [ev]  node [c, label=45:$#5$] {} (8)
\foreach \i/\j in {6/7,8/5}{
    (\i) edge [ev] (\j)
    }
\foreach \i\j in {5/6,7/8}{
	(\i)	edge [ev,bend left=30]	(\j)
	(\i)	edge [ev,bend right=30]	(\j)
	}
\foreach \i/\j in {1/5,3/7,4/6}{
    (\i) edge [es] (\j)
    };
}
\end{tikzpicture}
}

\newcommand{\rhfish}[4]{
\begin{tikzpicture}[x=2ex,y=2ex,baseline={([yshift=-.59ex]current bounding box.center)}]
\scriptsize{
\begin{scope}[rotate=-90]
\node [v]	(v1)at (0,0) 	{};
\node [vb]	(1)	at (-1,1)	{};
\node [vb]	(2)	at (-1,-1)	{};
\node [vb]	(3)	at (1,-1)	{};
\node [vb]	(4)	at (1,1)	{};
\end{scope}
\begin{scope}[xshift=1.2cm]
\node [v]	(v2)at (0,0) 	{};
\node [vb]	(5)	at (-1,1)	{};
\node [vb]	(6)	at (-1,-1)	{};
\node [vb]	(7)	at (1,-1)	{};
\node [vb]	(8)	at (1,1)	{};
\end{scope}
\path
\foreach \i in {1,2,3,4}{
    (\i) edge [e] (v1)
    }
\foreach \i in {5,6,7,8}{
    (\i) edge [e] (v2)
    }
(1) edge [ev]  node [c, label=above:$#1$] {} (2)
(1) edge [ev,bend left=30]  node [c, label=right:$#2$] {} (4)
(1) edge [ev,bend right=30] (4)
(3) edge [ev]  node [c, label=below:$#3$] {} (4)
(5) edge [ev]  node [c, label=above:$#4$] {} (8)
\foreach \i/\j in {3/4,6/7,8/5}{
    (\i) edge [ev] (\j)
    }
\foreach \i\j in {2/3,5/6,7/8}{
	(\i)	edge [ev,bend left=30]	(\j)
	(\i)	edge [ev,bend right=30]	(\j)
	}
\foreach \i/\j in {1/5,4/6}{
    (\i) edge [es] (\j)
    };
}
\end{tikzpicture}
}

%--------------- r=4 coloured graphs ------------------

\newcommand{\cvtvg}{
\begin{tikzpicture}[x=2ex,y=2ex,baseline={([yshift=-.59ex]current bounding box.center)}] 
\node [vb]	(1)	at (0,1)	{};
\node [vb]	(2)	at (0,-1)	{};
\path
\foreach \i/\j in {1/2,2/1}{
   (\i)	edge [ev,bend left=15]	(\j)
   (\i)	edge [ev,bend left=40]	(\j)
   };
\end{tikzpicture}
}

\newcommand{\cvt}{
\begin{tikzpicture}[x=2ex,y=2ex,baseline={([yshift=-.59ex]current bounding box.center)}] 
\node [v]	(v)	at (0,0)	{};
\node [vb]	(1)	at (1,0)	{};
\node [vb]	(2)	at (-1,0)	{};
\path
\foreach \i in {1,2}{
    (\i) edge [e] (v)
    }
\foreach \i/\j in {1/2,2/1}{
   (\i)	edge [ev,bend left=30]	(\j)
   (\i)	edge [ev,bend left=50]	(\j)
   };
\end{tikzpicture}
}

\newcommand{\cvfl}{
\begin{tikzpicture}[x=2ex,y=2ex,baseline={([yshift=-.59ex]current bounding box.center)}] 
\scriptsize{
\node [v]	(1)	at (0,0)	{};
\node [vb, label=above:1]	(b1)	at (-1,1)	{};
\node [vb, label=below:2]	(w1)	at (-1,-1)	{};
\node [vb, label=below:7]	(b2)	at (1,-1)	{};
\node [vb, label=above:8]	(w2)	at (1,1)	{};
\path
\foreach \i/\j in {1/2}{
   	(w\i) edge [ev]   (b\j)
	(b\i) edge [ev]  node [c, label=above:$c$] {} (w\j)
	}
\foreach \i in {1,2}{
	(w\i)   edge 	[ev]				(b\i)
	(w\i)	edge [ev,bend left=30]	(b\i)
	(w\i)	edge [ev,bend right=30]	(b\i)
	}
\foreach \i in {w1,b1,w2,b2}{
   	(\i) edge [e] (1)
	};
}
\end{tikzpicture}
}

\newcommand{\cvft}{
\begin{tikzpicture}[x=2ex,y=2ex,baseline={([yshift=-.59ex]current bounding box.center)}] 
\node [v]	(v)	at (0,0) 	{};
\node [vb]	(1)	at (-1,1)	{};
\node [vb]	(2)	at (-1,-1)	{};
\node [vb]	(3)	at (1,-1)	{};
\node [vb]	(4)	at (1,1)	{};
\path
\foreach \i in {1,2,3,4}{
    (\i) edge [e] (v)
    }
\foreach \i/\j in {1/2,2/1,3/4,4/3}{
   (\i)	edge [ev,bend left=15]	(\j)
   (\i)	edge [ev,bend left=40]	(\j)
   };
\end{tikzpicture}
}

\newcommand{\cvftl}{
\begin{tikzpicture}[x=2ex,y=2ex,baseline={([yshift=-.59ex]current bounding box.center)}] 
\scriptsize{
\node [v]	(v)	at (0,0) 	{};
\node [vb, label=above:1]	(1)	at (-1,1)	{};
\node [vb, label=below:2]	(2)	at (-1,-1)	{};
\node [vb, label=below:7]	(3)	at (1,-1)	{};
\node [vb, label=above:8]	(4)	at (1,1)	{};
\path
\foreach \i in {1,2,3,4}{
    (\i) edge [e] (v)
    }
\foreach \i/\j in {1/2,2/1,3/4,4/3}{
   (\i)	edge [ev,bend left=15]	(\j)
   (\i)	edge [ev,bend left=40]	(\j)
   };
}
\end{tikzpicture}
}

\newcommand{\cfishl}{
\begin{tikzpicture}[x=2ex,y=2ex,baseline={([yshift=-.59ex]current bounding box.center)}] 
\scriptsize{
\node [v]	(1)	at (0,0)	{};
\node [v]	(2)	at (4,0)	{};
\node [vb, label=above:1]	(b1)	at (-1,1)	{};
\node [vb, label=below:2]	(w1)	at (-1,-1)	{};
\node [vb, label=below:3]	(b2)	at (1,-1)	{};
\node [vb, label=above:4]	(w2)	at (1,1)	{};
\node [vb, label=above:5]	(b3)	at (3,1)	{};
\node [vb, label=below:6]	(w3)	at (3,-1)	{};
\node [vb, label=below:7]	(b4)	at (5,-1)	{};
\node [vb, label=above:8]	(w4)	at (5,1)	{};
\path
\foreach \i/\j in {1/2}{
   	(w\i) edge [ev]   (b\j)
	(b\i) edge [ev]  node [c, label=above:$c_1$] {} (w\j)
	}
\foreach \i/\j in {3/4}{
   	(w\i) edge [ev]   (b\j)
	(b\i) edge [ev]  node [c, label=above:$c_2$] {} (w\j)
	}
\foreach \i in {1,2,3,4}{
	(w\i) edge 	[ev]				(b\i)
	(w\i)	edge [ev,bend left=30]	(b\i)
	(w\i)	edge [ev,bend right=30]	(b\i)
	}
\foreach \i in {w1,b1,w2,b2}{
   	(\i) edge [e] (1)
	}
\foreach \i in {w3,b3,w4,b4}{
   	(\i) edge [e] (2)
	}	
\foreach \i/\j in {w2/b3,w3/b2}{
   	(\i) edge [es] (\j)
	};
}
\end{tikzpicture}
}

\newcommand{\cfishlc}{
\begin{tikzpicture}[x=2ex,y=2ex,baseline={([yshift=-.59ex]current bounding box.center)}] 
\scriptsize{
\node [v]	(1)	at (0,0)	{};
\node [v]	(2)	at (4,0)	{};
\node [vb, label=above:1]	(b1)	at (-1,1)	{};
\node [vb, label=below:2]	(w1)	at (-1,-1)	{};
\node [vb, label=below:3]	(b2)	at (1,-1)	{};
\node [vb, label=above:4]	(w2)	at (1,1)	{};
\node [vb, label=above:5]	(b3)	at (3,1)	{};
\node [vb, label=below:6]	(w3)	at (3,-1)	{};
\node [vb, label=below:7]	(b4)	at (5,-1)	{};
\node [vb, label=above:8]	(w4)	at (5,1)	{};
\path
\foreach \i/\j in {1/2,3/4}{
   	(w\i) edge [ev]   (b\j)
	(b\i) edge [ev]  node [c, label=above:$c$] {} (w\j)
    }
\foreach \i in {1,2,3,4}{
	(w\i) edge 	[ev]				(b\i)
	(w\i)	edge [ev,bend left=30]	(b\i)
	(w\i)	edge [ev,bend right=30]	(b\i)
	}
\foreach \i in {w1,b1,w2,b2}{
   	(\i) edge [e] (1)
	}
\foreach \i in {w3,b3,w4,b4}{
   	(\i) edge [e] (2)
	}	
\foreach \i/\j in {w2/b3,w3/b2}{
   	(\i) edge [es] (\j)
	};
}
\end{tikzpicture}
}

\newcommand{\cfish}{
\begin{tikzpicture}[x=2ex,y=2ex,baseline={([yshift=-.59ex]current bounding box.center)}] 
\scriptsize{
\node [v]	(1)	at (0,0)	{};
\node [v]	(2)	at (4,0)	{};
\node [vb]	(w1)	at (-1,-1)	{};
\node [vb]	(b1)	at (-1,1)	{};
\node [vb]	(w2)	at (1,1)	{};
\node [vb]	(b2)	at (1,-1)	{};
\node [vb]	(w3)	at (3,-1)	{};
\node [vb]	(b3)	at (3,1)	{};
\node [vb]	(w4)	at (5,1)	{};
\node [vb]	(b4)	at (5,-1)	{};
\path
\foreach \i/\j in {1/2,3/4}{
   	(w\i) edge [ev]   (b\j)
	(b\i) edge [ev]  node [c] {} (w\j)
	}
\foreach \i in {1,2,3,4}{
	(w\i) edge 	[ev]				(b\i)
	(w\i)	edge [ev,bend left=30]	(b\i)
	(w\i)	edge [ev,bend right=30]	(b\i)
	}
\foreach \i in {w1,b1,w2,b2}{
   	(\i) edge [e] (1)
	}
\foreach \i in {w3,b3,w4,b4}{
   	(\i) edge [e] (2)
	}	
\foreach \i/\j in {w2/b3,w3/b2}{
   	(\i) edge [es] (\j)
	};
}
\end{tikzpicture}
}

\newcommand{\cfisha}{
\begin{tikzpicture}[x=2ex,y=2ex,baseline={([yshift=-.59ex]current bounding box.center)}] 
\scriptsize{
\node [v]	(1)	at (0,0)	{};
\node [v]	(2)	at (4,0)	{};
\node [vb]	(w1)	at (-1,-1)	{};
\node [vb]	(b1)	at (-1,1)	{};
\node [vb]	(w2)	at (1,1)	{};
\node [vb]	(b2)	at (1,-1)	{};
\node [vb]	(w3)	at (3,-1)	{};
\node [vb]	(b3)	at (3,1)	{};
\node [vb]	(w4)	at (5,1)	{};
\node [vb]	(b4)	at (5,-1)	{};
\path
\foreach \i/\j in {1/2,3/4}{
   	(w\i) edge [ev]   (b\j)
	(b\i) edge [ev]  node [c] {} (w\j)
	}
\foreach \i in {1,2,3,4}{
	(w\i) edge 	[ev]				(b\i)
	(w\i)	edge [ev,bend left=30]	(b\i)
	(w\i)	edge [ev,bend right=30]	(b\i)
	}
\foreach \i in {w1,b1,w2,b2}{
   	(\i) edge [e] (1)
	}
\foreach \i in {w3,b3,w4,b4}{
   	(\i) edge [e] (2)
	}	
\foreach \i/\j in {w3/b2}{
   	(\i) edge [es] (\j)
	};
}
\end{tikzpicture}
}

\newcommand{\cfishb}{
\begin{tikzpicture}[x=2ex,y=2ex,baseline={([yshift=-.59ex]current bounding box.center)}] 
\scriptsize{
\node [v]	(1)	at (0,0)	{};
\node [v]	(2)	at (4,0)	{};
\node [vb]	(w1)	at (-1,-1)	{};
\node [vb]	(b1)	at (-1,1)	{};
\node [vb]	(w2)	at (1,1)	{};
\node [vb]	(b2)	at (1,-1)	{};
\node [vb]	(w3)	at (3,-1)	{};
\node [vb]	(b3)	at (3,1)	{};
\node [vb]	(w4)	at (5,1)	{};
\node [vb]	(b4)	at (5,-1)	{};
\path
\foreach \i/\j in {1/2,3/4}{
   	(w\i) edge [ev]   (b\j)
	(b\i) edge [ev]  node [c] {} (w\j)
	}
\foreach \i in {1,2,3,4}{
	(w\i) edge 	[ev]				(b\i)
	(w\i)	edge [ev,bend left=30]	(b\i)
	(w\i)	edge [ev,bend right=30]	(b\i)
	}
\foreach \i in {w1,b1,w2,b2}{
   	(\i) edge [e] (1)
	}
\foreach \i in {w3,b3,w4,b4}{
   	(\i) edge [e] (2)
	}	
\foreach \i/\j in {w2/b3}{
   	(\i) edge [es] (\j)
	};
}
\end{tikzpicture}
}

\newcommand{\cfishc}{
\begin{tikzpicture}[x=2ex,y=2ex,baseline={([yshift=-.59ex]current bounding box.center)}] 
\scriptsize{
\node [v]	(1)	at (0,0)	{};
\node [v]	(2)	at (4,0)	{};
\node [vb]	(w1)	at (-1,-1)	{};
\node [vb]	(b1)	at (-1,1)	{};
\node [vb]	(w2)	at (1,1)	{};
\node [vb]	(b2)	at (1,-1)	{};
\node [vb]	(w3)	at (3,-1)	{};
\node [vb]	(b3)	at (3,1)	{};
\node [vb]	(w4)	at (5,1)	{};
\node [vb]	(b4)	at (5,-1)	{};
\path
\foreach \i/\j in {1/2,3/4}{
   	(w\i) edge [ev]   (b\j)
	(b\i) edge [ev]  node [c] {} (w\j)
	}
\foreach \i in {1,2,3,4}{
	(w\i) edge 	[ev]				(b\i)
	(w\i)	edge [ev,bend left=30]	(b\i)
	(w\i)	edge [ev,bend right=30]	(b\i)
	}
\foreach \i in {w1,b1,w2,b2}{
   	(\i) edge [e] (1)
	}
\foreach \i in {w3,b3,w4,b4}{
   	(\i) edge [e] (2)
	};
}
\end{tikzpicture}
}

\newcommand{\cvsla}{
\begin{tikzpicture}[x=2ex,y=2ex,baseline={([yshift=0ex]current bounding box.center)}]
\scriptsize{
\node [v]	(1)	at (0,0)	{};
\begin{scope}[rotate=120]
\node [vb, label=left:1]		(b1)	at (.7,1.2)	{};
\node [vb, label=below:2]		(w1)	at (-.7,1.2)	{};
\end{scope}
\begin{scope}[rotate=240]
\node [vb, label=below:7]		(b2)	at (.7,1.2)	{};
\node [vb, label=right:8]		(w2)	at (-.7,1.2)	{};
\end{scope}
\begin{scope}[rotate=0]
\node [vb, label=right:5]		(b3)	at (.7,1.2)	{};
\node [vb, label=left:4]		(w3)	at (-.7,1.2)	{};
\end{scope}	
\path
(w1) edge [ev]  node [c, label=below:$c$] {} (b2)
\foreach \i/\j in {2/3,3/1}{
   	(w\i) edge [ev]   (b\j)
	}
\foreach \i in {1,2,3}{
	(w\i) edge 	[ev]				(b\i)
	(w\i)	edge [ev,bend left=30]	(b\i)
	(w\i)	edge [ev,bend right=30]	(b\i)
	}
\foreach \i in {w1,b1,w2,b2,w3,b3}{
   	(\i) edge [e] (1)
	}
(w3) edge [es, bend left=90] (b3);
}
\end{tikzpicture}
}

\newcommand{\cvslb}{
\begin{tikzpicture}[x=2ex,y=2ex,baseline={([yshift=-1.2ex]current bounding box.center)}]
\scriptsize{
\node [v]	(1)	at (0,0)	{};
\begin{scope}[rotate=60]
\node [vb, label=above:1]		(b1)	at (.7,1.2)	{};
\node [vb, label=left:2]		(w1)	at (-.7,1.2)	{};
\end{scope}
\begin{scope}[rotate=180]
\node [vb, label=left:3]		(b2)	at (.7,1.2)	{};
\node [vb, label=right:6]		(w2)	at (-.7,1.2)	{};
\end{scope}
\begin{scope}[rotate=300]
\node [vb, label=right:7]		(b3)	at (.7,1.2)	{};
\node [vb, label=above:8]		(w3)	at (-.7,1.2)	{};
\end{scope}	
\path
(b1) edge [ev]  node [c, label=above:$c$] {} (w3)
\foreach \i/\j in {1/2,2/3}{
   	(w\i) edge [ev]   (b\j)
	}
\foreach \i in {1,2,3}{
	(w\i) edge 	[ev]				node 	{}	(b\i)
	(w\i)	edge [ev,bend left=30]	node 	{}	(b\i)
	(w\i)	edge [ev,bend right=30]	node 	{}	(b\i)
	}
\foreach \i in {w1,b1,w2,b2,w3,b3}{
   	(\i) edge [e] (1)
	}
(w2) edge [es, bend left=90] (b2);
}
\end{tikzpicture}
}

\newcommand{\cvslc}{
\begin{tikzpicture}[x=2ex,y=2ex,baseline={([yshift=-.59ex]current bounding box.center)}]
\scriptsize{
\node [v]	(1)	at (0,0)	{};
\begin{scope}[rotate=90]
\node [vb, label=left:1]		(b1)	at (.7,1.2)	{};
\node [vb, label=left:2]		(w1)	at (-.7,1.2)	{};
\end{scope}
\begin{scope}[rotate=210]
\node [vb, label=below:5]		(b2)	at (.7,1.2)	{};
\node [vb, label=right:8]		(w3)	at (-.7,1.2)	{};
\end{scope}
\begin{scope}[rotate=330]
\node [vb, label=right:7]		(b3)	at (.7,1.2)	{};
\node [vb, label=above:4]		(w2)	at (-.7,1.2)	{};
\end{scope}	
\path
(b1) edge [ev]  node [c, label=150:$c_1$] {} (w2)
(w2) edge [ev]  node [c, label=30:$c_2$] {} (b3)
\foreach \i/\j in {1/1,1/2,3/2,3/3}{
   	(w\i) edge [ev]   (b\j)
	}
\foreach \i in {1,2,3}{
	(w\i)	edge [ev,bend left=30]	node 	{}	(b\i)
	(w\i)	edge [ev,bend right=30]	node 	{}	(b\i)
	}
\foreach \i in {w1,b1,w2,b2,w3,b3}{
   	(\i) edge [e] (1)
	}
(w2) edge [es, bend left=140] (b2);
}
\end{tikzpicture}
}

\newcommand{\cvsld}{
\begin{tikzpicture}[x=2ex,y=2ex,baseline={([yshift=-.59ex]current bounding box.center)}]
\scriptsize{
\node [v]	(1)	at (0,0)	{};
\begin{scope}[rotate=90]
\node [vb, label=left:1]		(b1)	at (.7,1.2)	{};
\node [vb, label=left:2]		(w1)	at (-.7,1.2)	{};
\end{scope}
\begin{scope}[rotate=210]
\node [vb, label=below:6]		(b2)	at (.7,1.2)	{};
\node [vb, label=right:8]		(w3)	at (-.7,1.2)	{};
\end{scope}
\begin{scope}[rotate=330]
\node [vb, label=right:7]		(b3)	at (.7,1.2)	{};
\node [vb, label=above:3]		(w2)	at (-.7,1.2)	{};
\end{scope}	
\path
(b1) edge [ev]  node [c, label=150:$c_1$] {} (w2)
(w2) edge [ev]  node [c, label=30:$c_2$] {} (b3)
\foreach \i/\j in {1/1,1/2,3/2,3/3}{
   	(w\i) edge [ev]   (b\j)
	}
\foreach \i in {1,2,3}{
	(w\i)	edge [ev,bend left=30]	node 	{}	(b\i)
	(w\i)	edge [ev,bend right=30]	node 	{}	(b\i)
	}
\foreach \i in {w1,b1,w2,b2,w3,b3}{
   	(\i) edge [e] (1)
	}
(w2) edge [es, bend left=140] (b2);
}
\end{tikzpicture}
}

\newcommand{\cvsd}{
\begin{tikzpicture}[x=2ex,y=2ex,baseline={([yshift=-.59ex]current bounding box.center)}]
\scriptsize{
\node [v]	(1)	at (0,0)	{};
\begin{scope}[rotate=90]
\node [vb]		(b1)	at (.7,1.2)	{};
\node [vb]		(w1)	at (-.7,1.2)	{};
\end{scope}
\begin{scope}[rotate=210]
\node [vb]		(b2)	at (.7,1.2)	{};
\node [vb]		(w3)	at (-.7,1.2)	{};
\end{scope}
\begin{scope}[rotate=330]
\node [vb]		(b3)	at (.7,1.2)	{};
\node [vb]		(w2)	at (-.7,1.2)	{};
\end{scope}	
\path
% (b1) edge [ev]  node [c, label=150:$c_1$] {} (w2)
% (w2) edge [ev]  node [c, label=30:$c_2$] {} (b3)
\foreach \i/\j in {1/1,1/2,2/1,2/3,3/2,3/3}{
   	(w\i) edge [ev]   (b\j)
	}
\foreach \i in {1,2,3}{
	(w\i)	edge [ev,bend left=30]	node 	{}	(b\i)
	(w\i)	edge [ev,bend right=30]	node 	{}	(b\i)
	}
\foreach \i in {w1,b1,w2,b2,w3,b3}{
   	(\i) edge [e] (1)
	}
(w2) edge [es, bend left=140] (b2);
}
\end{tikzpicture}
}

\title{Combinatorial Dyson-Schwinger Equations of Quartic Matrix Field Theory}
%{renormalization of the $\phi^4$ Matrix Field Theory Model in 4 Dimensions from analytic, algebraic and combinatorial perspective}

\author[A. Hock]{Alexander Hock \href{https://orcid.org/0000-0002-8404-4056}{\scaleto{\includegraphics{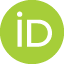}}{10pt}}}
\author[J. Thürigen]{Johannes Thürigen \href{https://orcid.org/0000-0002-6262-1430}{\scaleto{\includegraphics{ORCID-iD_icon-64x64.png}}{10pt}}}

\address{Mathematical Institute, University of Oxford,
  \newline
  Andrew Wiles Building, Woodstock Road, OX2 6GG, Oxford, United Kingdom
  \newline
  {\itshape e-mail:} \normalfont \texttt{alexander.hock@maths.ox.ac.uk}}

\address{Mathematical Institute\\ University of Münster
\newline
Einsteinstr.~62\\ D-48149 M\"unster\\ Germany
\newline
{\itshape e-mail:} \normalfont 
\texttt{johannes.thuerigen@uni-muenster.de }}

\begin{abstract}
  Matrix field theory is a combinatorially non-local field theory which
  has recently been found to be a non-trivial but solvable QFT example.
%  To this end functional methods have been applied to derive Ward identities and closed Dyson-Schwinger equations and the identification of a topological recursion allows for a systematic understanding of solutions.
  To generalize such non-perturbative structures to other models, a more combinatorial understanding of Dyson-Schwinger equations and their solutions is of high interest.
  To this end we consider combinatorial Dyson-Schwinger equations manifestly relying on the Hopf-algebraic structure of perturbative renormalization.
  We find that these equations are fully compatible with renormalization, relying only on the superficially divergent diagrams which are planar ribbon graphs, i.e.~decompleted dual combinatorial maps. 
  Still, they are of a similar kind as in realistic models of local QFT, featuring in particular an infinite number of primitive diagrams as well as graph-dependent combinatorial factors. 
\end{abstract}

\keywords{}

\maketitle

\markboth{\hfill\textsc\shortauthors}{%
\textsc{Combinatorial Dyson-Schwinger Equations of Quartic Matrix Field Theory}\hfill}

%\tableofcontents

\section{Introduction}

Quartic $\phi^4$ matrix field theory (MFT) has proven to be an instructive case of a field theory which is at the same time non-trivial and analytically solvable \cite{Grosse:2012uv,Grosse:2019jnv,Hock:2020rje}, 
related to integrable structures~\cite{Borot:2023thu,Grosse:2023jcb} 
and topological recursion~\cite{Hock:2021tbl,Branahl:2020yru,Branahl:2022uge}.
It provides a matrix representation of scalar $\phi^4$ Euclidean Quantum Field Theory (QFT) on noncommutative Moyal space at the self-dual point, 
also known as the model of Grosse and Wulkenhaar who have proven renormalizability in $D=4$ dimensions to all orders \cite{Grosse:2003aj,Grosse:2004yu,Rivasseau:2005bh}. 
It can be viewed as the quartic \emph{generalization} of Kontsevich's cubic matrix model \cite{Kontsevich:1992ti}.
On the other hand, in the framework of Tensorial and Group Field Theory  \cite{Oriti:1110, BenGeloun:1111, BenGeloun:1306, Carrozza:13},
or even more general combinatorially non-local field theory \cite{Oriti:1409, Thurigen:2102},
matrix field theory is the \emph{special case} of tensor fields of rank two.

Understanding of the analytic structure of MFT relies crucially on \emph{analytic Dyson-Schwinger equations} \cite{Grosse:2012uv}.
In the first place, these equations are a coupled system of infinitely many integro-differential equations derived from the formal path integral.
However, they decouple when applying Ward identities \cite{Disertori:2006nq} and a formal $1/N$ expansion.
It is then possible to solve the equation for the 2-point function \cite{Grosse:2019jnv,Grosse:2019qps}
and find all other planar correlation functions recursively \cite{deJong:2019oez}.
While extremely successful for quartic MFT, 
application of this method to tensorial field theories of higher rank has turned out to be difficult despite various attempts \cite{PerezSanchez:1608, Pascalie:1706,Pascalie:1810,Pascalie:1903}.
One reason might be the much more intricate combinatorial structure at higher rank \cite{Gurau:16}.

On the other hand, \emph{combinatorial Dyson-Schwinger equations} are a complementary way to understand correlation functions in a renormalizable QFT via recursive equations \cite{Bergbauer:2005fb, Kreimer:2005rw, Kreimer:2009iy, Foissy:1112, Yeats:17}.
They are based on the underlying Hopf algebra $\hfd$ of divergent Feynman diagrams which describes perturbative renormalization of such QFT \cite{Connes:1998qv, Connes:1999yr}.
Since it is entirely built on the combinatorics of Feynman diagrams which is well understood in tensorial theories \cite{Gurau:16, BenGeloun:1306}, 
combinatorial Dyson-Schwinger equations might be better suited to treat such theories.
In fact, Connes-Kreimer type Hopf algebras describe renormalization not only in point-like interacting QFT but much more general also in field theories with combinatorially non-local interactions \cite{Thurigen:2102} and can be used for explicit calculations of amplitudes \cite{Thurigen:2103} which match results in perturbative MFT \cite{Hock:2020rje}.
For the case of MFT such a Hopf algebra has been first constructed in \cite{Tanasa:0707}.
%discrepance to Hopf algebra to earlier Hopf algebra

MFT is the ideal testing ground for the applicability of combinatorial Dyson-Schwinger equations to combinatorially non-local field theories since, on the one hand, its perturbation theory is described by a Hopf algebra $\hfd$ and, on the other hand, we already know its non-perturbative solutions.
One can therefore address the question how the two types of Dyson-Schwinger equations (analytic and combinatorial) are related and whether combinatorial Dyson-Schwinger equations are an appropriate tool to find non-perturbative solutions for a non-trivial QFT,
as has been shown for various special models in QFT \cite{Broadhurst:2000dq, Bergbauer:2005fb, Kreimer:0612}.
In particular, one a priori expects Ward identities to play a crucial role like in the path-integral setup where they are relevant for decoupling the infinite tower of analytic Dyson-Schwinger equations. 
In the Hopf-algebraic picture, Ward identities define Hopf ideals $\mathcal{I}$ such that the relevant Hopf Algebra eventually is the quotient $\hfd/\mathcal{I}$ on which also the combinatorial Dyson-Schwinger equations usually simplify drastically \cite{Kreimer:2005rw, vanSuijlekom:0610, Kreimer:2009iy}.

\vspace{1ex}

Motivated by these questions, we arrive at the following results about the combinatorial Dyson-Schwinger equations for quartic MFT.
First of all, Thm.~\ref{theorem:Hopf algebra} states that there is a Hopf algebra $\hfd$ of 4-regular ribbon graphs $\rg$ of genus zero with a single boundary.
This is the Hopf algebra that captures the perturbative renormalization of quartic MFT \cite{Hock:2020rje, Thurigen:2103}.
It has already been sketched as an example in \cite{Thurigen:2102}, but for the current purpose it is necessary to work out the details:
\vtwo{The basis is a combinatorial definition of ribbon graphs via permutations such that the completion of a ribbon graph is dual to a combinatorial map (worked out in detail in App.~\ref{sec:ribbon graphs}.}
We give the definition of the contraction $\rg/\sg$ of a ribbon subgraph $\sg\sgr\rg$ directly at the level of permutations and \vtwo{argue} that $\rg/\sg$ 
%is not (the decompletion of) a combinatorial map 
\vtwo{gives the proper ribbon graph in the QFT sense only
when all components of $\sg$ have a single boundary;
for more than one boundary this leads to multi-trace vertices which cannot be} covered by combinatorial maps. 
Crucially, the Hopf algebra~$\hfd$ of MFT contains only connected ribbon graphs with a single boundary such that the coproduct $\cop:\rg\mapsto\sum_{\sg\sgr\rg} \sg \otimes \rg/\sg$ is closed in $\hfd$.
This aspect is not completely clear in the earlier construction of~\cite{Tanasa:0707} and seems to be the reason for explicitly including planar ribbon graphs with multiple boundaries (coined ``planar irregular'') in the Hopf algebra in \cite{Tanasa:2009hb}, %reprinted in 
\vtwo{also}~\cite{Tanasa:21}.

The second main result, Thm.~\ref{thm:DSE}, are the combinatorial Dyson-Schwinger equations for the series $\xe$ and $\xv$ over 2-point and 4-point ribbon graphs, 
\begin{align}
    \xe(\alpha) &= \edge %\mathbbm{1} 
    - \alpha \graft^{\edge}(Q\xe)
= \edge %\mathbbm{1} 
- \alpha (\graft^{\tadpoleup}+\graft^{\tadpoledown}) (Q\xe) \\
\xv(\alpha) &= \vertex %\mathbbm{1} 
+ \!\! \sum_{\substack{\Gamma \text{ primitive} \\ \res(\Gamma)=\vertex}} \!\!\!\! \alpha^{\nf_\Gamma} \graft^{\Gamma}(Q^{\nf_\Gamma} \xv)
=\vertex %\mathbbm{1} 
+ \alpha (\graft^{\fishright} + \graft^{\fishup})(Q\xe) + ...
\end{align}
wherein $Q = (\xe)^{-2} \xv$.
While these equation have the same form as in \cite{Tanasa:2009hb}, 
they differ in the definition of the grafting operator $\graft$.
Here we introduce a new definition
\footnote{
We thank Michael Borinsky for a discussion leading to this definition. % on May 27, 2021 
}, Def.~\ref{def:graftingOP},
\[
\graft^{\Gamma}(\sg) := \sum_{\cg\in [\Gamma]_{\sim_2}} \frac1{|\mathcal{I}(\sg,\cg)|} \sum_{\ib\in \mathcal{I}(\sg,\cg)} \frac{\cg\circ_\ib \sg}{\maxf(\cg\circ_\ib \sg)} \, ,
\]
which is based on a mathematically rigorous notion of insertion $\cg\circ_\ib \sg$ with respect to an isomorphism $\ib\in\mathcal{I}(\sg,\cg)$, Def.~\ref{def:insertion}, according to \cite{Borinsky:2018,Thurigen:2102}.
It is essential that these isomorphisms allow only insertions of ribbon graphs $\sg$ in $\cg$ for which the external structure of~$\sg$ matches the vertex structure of~$\cg$.
The expansion of $\graft$ in terms of ribbon graphs weighted by combinatorial factors common in the literature~\cite{Kreimer:2005rw} then follows as a consequence, Prop.~\ref{prop:grafting expansion}.

Two aspects in the combinatorial Dyson-Schwinger equation hint at the non-trivial structure of MFT.
On the one hand, it is necessary to include the factor $1/\maxf$ 
%\vtwo{(counting maximal forests)} 
in the definition of $\graft$ due to the presence of overlapping divergences. 
That is, there are ribbon graphs $\rg$ in the theory which have different, usually overlapping subgraphs $\sg_1,\sg_2$ whose contraction yields primitive ribbon graphs~$\rg/\sg_i$, i.e.~graphs without subdivergences. Their number is counted by $\maxf(\rg)$.
On the other hand, we find that there are infinitely many such primitives in the 4-point series $\xv$.
These facts are expected to relate to a specific anomalous dimension. 

For a precise statement, we use the known analytic solution of the 2-point function, Thm.~\ref{thm:hypergeom} to calculate the anomalous dimension $\gamma$ of quartic MFT 
and find in Cor.~\ref{cor:anomalous dimension} that 
\[
\gamma=-\frac{1}{\pi}\arctan (\lambda \pi) 
\]
corroborating the earlier results of a dimension drop in the spectral dimension of the non-commutative Grosse-Wulkenhaar model \cite{Grosse:2019qps}. 
This reduction of dimension below the otherwise critical dimension of quartic MFT is  why the theory avoids quantum triviality, i.e.~it shows that the theory is renormalizable as an interacting Euclidean QFT.

\vspace{1ex}

Mathematically, the combinatorial Dyson-Schwinger equations are supposed to follow from the relation of the Hopf algebra to Hochschild cohomology \cite{Connes:1998qv}.
In the related Hopf algebra of rooted trees, $\graft$ is a Hochschild 1-cocyle,
\[\label{eq:Hochschild1}
\cop\graft=\graft\otimes\bullet + (\id\otimes\graft)\cop \, ,
\]
from which recursive relations and Hopf subalgebras follow,
in particular the relation
\[\label{eq:coproduct of coefficients}
\Delta (c_n^{\bullet})=\sum_{k=0}^n P^\bullet_{n,k}\otimes c_{n-k}^{\bullet}
\] 
with polynomials $P^\bullet_{n,k}$ of order $k$ in the coefficients %$c_n^{\bullet}$ of loop order $\alpha^n$ in 
of the series $\xr= \sum_{j\ge0} \alpha^j c^\bullet_j$ for $\bullet=e,v$.
These results are expected to generalize to the Hopf algebra $\hfd$ of divergent Feynman diagrams of a perturbatively renormalizable QFT since $\hfd$ relates to a Hopf algebra of decorated rooted trees, or more precisely decorated posets \cite{Borinsky:2015mga}.
This has been worked out in specific cases \cite{Bergbauer:2005fb, Kreimer:2009iy} and it is exactly here where the above mentioned Ward identities might become relevant as the Hochschild property Eq.~\ref{eq:Hochschild1} can be true only on the quotient $\hfd/\mathcal{I}$ \cite{Kreimer:2005rw, vanSuijlekom:0610}. 
However, we are not aware of a complete proof for more realistic QFTs, in particular not for MFT. 

Here we revert the usual logic and start with the subalgebra structure to then determine the conditions for the Hochschild property to be true.
We find that the subalgebra structure is completely independent of Hochschild cocyles and recursive relations.
In particular, we prove the coproduct Eq.~\eqref{eq:coproduct of coefficients} in Thm.~\ref{thm:coproduct} providing also the explicit form of the polynomials $P^\bullet_{n,k}(c)$.
In fact, such a formula of the coproduct also holds more generally for any monomials in $1/\xe$ and $\xv$ as we prove in Prop.~\ref{prop:coproductmonomial}.
Together with the combinatorial Dyson-Schwinger equations it is then straightforward to prove the Hochschild property Eq.~\eqref{eq:Hochschild1} for a sum of~$\graft^\Gamma$ over all primitives $\Gamma$ of the theory, Thm.~\ref{Thm:Hochschild}.
However it is not clear whether this also holds for $\graft$ to any given order, that is, whether also
\[
\Delta \graft^{\bullet,n}=\graft^{\bullet,n}\otimes \bullet+(\id\otimes \graft^{\bullet,n})\Delta 
\]
is true at any loop order $n$ for 
$\graft^{\bullet,n} = \sum_\Gamma \graft^\Gamma$ summing over primitive $\bullet$-point ribbon graphs $\Gamma$ with $n$ faces.
What we can prove in Prop.~\ref{prop:Hochschildn} is that this is true if and only if for any single ribbon graph $\rg\in\hfd$ all primitive cographs $\cg=\rg/\sg$ have the same number of faces.
Though possibly true in quartic MFT, this is however a highly non-trivial statement to show and we leave its proof, or proof of a counter example, for the future.

It is important to note that Ward identities cannot improve this situation in quartic MFT.
Ward identities in the Hopf algebra are always of the type $Q_i \sim Q_j$  relating different types of monomials $Q_i$ in the perturbative series $\xr$ of a theory related to different types of vertices \cite{Kreimer:2009iy,Prinz:2001}. 
However, in quartic MFT we have one type of interactions, the quartic one, and thus only one type of divergent $n$-point functions next to the $2$-point functions, resulting in a single $Q=(\xe)^{-2} \xv$.
Thus, the known analytic Ward identities (Eq.~\eqref{DSE2W}, \eqref{W4P}) do not have a direct algebraic analogue.

\vspace{1ex}

Combinatorial Dyson-Schwinger equations in the series $\xe, \xv$ map to analytic equations of Green's functions by the character, given by the Feynman rules Def.~\ref{def:Feynman rules}.
In principle, this yields non-perturbative equations derived from the structure of perturbative renormalization.
A proper non-perturbative treatment is challenged in the case of MFT mainly for two reasons:
One is that the grafting operator consists of insertions weighted by graph-dependent combinatorial factors, in particular $\maxf(\rg)$, the number of maximal forests in the resulting ribbon graph $\rg$;
this means that in practice the equation has to be evaluated on each graph individually.
The other reason is that there are infinitely many 4-point primitive graphs, and thus infinitely many terms in the equation for $\xv$;
explicit evaluation is thus only possible up to a finite number of loops.
From this perspective, the question whether for any graph all primitive cographs have the same loop order seems less important since it becomes relevant only at higher loops.

These limitations of the applicability of the combinatorial Dyson-Schwinger equations is not specific to MFT but typical for more realistic QFTs.
Given the specific structure of MFT as a non-trivial Euclidean QFT with analytic solutions existing, one could still have expected some peculiar, possibly simplifying algebraic structure.
In this respective, our results so far are negative:
Even though the diagrammatics are considerably simplified due to the reduction on planar ribbon graphs, the Hopf algebra and recursive equations remain of a similar level of complexity as compared to, for example, local quartic scalar field theory.
Thus, it could be that perturbative series in Feynman diagrams are simply not the \vtwo{appropriate} starting point to reveal the analytic structure of MFT.
Still, it is also possible that some crucial aspects, e.g. the right implementation of the analytic Ward identities in the algebraic language, only remain to be uncovered. 
We will leave this for future research.

\section*{Acknowledgement}
We are grateful to Michael Borinsky, David Broadhurst, Henry Kißler, Dirk Kreimer, Erik Panzer and David Prinz for helpful discussions. A.H. is grateful to Harald Grosse and Raimar Wulkenhaar for previous joint work on the Grosse-Wulkenhaar model. A.H.'s work is funded by the German Research Foundation (Deutsche Forschungsgemeinschaft, DFG) through a Walter-Benjamin fellowship, project number 465029630,
while J.T.'s work is funded by DFG through the project ``Non-perturbative group field theory from combinatorial Dyson-Schwinger equations and their algebraic structure'', project number 418838388, and furthermore embedded in Germany's Excellence Strategy EXC~2044--390685587, Mathematics M\"unster: Dynamics–Geometry–Structure.

\section{4-dimensional $\phi^4$ matrix field theory/Grosse-Wulkenhaar model}

This section reviews the developments of $\phi^4$ matrix field theory (see \cite{Branahl:2021slr} for an extended review). 
This theory is also known as the Grosse-Wulkenhaar model since it is a matrix representation of scalar $\phi^4$ Euclidean QFT on $D=4$ dimensional, noncommutative Moyal space at the self-dual point for which Grosse and Wulkenhaar have proven renormalizability to all orders \cite{Grosse:2003aj,Grosse:2004yu,Rivasseau:2005bh}. 
We aim to minimize derivations while providing enough detail to understand how this model differs from ordinary QFT. More literature with additional details and proofs will be provided.
We start with the matrix representation and refer the reader to \cite{Wulkenhaar2019,Hock:2020rje} for information about the connection between QFT on noncommutative geometry and matrix field theory. 
The complex analogue of the $\phi^4$ matrix model with complex fields was considered and analysed in \cite{Branahl:2022xdm,Branahl:2022uge}, and behaves very similar at the simplest topologies. In this article, we will just consider the hermitian case.

Let $H_N$ be the space of hermitian $N\times N$-matrices and let $E\in H_N$ have positive eigenvalues ($E$ plays the role of the Laplacian). Then, we define the measure on the space of matrices $M\in H_N$
\begin{align}\label{measure}
    \d\mu(M)=\frac{\d M \exp[-N\mathrm{Tr}( E M^2+\frac{\lambda}{4}M^4)]}{\int_{H_N} \d M \exp[-N\mathrm{Tr}( E M^2+\frac{\lambda}{4}M^4)]},
\end{align}
where $\d M$ is the Lebesgue measure and  $\lambda\in \mathbb{R}_+$ the coupling constant.  Moments are defined in the sense of a formal matrix model, i.e. the exponential of the interaction term $\exp[-N\mathrm{Tr}( \frac{\lambda}{4}M^4)]$ is expanded and interchanged with the integration. We denote the moments by
\begin{align}\label{moment}
    \langle M_{a_1b_1}...M_{a_nb_n}\rangle:=\int_{H_N}\d\mu(M)\,M_{a_1b_1}...M_{a_nb_n},
\end{align}
where $M_{a b}$ are the components of the matrix $M\in H_N$. Accoording to classical probability theory the moments decompose into cumulants (connected components) $\langle M_{a_1b_1}...M_{a_nb_n}\rangle_c$ through partitions. 
For matrix models, more specifically, these partitions vanish unless the indices form a permutation.
This observation goes essentially back to a work of Brezin, Itzykson, Parisi and Zuber \cite{Brezin:1977sv}. We will denote such a non-vanishing cumulant by the disjoint cycles of the permutation.  Thus, it is natural to define the \textit{correlation function} by
\begin{align}\label{Npoint}
    N^{2-b} G_{|a_1^1...a_{n_1}^1|...|a_1^b...a_{n_b}^b|}:=N^n\langle \prod_{j=1}^b\prod_{i=1}^{n_j}M_{a_i^ja_{i+1}^j}\rangle_c,
\end{align}
where $a_i^j$ are pairwise different, $b$ is the number of disjoint cycles and $n_j$ the length of the $j$-th cycle summing up to $n=n_1+...+n_b$. 
We call the correlation function of \eqref{Npoint} a \textit{$(n_1,...,n_b)$-point function}, which is normalised to large-$N$ asymptotics. 
Each correlation function has a further genus expansion through a natural embedding into genus-$g$ Riemann surfaces which are suppressed by $N^{-2g}$ \cite{tHooft:1973alw,Brezin:1977sv},
\begin{align}
G_{|a_1^1...a_{n_1}^1|...|a_1^b...a_{n_b}^b|}=:\sum_{g=0}^{\infty}N^{-2g}G^{(g)}_{|a_1^1...a_{n_1}^1|...|a_1^b...a_{n_b}^b|}.
\end{align}
All correlation functions are defined for pairwise different $a_i^j$, but have a well-defined limit for coinciding indices.

The dimension of a QFT can be defined through the spectral dimension of the Laplace operator, here $E$, in the limit $N\to \infty$. 
This operator-algebraic definition due to Weyl \cite{Weyl1911} defines the dimension to be the smallest number $D\in\mathbb{R}$ such that
\begin{align}\label{specD}
    \lim_{N\to \infty} \frac{1}{N^{D/2}}\sum_{k=1}^N E_k^{-D/2-\epsilon}
\end{align}
converges for all $\epsilon>0$, where $E_k$'s are the eigenvalues (spectrum) of $E$.

\subsection{Perturbative series}\label{sec:perturbative}
In this subsection, we recall the perturbation theory for this matrix model and introduce its Feynman diagrams and their additional structure of boundary components and genus. We refer to \cite{Hock:2020rje,Branahl:2020uxs} for a comprehensive derivation and much more examples. 
Roughly speaking, this additional structure appears because the action
\begin{align}\label{action}
    N\mathrm{Tr}( E M^2+\frac{\lambda}{4}M^4)]
\end{align}
is invariant \textit{just} under cyclic permutations due to the cyclic invariance of the trace, whereas in ordinary QFTs the action is invariant under any permutation.

Let $(E_1,...,E_N)$ be the eigenvalues of $E$.
The covariance of the \textit{free theory}, i.e. $\lambda=0$, is
\begin{align*}\label{propagator}
    N\langle M_{a,b}M_{c,d}\rangle_{\lambda=0}=\frac{\delta_{a,d}\delta_{b,c}}{E_a+E_b},
\end{align*}
which is typically called the \textit{free propagator}.
Note that the Kronecker $\delta$'s in the free propagator are responsible for the fact that a moment \eqref{moment}  factorises into partitions which necessarily have to form a permutation \eqref{Npoint}, otherwise, it will vanish.

Since we are looking at a formal matrix model, any correlation function can be computed, in principle, from Feynman rules (coming from Wick contraction together with the free propagator) applied to Feynman diagrams. The Feynman diagrams for matrix models are \textit{ribbon graphs} which take into account cyclic ordering at a vertex. This observation goes originally back to \cite{Brezin:1977sv}. 
More formally, we define a ribbon graph to be a connected graph embedded in a Riemann surface (with genus and boundary components) such that the complement of the graph is a disjoint union of disks 
and any disk is adjacent to at most one boundary. 
These disks are called \textit{faces}. A face is called \textit{external} if it is adjacent to the boundary of the ribbon graph, otherwise it is called \textit{internal}. 
In this article we will consider the special case of 4-valent ribbon graphs coming from the quartic interaction \eqref{action}, i.e. all vertices of a ribbon graph have degree 4, except for the vertices ending in a boundary which are univalent. See Fig. \ref{Fig:2+2G1} for an example with genus $g=1$ and $b=2$.

For a Riemann surface of genus $g$ and $b$ boundary components, the number of topologically non-equivalent ribbon graphs for a fixed number $V$ of vertices is finite. In other words, the number of Feynman graphs contributing at a fixed order in perturbation theory for a fixed topology $(g,b)$ is finite. This is easily shown by the Euler characteristic
\begin{align}
    \chi=-2g-b+2
\end{align}
together with Euler's formula
\begin{align}
    \chi=V-E+F,
\end{align}
where $E$ is the number of edges and $F$ the number of faces.

\begin{figure}[htb]
\begin{center}
\includegraphics{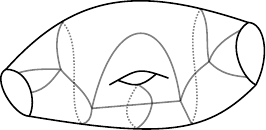}
\caption{\small This is an example of a ribbon graph in the sense above. The Riemann surface has genus $g=1$ and two boundary components $b=2$, thus Euler characteristic $\chi=-2g-b+2=-2$. 
The complement of the grey ribbon graph consists of 4 external faces and one internal face, thus $F=5$. 
The number of vertices is $V=5$ and edges $E=12$, which gives of course the same Euler characteristic $\chi=5-12+5=-2$.}
\label{Fig:2+2G1}
\end{center}
\end{figure}

Similarly to ordinary QFT, there is a fixed external structure labelled by external momenta. 
In matrix theory, these momenta are the eigenvalues $E_n$ of the Laplacian $E$ labelling all \emph{faces} (not edges). 
An edge adjacent to the faces labelled by $E_n$ and $E_m$ receives the weight $\frac{1}{E_n+E_m}$ coming from the free propagator Eq.~\eqref{propagator}. 
The internal faces are the loops and we have to sum over all eigenvalues, i.e.~the spectrum of the Laplacian $E$. \vtwo{Important examples appear if the eigenvalues degenerate, that is, eigenvalues appear with multiplicities. Thus, it is more general to assume that weo have eigenvalues $(E_1,....,E_d)$ together with their multiplicities $(r_1,...,r_d)$ with $\sum_nr_n=N$.}
In summary, we find the following Feynman rules:\\

\begin{definition}[Feynman rules]\label{def:Feynman rules}
Let $\Gamma$ be a ribbon graph with labelled faces and $\varpi$ be a weight defined as the product of
\begin{itemize}
    \item weights $\frac{1}{E_n+E_m}$ for each edge adjacent to faces labelled by $E_{n}$ and $E_m$,
    \item a factor $-\lambda$ for each 4-valent vertex, 
    \item a factor $\frac{1}{N}$ for each internal face. Finally, there is a sum over all the eigenvalues $E_{n}$ \vtwo{with multiplicity $r_n$} for each such internal face.
\end{itemize}
    
\end{definition} 

\begin{example}
    We take the example of Fig. \ref{Fig:Rainbow} and obtain from the Feynman rules the contribution
    \begin{align*}
        \varpi = \frac{1}{(E_a+E_b)^2}\frac{(-\lambda)^2}{N^2}\sum_{n,m}\frac{r_nr_m}{(E_a+E_n)(E_b+E_m)}.
    \end{align*}
\end{example}

\begin{figure}[htb]
{\hspace*{15ex}
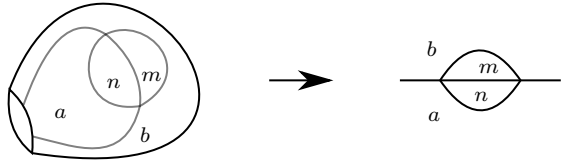}
\caption{\small This is an example of a ribbon graph with one boundary and genus zero. It has two external faces labelled by $E_a$ and $E_b$ and two internal labelled by $E_n$ and $E_m$. Since the ribbon graph is of genus zero, it has an embedding into the plane.}
\label{Fig:Rainbow}
\end{figure}

\begin{example}
    We take the example of Fig. \ref{Fig:2+2G1}. Label the left two external faces by $E_{a_1},E_{a_2}$ and the right two by $E_{b_1},E_{b_2}$. The remaining internal face is labelled by $E_n$. There are four different ways to label the external faces, one of them gives for instance the contribution
    \begin{align*}
        \varpi = \frac{1}{(E_{a_1}+E_{a_2})^2(E_{b_1}+E_{b_2})^2}\frac{(-\lambda)^5}{N}\sum_{n}\frac{r_n}{(E_{a_1}+E_n)^2(E_{b_1}+E_n)^2(E_n+E_n)^4}.
    \end{align*}
\end{example}

It is straightforward to show that any correlation function can be determined from the Feynman diagrams (see, for instance, \cite{Hock:2020rje}):
\begin{proposition}
    An $(n_1,...,n_b)$-point %$(n_1+...+n_b)$-point 
    function of genus $g$ has a perturbative expansion in ribbon graphs $\rg\in\RG{g,b}{n_1,...,n_b}$ of genus $g$ with $b$ boundaries,  
    \begin{align}\label{eq:correlation genus g}
        G^{(g)}_{|a_1^1..|...|..a_{n_b}^b|}=\sum_{\rg\in\RG{g,b}{n_1,...,n_b} }\varpi(\rg) \, ,
    \end{align}
    %where $\RG{g,b}{n_1,...,n_b}$ is the set of ribbon graphs of genus $g$ and $b$ boundaries, 
    where the external faces attached to the $j$-th boundary are labelled with $E_{a^j_i}$ with $i\in \{1,...,n_j\}$ and $j\in \{1,...,b\}$.
\end{proposition}

\begin{remark}
    The number of ribbon graphs for a given genus $g$, a given number of boundary components and a given number $n_1+...+n_b$ of faces attached to the boundaries can be computed via topological recursion \cite{Eynard:2007kz}. The ribbon graphs considered here are dual to quadrangulations of so-called fully simple maps \cite{Borot:2021eif,Bychkov:2021hfh}, which is related to the $x-y$ duality transformation in topological recursion \cite{Eynard_2008,Hock:2022wer,Hock:2022pbw,Alexandrov:2022ydc}.
\end{remark}

The planar (genus $g=0$) 2-point function has up to order $\lambda^2$ the expansion shown in Fig. \ref{Fig:2pointPerturb} consisting of 12 graphs. If we assign to the top face the eigenvalue $E_a$ and bottom face $E_b$, the 2-point function has the leading order expansion
\begin{align*}
    G^{(0)}_{|ab|}=\frac{1}{E_a+E_b}+\frac{1}{(E_a+E_b)^2}\frac{-\lambda}{N}\sum_n \bigg(\frac{r_n}{E_a+E_n}+\frac{r_n}{E_b+E_n}\bigg)+\mathcal{O}(\lambda^2).
\end{align*}
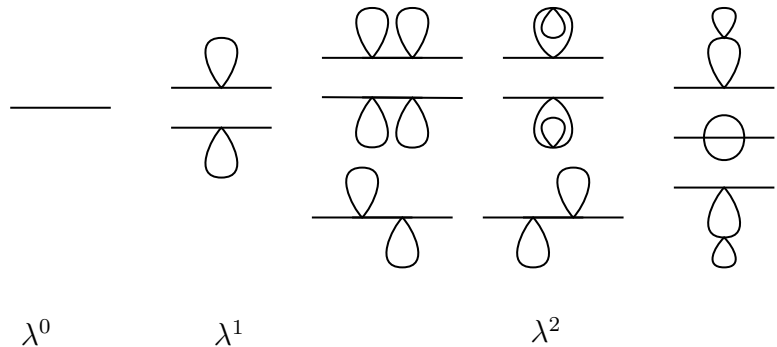
\begin{figure}[htb]
{\hspace*{10ex}
%% Creator: Inkscape 1.1.2 (0a00cf5339, 2022-02-04), www.inkscape.org
%% Accompanies image file '2pointPerturb.pdf' (pdf, eps, ps)
\begingroup%
  \makeatletter%
  \providecommand\color[2][]{%
    \errmessage{(Inkscape) Color is used for the text in Inkscape, but the package 'color.sty' is not loaded}%
    \renewcommand\color[2][]{}%
  }%
  \providecommand\transparent[1]{%
    \errmessage{(Inkscape) Transparency is used (non-zero) for the text in Inkscape, but the package 'transparent.sty' is not loaded}%
    \renewcommand\transparent[1]{}%
  }%
  \providecommand\rotatebox[2]{#2}%
  \newcommand*\fsize{\dimexpr\f@size pt\relax}%
  \newcommand*\lineheight[1]{\fontsize{\fsize}{#1\fsize}\selectfont}%
  \ifx\svgwidth\undefined%
    \setlength{\unitlength}{396.50689461bp}%
    \ifx\svgscale\undefined%
      \relax%
    \else%
      \setlength{\unitlength}{\unitlength * \real{\svgscale}}%
    \fi%
  \else%
    \setlength{\unitlength}{\svgwidth}%
  \fi%
  \global\let\svgwidth\undefined%
  \global\let\svgscale\undefined%
  \makeatother%
  \begin{picture}(1,0.33173155)%
    \lineheight{1}%
    \setlength\tabcolsep{0pt}%
    \put(0,0){\includegraphics[width=\unitlength,page=1]{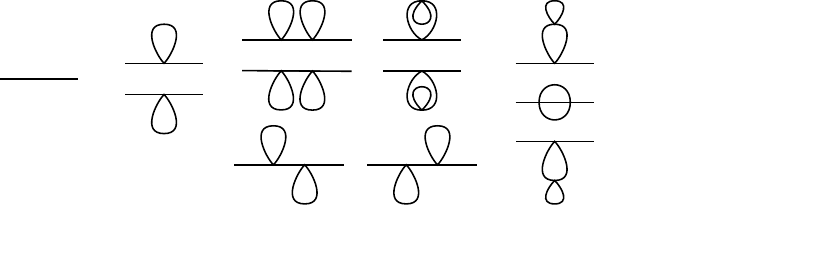}}%
    \put(0.00999943,0.01135115){\color[rgb]{0,0,0}\makebox(0,0)[lt]{\lineheight{1.25}\smash{\begin{tabular}[t]{l}$\lambda^0$\end{tabular}}}}%
    \put(0.19154437,0.01112002){\color[rgb]{0,0,0}\makebox(0,0)[lt]{\lineheight{1.25}\smash{\begin{tabular}[t]{l}$\lambda^1$\end{tabular}}}}%
    \put(0.48947645,0.01380817){\color[rgb]{0,0,0}\makebox(0,0)[lt]{\lineheight{1.25}\smash{\begin{tabular}[t]{l}$\lambda^2$\end{tabular}}}}%
  \end{picture}%
\endgroup%
}
\caption{\small Ribbon graphs contributing to the planar 2-point function up to second order in $\lambda$.}
\label{Fig:2pointPerturb}
\end{figure}
In Fig. \ref{Fig:4pointPerturb}, we have listed all graphs contributing to the planar 4-point function up to the second order in $\lambda^2$, where additional permutations permuting the external faces have to be taken into account.
\begin{figure}[htb]
{\hspace*{10ex}
%% Creator: Inkscape 1.1.2 (0a00cf5339, 2022-02-04), www.inkscape.org
%% Accompanies image file '4pointPerturb.pdf' (pdf, eps, ps)
\begingroup%
  \makeatletter%
  \providecommand\color[2][]{%
    \errmessage{(Inkscape) Color is used for the text in Inkscape, but the package 'color.sty' is not loaded}%
    \renewcommand\color[2][]{}%
  }%
  \providecommand\transparent[1]{%
    \errmessage{(Inkscape) Transparency is used (non-zero) for the text in Inkscape, but the package 'transparent.sty' is not loaded}%
    \renewcommand\transparent[1]{}%
  }%
  \providecommand\rotatebox[2]{#2}%
  \newcommand*\fsize{\dimexpr\f@size pt\relax}%
  \newcommand*\lineheight[1]{\fontsize{\fsize}{#1\fsize}\selectfont}%
  \ifx\svgwidth\undefined%
    \setlength{\unitlength}{357.06219539bp}%
    \ifx\svgscale\undefined%
      \relax%
    \else%
      \setlength{\unitlength}{\unitlength * \real{\svgscale}}%
    \fi%
  \else%
    \setlength{\unitlength}{\svgwidth}%
  \fi%
  \global\let\svgwidth\undefined%
  \global\let\svgscale\undefined%
  \makeatother%
  \begin{picture}(1,0.25481487)%
    \lineheight{1}%
    \setlength\tabcolsep{0pt}%
    \put(0,0){\includegraphics[width=\unitlength,page=1]{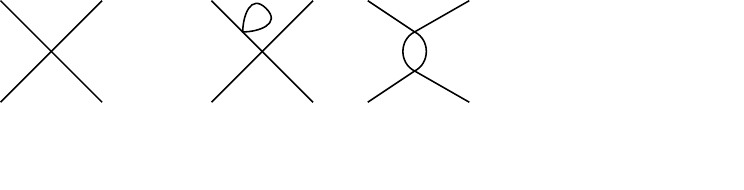}}%
    \put(0.06375688,0.01234845){\color[rgb]{0,0,0}\makebox(0,0)[lt]{\lineheight{1.25}\smash{\begin{tabular}[t]{l}$\lambda^1$\end{tabular}}}}%
    \put(0.43307886,0.01624396){\color[rgb]{0,0,0}\makebox(0,0)[lt]{\lineheight{1.25}\smash{\begin{tabular}[t]{l}$\lambda^2$\end{tabular}}}}%
  \end{picture}%
\endgroup%
}
\caption{\small Ribbon graphs contributing to the planar 4-point function up to second order in $\lambda$ up to permutations. In total there are 10 graphs at the second order, 8 different permutations of the left and 2 of the right. The third order has already 90 contributing ribbon graphs.}
\label{Fig:4pointPerturb}
\end{figure}

For completeness, we list some known generating series for the number of Feynman diagrams (see \cite{Garcia-Failde:2018ylj} for some details).
%, where these Feynman graphs are dual to quadrangulated fully simple maps). 
Let $c=\sqrt{\frac{\sqrt{1+12\lambda}-1}{6 \lambda}}$
and denote $|G^{(0)}_{|ab|}|$ the generating function counting the number of Feynman diagrams contributing to the planar 2-point function at given loop order. This is
\begin{align*}
    |G^{(0)}_{|ab|}|=\lambda c^6+c^2=1+2(-\lambda)+9(-\lambda)^2+54(-\lambda)^3+... \,.
\end{align*}
Similarly, for
%Let $|G^{(0)}_{|abcd|}|$ be the number of Feynman diagrams contributing to 
the planar 4-point function one has
\begin{align*}
    |G^{(0)}_{|abcd|}|=2c^{12}\lambda^2 + c^8 \lambda=-\lambda+10(-\lambda)^2+90(-\lambda)^3+810(-\lambda)^4+... \,.
\end{align*}
The convergence radius is $|\lambda|<\frac{1}{12}$ due to the definition of $c$. 
This extends to all correlation functions and any genus in this model.
%This convergence radius holds for the generating series of any number of any correlation function for any genus in this model. 
Resummability of the genus series itself for matrix models is a very active research field and related to resurgence, but beyond the scope of this article.

%The correlation functions defined as connected expectation value do not consist of only 
In the same way one can also count the number of 1PI diagramms, that is bridgeless ribbon graphs, which are the crucial objects for renormalization (see Sec. \ref{Sec:Hopf}). 
These numbers follow from the previous results using the inverse of the geometric series. 
Let $|\Pi^{(0)}_{|ab|}|$  be the generating function counting the number of planar 1PI 2-point graphs at given loop order and $|\Pi^{(0)}_{|abcd|}|$ for planar 1PI 4-point graphs. These are given by
\begin{align}\label{eq:correlations1PI}
    |\Pi^{(0)}_{|ab|}| &= \frac{|G^{(0)}_{|ab|}|-1}{|G^{(0)}_{|ab|}|}=\frac{\lambda c^6+c^2-1}{\lambda c^6+c^2}=2(-\lambda)+5(-\lambda)^2+26(-\lambda)^3+...\\
    |\Pi^{(0)}_{|abcd|}| &= \frac{|G^{(0)}_{|abcd|}|}{|G^{(0)}_{|ab|}|^4}=\frac{2c^{12}\lambda^2 + c^8 \lambda}{(\lambda c^6+c^2)^4}=(-\lambda)+2(-\lambda)^2+14(-\lambda)^3+... \,,
\end{align}
where $c=\sqrt{\frac{\sqrt{1+12\lambda}-1}{6 \lambda}}$ implies the same radius of convergence.

\subsection{Dyson-Schwinger equations and Ward identity}
Unlike ordinary QFT models, the $\phi^4$ matrix model \eqref{measure} has a closed, convergent (in $\lambda$) expression for any $(n_1+...+n_b)$-point function at any genus $g$. 
This is possible due to the additional expansion in the genus which separates some part of the resummability issue apparent in QFT from renormalization. 

We will focus on the explicit result of the planar 2- and 4-point function since these are the only ones directly affected by renormalization. The first step towards exact solutions are Dyson-Schwinger equations (DSE), which are in general equations between different correlation functions. The DSEs for the $\phi^4$ matrix model decouple due to the genus expansion for large $N$ and due to Ward identities. 
The closed DSE for the planar 2-point function was derived in \cite{Grosse:2009pa} (and generalised in \cite{Hock:2020rje,Branahl:2020yru}) based on the Ward identity derived in \cite{Disertori:2006nq} (and generalised in \cite{Hock:2018wup}).

We recall for pedagogical reasons the DSE for the planar 2-point function in matrix basis (before applying the Ward identity), which is essentially achieved by classical methods of integration by parts:
\begin{align}
G^{(0)}_{|ab|}&=\frac{1}{E_a+E_b}- 
\frac{\lambda}{E_a+E_b}
\Big\{ 
\frac{1}{N^2}
\sum_{k,l} r_kr_l G^{(0)}_{|aklb|}
+ \frac{1}{N}
\sum_{k} r_k (G^{(0)}_{|ab|}G^{(0)}_{|ak|}+G^{(0)}_{|ab|} 
G^{(0)}_{|kb|})\Big\}.
\label{DS2}
\end{align}
The DSE has the following graphical interpretation\\
\vspace*{2ex}
%% Creator: Inkscape 1.1.2 (0a00cf5339, 2022-02-04), www.inkscape.org
%% Accompanies image file '2PDSE.pdf' (pdf, eps, ps)
\begingroup%
  \makeatletter%
  \providecommand\color[2][]{%
    \errmessage{(Inkscape) Color is used for the text in Inkscape, but the package 'color.sty' is not loaded}%
    \renewcommand\color[2][]{}%
  }%
  \providecommand\transparent[1]{%
    \errmessage{(Inkscape) Transparency is used (non-zero) for the text in Inkscape, but the package 'transparent.sty' is not loaded}%
    \renewcommand\transparent[1]{}%
  }%
  \providecommand\rotatebox[2]{#2}%
  \newcommand*\fsize{\dimexpr\f@size pt\relax}%
  \newcommand*\lineheight[1]{\fontsize{\fsize}{#1\fsize}\selectfont}%
  \ifx\svgwidth\undefined%
    \setlength{\unitlength}{316.50005799bp}%
    \ifx\svgscale\undefined%
      \relax%
    \else%
      \setlength{\unitlength}{\unitlength * \real{\svgscale}}%
    \fi%
  \else%
    \setlength{\unitlength}{\svgwidth}%
  \fi%
  \global\let\svgwidth\undefined%
  \global\let\svgscale\undefined%
  \makeatother%
  \begin{picture}(1,0.15553999)%
    \lineheight{1}%
    \setlength\tabcolsep{0pt}%
    \put(0,0){\includegraphics[width=\unitlength,page=1]{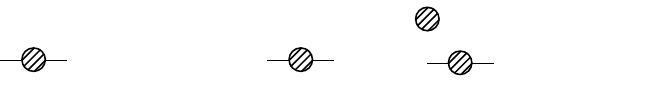}}%
    \put(0.11587144,0.05550796){\color[rgb]{0,0,0}\makebox(0,0)[lt]{\lineheight{1.25}\smash{\begin{tabular}[t]{l}$=$\end{tabular}}}}%
    \put(0.00632205,0.08352454){\color[rgb]{0,0,0}\makebox(0,0)[lt]{\lineheight{1.25}\smash{\begin{tabular}[t]{l}$a$\end{tabular}}}}%
    \put(0.00712049,0.01393101){\color[rgb]{0,0,0}\makebox(0,0)[lt]{\lineheight{1.25}\smash{\begin{tabular}[t]{l}$b$\end{tabular}}}}%
    \put(0.2925224,0.05629161){\color[rgb]{0,0,0}\makebox(0,0)[lt]{\lineheight{1.25}\smash{\begin{tabular}[t]{l}$+$\end{tabular}}}}%
    \put(0.52823479,0.05928833){\color[rgb]{0,0,0}\makebox(0,0)[lt]{\lineheight{1.25}\smash{\begin{tabular}[t]{l}$+$\end{tabular}}}}%
    \put(0.79814903,0.05967797){\color[rgb]{0,0,0}\makebox(0,0)[lt]{\lineheight{1.25}\smash{\begin{tabular}[t]{l}$+$\end{tabular}}}}%
    \put(0,0){\includegraphics[width=\unitlength,page=2]{2PDSE.pdf}}%
  \end{picture}%
\endgroup%
\\
\noindent
where %% Creator: Inkscape 1.1.2 (0a00cf5339, 2022-02-04), www.inkscape.org
%% Accompanies image file '2P.pdf' (pdf, eps, ps)
\begingroup%
  \makeatletter%
  \providecommand\color[2][]{%
    \errmessage{(Inkscape) Color is used for the text in Inkscape, but the package 'color.sty' is not loaded}%
    \renewcommand\color[2][]{}%
  }%
  \providecommand\transparent[1]{%
    \errmessage{(Inkscape) Transparency is used (non-zero) for the text in Inkscape, but the package 'transparent.sty' is not loaded}%
    \renewcommand\transparent[1]{}%
  }%
  \providecommand\rotatebox[2]{#2}%
  \newcommand*\fsize{\dimexpr\f@size pt\relax}%
  \newcommand*\lineheight[1]{\fontsize{\fsize}{#1\fsize}\selectfont}%
  \ifx\svgwidth\undefined%
    \setlength{\unitlength}{32.25000189bp}%
    \ifx\svgscale\undefined%
      \relax%
    \else%
      \setlength{\unitlength}{\unitlength * \real{\svgscale}}%
    \fi%
  \else%
    \setlength{\unitlength}{\svgwidth}%
  \fi%
  \global\let\svgwidth\undefined%
  \global\let\svgscale\undefined%
  \makeatother%
  \begin{picture}(1,0.37212954)%
    \lineheight{1}%
    \setlength\tabcolsep{0pt}%
    \put(0,0){\includegraphics[width=\unitlength,page=1]{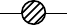}}%
  \end{picture}%
\endgroup%
 corresponds to the planar 2-point function, the straight line to the free propagator and %% Creator: Inkscape 1.1.2 (0a00cf5339, 2022-02-04), www.inkscape.org
%% Accompanies image file '4P.pdf' (pdf, eps, ps)
\begingroup%
  \makeatletter%
  \providecommand\color[2][]{%
    \errmessage{(Inkscape) Color is used for the text in Inkscape, but the package 'color.sty' is not loaded}%
    \renewcommand\color[2][]{}%
  }%
  \providecommand\transparent[1]{%
    \errmessage{(Inkscape) Transparency is used (non-zero) for the text in Inkscape, but the package 'transparent.sty' is not loaded}%
    \renewcommand\transparent[1]{}%
  }%
  \providecommand\rotatebox[2]{#2}%
  \newcommand*\fsize{\dimexpr\f@size pt\relax}%
  \newcommand*\lineheight[1]{\fontsize{\fsize}{#1\fsize}\selectfont}%
  \ifx\svgwidth\undefined%
    \setlength{\unitlength}{21.75001228bp}%
    \ifx\svgscale\undefined%
      \relax%
    \else%
      \setlength{\unitlength}{\unitlength * \real{\svgscale}}%
    \fi%
  \else%
    \setlength{\unitlength}{\svgwidth}%
  \fi%
  \global\let\svgwidth\undefined%
  \global\let\svgscale\undefined%
  \makeatother%
  \begin{picture}(1,0.99999948)%
    \lineheight{1}%
    \setlength\tabcolsep{0pt}%
    \put(0,0){\includegraphics[width=\unitlength,page=1]{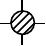}}%
  \end{picture}%
\endgroup%
 to the planar 4-point function. A closed face in the graphical formula corresponds to a summation in \eqref{DS2} and the first separate vertex to the factor $-\lambda$.

The Ward identity is another way to relate correlation functions due to symmetry transformations under which the partition function is invariant, but neither the correlation function nor the action. From the Ward identity \vtwo{coming from unitary transformation of $M$ in \eqref{measure}}, one deduces (see \cite{Grosse:2012uv,Hock:2020rje,schuermann_wulkenhaar_2022,deJong:2019oez} for details and more general formulas)
\begin{align}\label{DSE2W}
    \frac{1}{N} \sum_{k}r_k G^{(0)}_{|akpb|} 
+G^{(0)}_{|ab|} G^{(0)}_{|pb|} 
&=-\frac{G^{(0)}_{|pb|} - G^{(0)}_{|ab|}}{E_p-E_a}
 \;,\\\label{W4P}
 G^{(0)}_{|abcd|} &=-\lambda\frac{G^{(0)}_{|ab|}G^{(0)}_{|cd|}-G^{(0)}_{|ad|}G^{(0)}_{|cb|}}{(E_a-E_c)(E_b-E_d)}.
\end{align}

Note that these two identities lead to two further relations between the planar 4-point function and the planar 2-point function. Even more interestingly, the planar 4-point function is entirely determined by an algebraic formula form planar 2-point functions in \eqref{W4P}. 

Inserting the first relation \eqref{DSE2W} obtained from the Ward identity into the DSE \eqref{DS2}, we derive the closed DSE for the planar 2-point function
\begin{align}
G^{(0)}_{|ab|}&=\frac{1}{E_a+E_b}- 
\frac{\lambda}{E_a+E_b}
\Big\{ 
\frac{1}{N}
\sum_{k}r_k \Big(-\frac{G^{(0)}_{|pb|} - G^{(0)}_{|ab|}}{E_p-E_a}+
G^{(0)}_{|ab|}G^{(0)}_{|ak|}\Big)\Big\}.
\label{cDS2}
\end{align}

For finite $N$, the solution of the 2-point function and more fundamental correlation functions was derived in \cite{Hock:2021tbl} for genus $g=0$ and in \cite{Hock:2023nki} for genus $g=1$. The algebraic structure behind this is \textit{blobbed topological recursion} \cite{Borot:2015hna} a generalization of topological recursion \cite{Eynard:2007kz}. However, these theories do not play a role from renormalization perspective.

\subsection{Renormalization \& exact solution in 4 dimensions}
More interestingly, going to a continuum limit corresponding to the 4-dimensional Moyal space, the correlation functions become well-defined after appropriate renormalization. In other words, renormalization is compatible, in a certain sense, with the exact solution already established at finite $N$ at the level of a formal expansion.

Now, we will focus on the 4-dimensional Moyal space. The renormalized energy eigenvalues for the Laplacian $E$ are given by the sequence 
\[ 
(E_1,...,E_N)=(e_1,e_2,e_2,e_3...) \, ,
\quad\textrm{ where } 
e_k=Z \left(\frac{k}{N}+\frac{\mu^2_{bare}}{2} \right) \, ,
\]
and each $e_k$ has multiplicity $r_k=k$, i.e. $e_1$ appears once, $e_2$ twice etc. The constant $Z$ is the field renormalization constant and $\mu_{bare}$ the mass renormalization constant. The largest eigenvalue $e_d$ fixes the cut-off $\Lambda^2$. \vtwo{First a limit $N,d\to \infty$ with constant $\Lambda^2:=\lim_{d,N\to \infty}e_d=\lim_{d,N\to \infty}Z \left(\frac{d}{N}+\frac{\mu^2_{bare}}{2} \right)$ is performed. In this limit the size of the matrices increases as well as the number of eigenvalues, while the largest eigenvalue $e_{d\to \infty}$ stays constant.} Renormalization is needed to perform an additional well-defined limit, when $\Lambda\to \infty$. 
Since the eigenvalue $e_k$ itself increases linearly in $k$, as well as its multiplicity, we find from \eqref{specD} a 4-dimensional theory in the large $N$ limit
\begin{align*}
    &\lim_{N\to \infty}\frac{1}{N^{D/2}}\sum_{k=1}^N E_k^{-D/2-\epsilon}=\lim_{N\to \infty} \frac{1}{N^{D/2}}\sum_{k=1}^N e_k^{-D/2-\epsilon}\cdot k\sim \sum_{k}k^{-D/2-\epsilon+1}<\infty\\&
    \Rightarrow\qquad D=4.
\end{align*}

A canonical analytic continuation of the 2-point function $G^{(0)}(z,w)$ is defined through Dyson-Schwinger equation to coincide at the points $e_a$, i.e. $G^{(0)}(e_a,e_b)=G^{(0)}_{|ab|}$. In this setting the Dyson-Schwinger equation for the planar 2-point function together with all renormalization constants reads \cite{Grosse:2012uv}:
\begin{align}
&\bigg[z+w+\mu^2_{bare}+\lambda\dashint_0^{\Lambda^2}\d t\,\varrho_0(t)\bigg(ZG^{(0)}(z,t)+\frac{1}{t-z}\bigg)\bigg]Z G^{(0)}(z,w)
\label{eq:GcomplexCont}
\\
&=1+\lambda\dashint_0^{\Lambda^2}
dt\,\varrho_0(t)
\Big(
 \frac{ZG^{(0)}(t,w)  }{t-z}\Big) ,
\nonumber
\end{align}
where $\dashint$ is the Cauchy principal value, the measure
$\varrho_0(t):=\frac{1}{N}\sum_{k=1}^d k\,\delta(t-e_k)$ and we have neglected all $\frac{1}{N^2}$ contributions. 

The large $N$-limit can be seen as a continuum limit in which the eigenvalues $e_k$ converge to the interval $[0,\Lambda^2]$ with fixed $\Lambda^2$ \vtwo{coming from fixing the largest eigenvalue $e_d$ to $\Lambda^2$, but sending simultaneously $d,N\to \infty$}. The measure $\varrho_0(t)$ converges to $ \varrho_0(t)=t$ and the renormalization constants $Z$ and 
$\mu_{bare}$ obtain a dependence on the cut-off $\Lambda$ such that the UV-limit $\Lambda\to \infty$ is well-defined on the level of correlation functions.

The exact solution of the nonlinear DSE \eqref{eq:GcomplexCont} was derived in \cite{Grosse:2019jnv} based on an idea of the 2-dimensional solution \cite{Panzer:2018tvy}. This was achieved by essentially using advanced complex analysis with singular integral equations theory together with Lagrange-B\"urmann formula. We will state the result for the 4-dimensional case first for a more general measure $\varrho_0(t)\sim t$. Define implicitly for a given $\varrho_0(t)$ the two functions $\varrho_\lambda(t)$ and $R(t)$ via the system of equations
\begin{align}\label{R}
    R(z)=&z-\lambda z^2\int_0^{\infty}\frac{dt\, \varrho_\lambda(t)}{(\mu^2+t)^2(\mu^2+t+z)},\\\nonumber
    \varrho_\lambda(t)=&\varrho_0(R(t)),
\end{align}
where $\mu^2$ is the renormalized mass. The function $\varrho_\lambda(t)$ takes the role of a \textit{ $\lambda$-deformed measure} with $\lim_{\lambda\to 0}\varrho_\lambda(t)=\varrho_0(t)$.

\begin{theorem}[\cite{Grosse:2019jnv}]\label{thm:hypergeom}
     The renormalized 2-point function of the $\phi^4$ matricial QFT-model in 4 dimensions is given by
     \begin{align}\label{2Pasy}
  G^{(0)}(a,b) \!=\!\frac{(\mu^2{+}a{+}b)\exp\Bigg\{\!\!
    \mbox{\small$\displaystyle\frac{1}{2\pi \mathrm{i}} \int_{-\infty}^\infty \!\!\!\!\!dt \log\bigg(\frac{a-R(-\frac{\mu^2}{2}{-}\mathrm{i} t)}{a-(-\frac{\mu^2}{2}{-}\mathrm{i} t)}\bigg)\frac{d}{dt}\log\bigg(\frac{b-R(-\frac{\mu^2}{2}{+}\mathrm{i} t)}{
   b-(-\frac{\mu^2}{2}{+}\mathrm{i} t)}\bigg)$}
   \!\!\Bigg\}}{(\mu^2+b+R^{-1}(a))(\mu^2+a+R^{-1}(b))},
	\end{align}
 where $R(t)$ is implicitly defined via \eqref{R}.
\end{theorem}
We emphasise that this solution is exact for any H\"older continuous measure $\varrho_0(t)$ and simplifies even further on the 4-dimensional Moyal space, where $\varrho_0(t)=t$. This implies $\varrho_\lambda(t)=R(t)$ and thus
\begin{align}\label{linEQ}
    R(z)=&z-\lambda z^2\int_0^{\infty}\frac{dt\,  R(t)}{(\mu^2+t)^2(\mu^2+t+z)}.
\end{align}
The solution has a closed expression in terms of a Gaussian hypergeometric function:
\begin{theorem}[\cite{Grosse:2019qps}]
The linear integral equation \eqref{linEQ} is solved by
\begin{align}
		R(z)=z  \;_2F_1\Big(\genfrac{}{}{0pt}{}{
			\alpha_\lambda,\;1-\alpha_\lambda}{2}\Big|-\frac{z}{\mu^2}\Big),\quad
		\text{where } ~
		\alpha_\lambda:=\begin{cases}
                  \frac{\arcsin(\lambda\pi)}{\pi} & \text{for }|\lambda|\leq
                  \frac{1}{\pi} \;,\\ \frac{1}{2}
                  +\mathrm{i} \frac{\mathrm{arcosh}(\lambda\pi)}{\pi}
                  & \text{for } \lambda \geq \frac{1}{\pi}\;.
                \end{cases}
              \end{align}
\end{theorem}

 An important fact is that the linear dependence of $\lambda$ within
the integral equation \eqref{R} is collected into a highly nonlinear
dependence given by the $\arcsin$-function into the coefficients of
the hypergeometric function. Thus, the convergence radius turns out to be $|\lambda|<\frac{1}{\pi}$.
The two functions $\varrho_\lambda=R$ and $\varrho_0$ have a different asymptotic behaviour. A hypergeometric function behaves like
\begin{align}
	\,_2F_1\Big(\genfrac{}{}{0pt}{}{
          a,\;1-a}{2}\Big|-x\Big) \stackrel{x\to \infty}{\sim}
        \frac{1}{x^a}\quad \text{for } |a|<\frac{1}{2}.
\end{align} 

Together with the definition of the spectral dimension
\eqref{specD}, we conclude:
\begin{corollary}[\cite{Grosse:2019qps}]\label{Cor:eff}
  For $|\lambda|<\frac{1}{\pi}$, the $\lambda$-deformed measure
  $\varrho_\lambda=R$ of four-dimensional Moyal space has effectively the
  spectral dimension $D_\lambda=4-2\frac{\arcsin(\lambda \pi)}{\pi}$.
\end{corollary}\noindent
The 4-dimensional matricial $\phi^4$ model admits a dimensional
drop to an effective spectral dimension related to an effective
spectral measure. This is revealed by resummation of the planar connected 2-point function.

From the QFT perspective, this dimension drop is the
most important result. 
It means \emph{that the $\phi^{4}$ theory on 4D Moyal space is non-trivial} after degeneration into a genus expansion. 
In other words, it is well-defined at any (energy) scale $\Lambda^2\to \infty$. 
So far, the relation of this dimensional drop to the anomalous dimension of the theory has not been studied%
\footnote{We are grateful to David Broadhust and Michael Borinsky for raising this equation on the conference "From perturbative to non-perturbative QFT" which took place in Münster 2023}.
However, here we find that it is essentially the same:
\begin{corollary}\label{cor:anomalous dimension}
    The anomalous dimension $\gamma$ of the field renormalization constant $Z$ is half of the effective dimension drop of Corollary \ref{Cor:eff}, i.e.
    \begin{align}\label{eq:anomalous dimension}
        \gamma=\lim_{\Lambda^2\to \infty} -\frac{\partial \log Z}{\partial \log \Lambda^2}=-\frac{\arctan (\lambda \pi)}{\pi }.
    \end{align}
    \begin{proof}
        The field renormalization constant $Z$ is given in \cite{Grosse:2019jnv} as
        \begin{align*}
            Z=C_r e^{\mathcal{H}^{\Lambda}_r[\tau_r(\bullet)]},
        \end{align*}
        where $C_r$ is some finite constant, $\mathcal{H}^{\Lambda}_r[f(\bullet)]$ is the Hilbert transform $\mathcal{H}^{\Lambda}_r[f(\bullet)]=\frac{1}{\pi }\int_{[0,\Lambda^2]\setminus [r-\epsilon,r+\epsilon]}\frac{dt\, f(t)}{t-r}$ 
        and $\tau_r(t)=\arctan\frac{\lambda\pi t}{\mathrm{Re}(r+I(t+i \epsilon))}$ an auxilliary function. %, see \cite{Grosse:2019jnv} for all the details. 
        Furthermore, the function $I(x)$ is defined through $R(z)$ of Theorem \ref{thm:hypergeom} by $I(x)=-R(-\mu^2-R^{-1}(x))$. With $I(x)\overset{x\to \infty}{\sim} x +\mathcal{O}(1)$, which was also shown in \cite{Grosse:2019jnv}, we finally find
        \begin{align*}
            \gamma=&\lim_{\Lambda^2\to \infty} -\frac{\partial \log Z}{\partial \log \Lambda^2}=\lim_{\Lambda^2\to \infty}- \Lambda^2\frac{\partial \mathcal{H}^{\Lambda}_r[\tau_r(\bullet)]}{\partial \Lambda^2}\\
            =&\lim_{\Lambda^2\to \infty} -\frac{\Lambda^2}{\pi}\frac{\tau_r(\Lambda^2)}{\Lambda^2-r}=\lim_{\Lambda^2\to \infty} -\frac{1}{\pi}\arctan\frac{\lambda\pi \Lambda^2}{\mathrm{Re}(r+I(\Lambda^2+i \epsilon))}\\
            =&-\frac{1}{\pi}\arctan \lambda \pi.
        \end{align*}
    \end{proof}
\end{corollary}
The difference of the factor of two between the spectral and anomalous dimension,  $D_\lambda = 2(2+\gamma)$, has an explanation in terms of the relation between $\phi_D^4$ noncommutative field theory and $\phi_{d,r=2}^4$ matrix field theory:
As a noncommutative QFT, the Grosse-Wulkenhaar model has space-time dimension $D=4$.
Its matrix representation Eq.~\eqref{measure} viewed as an $r=2$ tensor (i.e.~matrix) field theory \cite{Thurigen:2102}, however, is a field theory on a $d=2$ dimensional domain since variables are integrated with measure $\varrho(t)\d t = t \d t$, Eq.~\eqref{eq:GcomplexCont}.
The anomalous dimension is the renormalization correction to this dimension and therefore half the value of the dimensional drop of the spectral dimension related to the space-time dimension $D=4$ of the non-commutative QFT.

The Landau ghost problem \cite{Landau:1954??}, or triviality, is a fundamental problem in QFT. The Standard Model however is rescued by the discovery of asymptotic
freedom coming from non-Abelian Yang-Mills theories. However, a simpler 4D QFT-model without the triviality problem is
not known so far. For the ordinary scalar $\phi^4$-model, triviality was proved
in $D=4+\epsilon$ dimensions \cite{Aizenman:1981du, Frohlich:1982tw} and recently by Aizenman and Duminil-Copin
in $D=4$ \cite{Aizenman:2019yuo} as (marginal) triviality.
Therefore, the construction and detailed understanding of the renormalization procedure of a simple, solvable and non-trivial
QFT-model in four dimensions is a
major task for renormalization theory.

The perturbative renormalization procedure invented by Connes and Kreimer using the algebraic structure of Hopf algebras \cite{Connes:1998qv,Connes:1999yr} was applied efficiently in higher order perturbation theory \cite{Broadhurst:2000dq}. This theory gave a new perspective on renormalization since it includes not just Zimmermann's forest formula but also other renormalization schemes into a Hopf algebra. The precise Hopf algebraic structure of the Grosse-Wulkenhaar model was claimed to be already analysed in \cite{Tanasa:2009hb}. We will give a more comprehensive analysis and fix the artificially included terms by the correct consideration of higher boundary structures and the correct way of contracting them.

\

\tikzstyle{v} =  [circle, draw=black, line width=.2pt, fill=black, inner sep=0pt, minimum size=1.5mm]

\section{Hopf algebra and Hochschild cocyles}\label{Sec:Hopf}

Renormalization of divergent amplitudes in perturbative QFT has an underlying Hopf-algebraic structure.
One may disentangle combinatorics and analytics by considering the free algebra $\btg$ generated by the set of connected ribbon graphs~$\RG{}{}$. 
Then, the Feynman rules Def.~\ref{def:Feynman rules} can be seen as an algebra homomorphism $\varpi:\btg\to\uca$. % on the algebra $\btg$. 
This allows to understand correlation functions \eqref{eq:correlation genus g} as evaluations 
\[
G^{(g)}_{|a_1^1..|...|..a_{n_b}^b|}= \varpi\left(X^{(g)}_{n_1,...,n_b}\right) (\vec{a}^1,...,\vec{a}^b)
\]
of formal series $X^{(g)}_{n_1,...,n_b} = \sum_{\substack{\rg\in \RG{g,b}{}}} \rg$ 
in the subalgebra $\btg_{g,b}$ generated by connected ribbon graphs with given genus $g$ and $b$ labelled boundaries.

Since ribbon graphs are the special case of strand graphs with two strands per edge, this is an instance of the algebra of strand graphs, the Feynman diagrams of the more general class of combinatorially non-local field theories (cNLFT).
It has been shown that, like for field theory with point-like interactions, such algebras extend to Hopf algebras \cite{Thurigen:2102}.
Thus, perturbative renormalization of matrix theory can be described by a Hopf algebra of divergent bridgeless (1PI) ribbon graphs which we detail in the following.

\tikzstyle{vb} = [coordinate] % no red dots on edges from here on
\subsection{Hopf algebra of perturbative renormalization}

The central operation in perturbative renormalization is to identify for a graph $\sg$ certain subgraphs $\sg\sgr\rg$ and their counterpart, that is, \emph{contraction} $\rg/\sg$.
\vtwo{To be able to define these operations on ribbon graphs rigorously, we need to introduce a combinatorial definition of ribbon graphs:}

\begin{definition}%[ribbon graph, combinatorially]
\label{def:ribbon graph}
    A \emph{ribbon graph} $G=(\H,\sigma,\ei)$ is a triple of a finite set $\H$ of \emph{half edges} and two permutations $\sigma, \ei:\H\to\H$ where $\ei$ is an involution, \emph{including} fixed points. % which are understood as open half-edges, i.e.~external edges. % $\ei:\H\to\H$. 
    The graph interpretation is the following:
    \begin{itemize}
    \item  Each cycle of $\sigma$ defines an oriented \emph{vertex} and
    \item  each cycle of $\ei$ with two elements defines an \emph{(internal) edge}. 
    \item Fixed points $\ei(h)= h$ define \emph{external edges}, i.e.~open half edges~$h$. 
    This yields a partition into internal and external half edges, $\H = \Hint\sqcup\Hext$.
    \item \vtwo{Each cycle of $\sigma^{-1}\circ\ei$ defines either an \emph{internal face} if it consist only of internal half edges, 
    \item or a \emph{boundary} if it contains an external half edge.}
    \end{itemize}
\end{definition}

This is %the usual QFT Feynman graph. It is 
the combinatorial definition of the above notion of embedded ribbon graphs in Sec.~\ref{sec:perturbative}, more precisely, of their equivalence class under homeomorphisms.
Thus, such a ribbon graph is dual to a discrete surface, more precisely to a combinatorial map as explained in detail in App.~\ref{sec:ribbon graphs}.

\

There are several diagrams isomorphic under relabelling, but we can choose a canonical labelling in our case for simplicity.
In general, a combinatorial map is considered as the equivalence class of isomorphic relabellings \cite{Eynard:2016yaa}.
In field theory, however, we are always dealing with open diagrams where external labels are fixed. In the case of ribbon graphs, this allows to canonically choose a representative of this equivalence class in the following way:
Denote the unique external half-edge that belongs to the external face of momentum $p_1$ as $1\in\Hext$.
Define all vertices as cycles with subsequent labels, $\mathcal{C}(\sigma^{-1})=(1234)(5678)...$ . 
Finally, starting with the first vertex, define edges by involutions from the half edge with the lowest possible label to the lowest next possible one %, e.g. in the example Eq.~\eqref{eq:fish graph}, first $\ei(3)=5$ and next $\ei(4)=8$, 
and repeat with yet unpaired half-edges of the next vertex.
In this way, we obtain a canonical, ordered labelling for any ribbon graph.
For example, the planar fish graph is
\[\label{eq:fish graph}
\fishl \equiv \left(\{1,2,3,4,5,6,7,8\},(4321)(8765),(1)(2)(35)(48)(6)(7))\right) \, .
\]

Ribbon graphs as defined in Def.~\ref{def:ribbon graph} %since these 
are a special case of 2-graphs with two strands at each edge \cite{Thurigen:2102}. 
Thus, the concept of contraction can be directly imported from the 2-graph definition:

\newpage

\begin{definition}\label{def:contraction}
    $\sg=(\H',\sigma',\ei')$ is a ribbon subgraph  of $\rg = (\H,\sigma,\ei)$, write $\sg\sgr\rg$, iff  $\H'=\H$, $\sigma'=\sigma$ and the set of pairs of $\ei'$ %(without fixed points) 
    is a subset of those of $\ei$.

    Then, contraction $\rg/\sg = (\H_{\rg/\sg},\sigma_{\rg/\sg},\ei_{\rg/\sg})$ 
    is defined as shrinking the edges of $\sg$ in $\rg$, that is 
    \begin{itemize}
        \item $\H_{\rg/\sg}= \H'_\textrm{ext}$, internal half edges of the subgraph are deleted,
        \item $\mathcal{C}(\sigma_{\rg/\sg}) = \mathcal{C}(\ei'\circ\sigma) \vert_{\H_{\rg/\sg}}$ where the restriction to $\H_{\rg/\sg}$ means deleting all other elements in the cycles $\mathcal{C}(\ei'\circ\sigma)$,
        \item and $\ei_{\rg/\sg} = \ei\vert_{\H_{\rg/\sg}}$.
    \end{itemize}
\end{definition}
According to this definition following the convention of \cite{Borinsky:2018}, the notion of subgraph refers to a graph possibly with several connected components, covering all the vertices of the original graph $\rg$. 
In this sense, one might also call it a \emph{decomposition} of $\rg$. 

\vtwo{It has to be noted that the framework of combinatorial ribbon graphs can match the physics only for contractions which do not lead to so-called \emph{multi-trace} vertices.
To properly take these into account, the more general framework of 2-graphs \cite{Thurigen:2102} would be needed.
The reason is that, in physics,}
%Crucially, 
each connected component of the subgraph $H$ is \vtwo{meant to shrink} to a single vertex in $\rg/\sg$ with the vertex structure given by its boundary. 
\vtwo{However, if a component $H_i$ of $H$ has more than one boundary, $b_i>1$, the contraction $G/H$ as defined in \ref{def:contraction} yields $b_i$ cycles; according the definition of combinatorial ribbon graphs, this defines $b_i$ vertices in $G/H$ but from a physics perspective it is understood as a single $b_i$-fold trace vertex.
}

For example,
contracting the above fish diagram (with a single boundary) yields a single quartic vertex
\[
\fishl \bigg/ \fishl \cong \left(\{1,2,6,7\},(1267), \id \right)
\cong \vertexl
\]
Contracting a ribbon graph within itself yields its external structure, the \emph{residue}  
    \[\label{eq:residue}
    \res:\RG{}{}\to\Res^{*},\quad \rg\mapsto\rg/\rg.
    \]
where we denote \vtwo{$\Res^{*}\subset \RG{}{}$ the subset of the set $\RG{}{}$ of ribbon graphs} with only external edges, i.e.~with $\ei = \id$, and $\Res \subset \Res^{*}$ the subset of those with single connected components, that is just single vertices.

On the other hand, the ribbon graph
\[\label{eq:2-2-fish}
\fishtwobl \cong \left(\{1,2,3,4,5,6,7,8\},(4321)(8765),(25)(47)\right)
\]
has two boundaries
\[
\mathcal{C}(\sigma^{-1} \circ \ei) = (1267) (3485)
\]
such that also its residue
\[\label{eq:multitrace residue}
\fishtwobl \bigg/ \fishtwobl \cong \left(\{1,3,6,8\},(16)(38), \id \right)
\cong \multivertexl%\vertexmultil
\]
has two cycles in $\sigma$ even though this represents a single vertex in the physics' sense. 
Such a %so-called 
\emph{multi-trace} vertices is not properly included in the definition of combinatorial \vtwo{ribbon graphs Def.~\ref{def:ribbon graph} in line with the common concept of combinatorial maps; instead,} 
the recent more general definition of 2-graphs \cite{Thurigen:2102} would be necessary to cover these cases;
for the current purpose the framework of combinatorial maps is sufficient, though, as we will deal only with connected diagrams with one boundary component \vtwo{in the renormalization Hopf algebra of quartic MFT}.

\

\tikzstyle{v} =  [circle, draw=black, line width=.2pt, fill=black, inner sep=0pt, minimum size=1mm]
    
Contraction gives rise to the general notion of a coproduct
\[\label{eq:coproduct}
\Delta_{\RG{}{}} : \rg \mapsto \sum_{\sg\sgr\rg} \sg \otimes \rg/\sg  \, ,
\]
which defines a coalgebra that can be extended to a Hopf algebra \cite{Borinsky:2018, Thurigen:2102}.
However, for renormalization one is not interested in any subgraphs but specifically in the subclass of Feynman diagrams which are superficially divergent and 1PI (i.e.~bridgeless)  \cite{Kreimer:1997dp, Connes:1998qv}.

In the quartic matrix theory, the superficial degree of divergence $\sdd$ of a dia\-gram $\rg$ of genus $g_\rg$ with $\nv^{(4)}_\rg$ four-valent-vertices and boundaries $j=1,...,b_\rg$  of lengths $\vec{n}=(n_1,...,n_{b_\rg})$ is given by\cite{Hock:2020rje}
\[
%\sdd(\rg) = \frac{D}{2} -\frac{D-2}{4} \sum_{j=1}^{b_\rg} n_j %\nv_{\partial \rg} +\frac{D-4}{2}\nv_\rg - \frac{D}{2} \left(2g_\rg + b_\rg -1 \right) 
2\sdd(\rg) = D -\frac{D-2}{2} \sum_{j=1}^{b_\rg} n_j %\nv_{\partial \rg} 
+(D-4)\nv^{(4)}_\rg - D \left(2g_\rg + b_\rg -1 \right) \, ,
\]
where we have included not just four-valent but also two-valent vertices, the latter with a vertex weight $\omega(\edgevertex)=1$ cancelling the dependence on their number $\nv^{(2)}_\rg$.
The four-valent vertices come with a weight $\omega(\vertex)=4-D$ such that they are relevant in $D>4$ and marginal in $D=4$ where
\[
2\sdd(\rg) = 4 - \sum_{j=1}^{b_\rg} n_j - 4(2g_\rg + b_\rg -1) .
\]
Thus, $\sdd$ can be non-negative only for ribbon graphs $\rg\in\RG{0,1}{n_1}$, i.e.~with $g_\rg = 0$ and $b_\rg=1$, for which
\[
\sdd(\rg) = \frac{4 - n_1}{2} 
\]
As a consequence, only $\vec{n}=(2)$ and $\vec{n}=(4)$ point functions need renormalization.
In particular, $\vec{n}=(2,2)$ point functions are convergent, including for example the diagram \eqref{eq:2-2-fish}.
Accordingly, the renormalization Hopf algebra has to be generated only by ribbon graphs of genus zero with a single boundary.

\begin{theorem}[Connes-Kreimer Hopf algebra of $D=4$ quartic MFT]
\label{theorem:Hopf algebra}
    \vtwo{Let ${}^\opi\RG{0,1}{n}(\vertex)$ be the set of connected, bridgeless ribbon graphs of genus zero and a single boundary with length $n$ and vertices of degree four.}
    Let
    \[
    \hfd=\langle {}^\opi\RG{0,1}{2}(\vertex)\cup {}^\opi\RG{0,1}{4}(\vertex) \rangle
    \]
    be the $\Q$-algebra freely generated by 
    those such ribbon graphs with $n=2$ or $n=4$. 
    With multiplication~$m$ given by disjoint union,
    this is a unital commutative algebra with unit $u:\Q\to\hfd, q\mapsto q\one$ where $\one$ is the empty ribbon graph.

    {Let $\ep$ be the projection deleting all bivalent vertices in a ribbon graph and define}
    \[
    \cop:\hfd \to \hfd\otimes\hfd , \quad
    \rg\mapsto \sum_{\substack{\sg\sgr\rg\\\sg\in\hfd}} \sg \otimes \ep(\rg/\sg) \, .
    \]
    This coproduct together with the counit
    \[
    \cou: \hfd \to \mathbb{Q} ,  \quad  \rg \mapsto \begin{cases}
    1 \textrm{ if }  \rg \in \Res^{*} \\
    0 \textrm{ else }
    \end{cases} 
    \]
    further defines on $\hfd$ the structure of a coassociative counital coalgebra,
    as well as a bialgebra.
    
    Finally, there is a unique inverse $\anti$ to the identity $\id:\rg\mapsto\rg$,
    \[
    \anti \conp \id = \id \conp \anti = u \circ \cou \, 
    \]
    with respect to the convolution product 
    \[
    \phi \conp \psi := m \circ (\phi\otimes\psi) \circ \cop
    \]
    of automorphisms $\phi,\psi:\hfd\to\hfd$. 
    This turns $\hfd$ into a Hopf algebra.

\begin{proof}
    As shown in \cite{Thurigen:2102}, ribbon graphs are the subset of 2-graphs, also called strand graphs, with two strands at each edge and whose vertex graphs are polygons. 
    Restricting to vertex degree %two and 
    four the vertex graphs are %either 2-gons or 
    4-gons \vtwo{(cycle graphs of length 4)}.
    Since $\hfd$ is generated by ribbon graphs with single boundary of length $n=2$ or $n=4$, contraction yields ribbon graphs with vertices which have as vertex graphs 2-gons or 4-gons.
    {However, the resulting bivalent vertices on the right in the coproduct~$\cop$ are deleted by the projection $\ep$.%
    \footnote{For the purpose of renormalization, e.g.~in the BPHZ momentum scheme \cite{Thurigen:2103}, it is in fact necessary to keep the bivalent vertices to account for the right number of propagators. The necessary map $\cop$ without projection $\ep$ is then a coaction on the left $\hfd$-comodule algebra which includes ribbon graphs with bivalent vertices.}
    }
    Furthermore, contraction does not change the boundary and preserves the genus of a ribbon graph.
    Together this means that $\hfd$ is contraction closed, i.e.~any ribbon graph resulting of the contraction of two ribbon graphs in $\hfd$ is element of $\hfd$ itself.
    Thus, according to Thm.~2 in \cite{Thurigen:2102} $\hfd$ is a Hopf algebra because it is a contraction-closed subalgebra of the Hopf algebra of 2-graphs.
\end{proof}

\end{theorem}

The relevance of the Connes-Kreimer Hopf algebra for renormalization lies in the fact that the coinverse $\anti$ induces a group structure on the set of algebra homomorphisms $\phi,\psi:\hfd \to \mathcal{A}$ to any unital commutative algebra $\mathcal{A}$, e.g.~an algebra of Feynman amplitudes and Green's functions. 
To this end one extends the convolution product to 
$\phi \conp \psi := m_\mathcal{A} \circ (\phi\otimes\psi) \circ \cop$.
It is then straightforward to check that $\anti^\phi:=\phi\circ S$ defines an inverse element to~$\phi$.
This inverse is the crucial part in the definition of counter terms subtracting the divergent part of a Feynman amplitude \cite{Kreimer:1997dp, Connes:1998qv}.

\begin{remark}
    The Connes-Kreimer Hopf algebra of Thm.~\ref{theorem:Hopf algebra} in essence is the same as the one of \cite{Tanasa:0707} for the Grosse-Wulkenhaar model. 
    There, ribbon graphs with $(g,b)=(0,1)$ are called planar ``regular'' while those with $b>1$ are coined planar ``irregular''. 
    Furthermore, the coproduct is defined with respect to subgraphs which are ``shrinkable'' with respect to the given set of residues, i.e.~their external structure must be of the 4-gon or 2-gon vertex type. 
    The crucial step in the proof is then to show for coassociativity the possibility to insert planar regular graphs into each other in such a way that again a planar regular graph results.
    This is the analogue to contraction completeness in our proof.
    However, they allow the possibility to insert graphs also in different ways but show that there exist the right (cyclic) insertion to prove coassociativity \cite{Tanasa:0707}.
    However, issues in the definition of Dyson-Schwinger equations are addressed in \cite{Tanasa:2009hb} by allowing to insert planar irregular graphs into single-trace vertices, see e.g. Fig.~9 therein. 
    
    Contrary, this is not possible with our mathematically rigorous definitions here and, as we show in the following, it is also not necessary for a consistent definition of combinatorial Dyson-Schwinger equations.
    With the proper definition of ribbon graphs, Def.~\ref{def:ribbon graph}, consistent with combinatorial maps (App.~\ref{sec:ribbon graphs}), residues are fixed in a unique way, Eq.~\eqref{eq:residue}, and always have to match the structure of vertices in contraction and insertion operations.  
\end{remark}

\subsection{Combinatorial Dyson-Schwinger equations} 

Dyson-Schwinger equations directly determine the Green's functions of a QFT.
Here we want to focus on the planar 2- and 4-point function or the $4D$ quartic matrix theory which need renormalization. 
Their corresponding combinatorial perturbative series 
\begin{align}
\xe &= \edge - \sum_{\substack{\rg\in{}^\opi\RG{0,1}{2}(\vertex)\\  \res(\rg)=\,\edgevertex}} 
\alpha^{\nf_\rg} {\rg} %{|\aut \rg|}
= \edge - \sum_{j=1}^\infty \alpha^j \ce_j 
\label{eq:seriesXe}\\
\xv &= \vertex + \sum_{
\substack{\rg\in{}^\opi\RG{0,1}{4}(\vertex)\\ \res(\rg)=\vertex}} 
\alpha^{\nf_\rg} {\rg} %{|\aut \rg|}
= \sum_{j=0}^\infty \alpha^j \cv_j.
\label{eq:seriesXv}
\end{align}
are element of the Connes-Kreimer Hopf algebra $\hfd$.
Therein, the parameter $\alpha$ counts the number of internal faces $\nf_\rg$ which is the analogue to the number of loops in point-like interacting field theories.
For the quartic theory this is a trivial redefinition of the quartic coupling $\alpha=\lambda$.
More generally, for a single $k$-valent interaction the Euler formula yields
\[
\nf = E-V+2-2g-b = \frac{k-2}{2}V -\left(\frac{1}{2}\sum_{j=1}^b n_j -1 \right) - (2g+b-1) 
\]
using the handshake relation $2E= k V - \sum_j n_j$.
For given boundary and genus this is an affine relation in the number of vertices $V$.
For $k=4$ one has simply $F=V+cons$, in particular with $(g,b)=(0,1)$ it is $F=V$ for $n_1=2$ and $F=V-1$ for $n_1=4$.

\

Combinatorial Dyson-Schwinger equations rely on the fact that divergent diagrams can be constructed iteratively by inserting divergent diagrams into each other.
In this way one may obtain all elements in the series $X^e, X^v$.
To this end, the starting point are primitive diagrams, i.e. those which do not contain any divergent subgraphs in $\hfd$:
\begin{definition}
    A ribbon graph $\rg\in \hfd$ is called primitive if 
    %$\sdd(\rg)\geq 0$ and 
    there exists no ribbon subgraph $\sg\subsetneq \rg$ in $\hfd\setminus\Res^*$. %with $\sdd(\sg)\geq 0$. 
\end{definition}

While there is only one kind of primitive 2-point diagrams in quartic MFT, there are infinitely many primitive 4-point diagrams:

\begin{lemma}[primitive 2-point diagrams]\label{lemma:primitives2}
The tadpole diagrams $\tadpoleup$ and $\tadpoledown$
%$\def\svgwidth{20pt}\input{graphs/tadpoleoben.pdf_tex},\def\svgwidth{20pt}\input{graphs/tadpoleunten.pdf_tex}$ 
are the only primitive superficially divergent ribbon graphs with $\vec{n}=(2)$.  
\begin{proof}
A 1PI 2-point graph can be %written in 
\vtwo{of} the form \def\svgwidth{20pt}%% Creator: Inkscape 1.0.2 (394de47547, 2021-03-26), www.inkscape.org
%% Accompanies image file '2pointgrob.pdf' (pdf, eps, ps)
\begingroup%
  \makeatletter%
  \providecommand\color[2][]{%
    \errmessage{(Inkscape) Color is used for the text in Inkscape, but the package 'color.sty' is not loaded}%
    \renewcommand\color[2][]{}%
  }%
  \providecommand\transparent[1]{%
    \errmessage{(Inkscape) Transparency is used (non-zero) for the text in Inkscape, but the package 'transparent.sty' is not loaded}%
    \renewcommand\transparent[1]{}%
  }%
  \providecommand\rotatebox[2]{#2}%
  \newcommand*\fsize{\dimexpr\f@size pt\relax}%
  \newcommand*\lineheight[1]{\fontsize{\fsize}{#1\fsize}\selectfont}%
  \ifx\svgwidth\undefined%
    \setlength{\unitlength}{82.50109976bp}%
    \ifx\svgscale\undefined%
      \relax%
    \else%
      \setlength{\unitlength}{\unitlength * \real{\svgscale}}%
    \fi%
  \else%
    \setlength{\unitlength}{\svgwidth}%
  \fi%
  \global\let\svgwidth\undefined%
  \global\let\svgscale\undefined%
  \makeatother%
  \begin{picture}(1,0.56825796)%
    \lineheight{1}%
    \setlength\tabcolsep{0pt}%
    \put(0,0){\includegraphics[width=\unitlength,page=1]{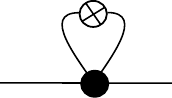}}%
  \end{picture}%
\endgroup%
 or \def\svgwidth{20pt}%% Creator: Inkscape 1.0.2 (394de47547, 2021-03-26), www.inkscape.org
%% Accompanies image file '2pointgrun.pdf' (pdf, eps, ps)
\begingroup%
  \makeatletter%
  \providecommand\color[2][]{%
    \errmessage{(Inkscape) Color is used for the text in Inkscape, but the package 'color.sty' is not loaded}%
    \renewcommand\color[2][]{}%
  }%
  \providecommand\transparent[1]{%
    \errmessage{(Inkscape) Transparency is used (non-zero) for the text in Inkscape, but the package 'transparent.sty' is not loaded}%
    \renewcommand\transparent[1]{}%
  }%
  \providecommand\rotatebox[2]{#2}%
  \newcommand*\fsize{\dimexpr\f@size pt\relax}%
  \newcommand*\lineheight[1]{\fontsize{\fsize}{#1\fsize}\selectfont}%
  \ifx\svgwidth\undefined%
    \setlength{\unitlength}{82.50067894bp}%
    \ifx\svgscale\undefined%
      \relax%
    \else%
      \setlength{\unitlength}{\unitlength * \real{\svgscale}}%
    \fi%
  \else%
    \setlength{\unitlength}{\svgwidth}%
  \fi%
  \global\let\svgwidth\undefined%
  \global\let\svgscale\undefined%
  \makeatother%
  \begin{picture}(1,0.56823689)%
    \lineheight{1}%
    \setlength\tabcolsep{0pt}%
    \put(0,0){\includegraphics[width=\unitlength,page=1]{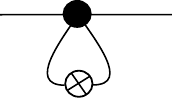}}%
  \end{picture}%
\endgroup%
. % which is not necessarily \vtwo{same}. 
The two vertices (2- and 4-valent vertex) correspond to a not necessarily connected 2-point subgraphs and a 4-point subgraph, respectively. A 2-point subgraph has superficial degree of divergence $\sdd=2$. It is either not primitive, than it has a primitive subgraph or it is primitive. The same hold for the 4-point subgraph where the superficial degree of divergence is $\sdd=0$. 
\end{proof}
\end{lemma}

\begin{lemma}[primitive 4-point diagrams]\label{lemma:primitives4}
There are infinitely many primitive divergent ribbon graphs with $\vec{n} = (4)$.
\begin{proof}
The following series is an infinite series of primitive 4-point graphs:
\[
\boxgraph \, , \quad 
\boxxgraph\, , ... \quad
\boxxxxgraph
\]
This series of 4-point graphs has no 1PI subgraph in the Connes-Kreimer Hopf algebra, i.e. there is no 2- or 4-point subgraph. Note that there are 1PI 6-point subgraphs which however do not belong to the Connes-Kreimer Hopf algebra since there superficial degree of divergence is $\sdd=-2$. To have 4-point subgraphs, the graph needs to be 3-cut irreducible, where all three cuts are not allowed to be at a vertex at a external leg. It is easy to see that the constructed class of 4-point graphs does not obey this property. 
\end{proof}
\end{lemma}

The main ingredient to the combinatorial DSE is the grafting operator $\graft$
which acts as a weighted insertion of one graph into another one.
To insert a ribbon graph $\sg$ into another ribbon graph $\cg$ the external structure of $\cg$ has to match the vertex structure of $\sg$.
We follow \cite{Thurigen:2102} generalized from \cite{Borinsky:2018} defining this vertex structure as the \emph{skeleton} of the ribbon graph, given by the projection deleting all internal edges
\[
\skl:\RG{}{}\to\Res^{*} 
\, , \, 
(\H,\sigma,\ei) \mapsto (\H,\sigma,\id) \, .
\]
Only when residue of $\sg$ and skeleton of $\cg$ agree one can choose a specific identification of them to insert $\sg$ into $\cg$:

\begin{definition}\label{def:insertion}
Let $\sg, \cg$ be ribbon graphs and $\ib:\res(\sg)\to\skl(\cg)$ an isomorphism on $\Res^{*}$,
\vtwo{i.e. $\ib:\H_{\res(\sg)}\rightarrow \H_{\cg}$ is bijective and $\sigma_{\res(\sg)} = \ib^{-1}\circ \sigma_\cg \circ\ib$.}
The insertion of $\sg$ into $\cg$ according to $\ib$ is 
\[
\cg\circ_\ib \sg := \left(\H_\sg, \sigma_\sg,\ei \right) \, 
\]
where
\[
\ei = \begin{cases}
    \ib^{-1} \circ \ei_{\cg}\circ \ib & \text{ on } \H_{\res(\sg)} \\
    \ei_\sg & \text{ otherwise}
\end{cases}
\]
and we denote the set of possible insertions, i.e.~isomorphisms $\ib$, $\mathcal{I}(\sg,\cg)$.
\end{definition}

The set of possible insertions factorizes into automorphisms of the structure of each vertex and permutations of vertices of the same kind. 
For ribbon graphs, vertices have the structure of polygons and their automorphisms are cyclic permutations:
\begin{lemma}\label{lem:factorization}
    Let $\sg, \cg$ be ribbon graphs with $\res(\sg)\cong\skl(\cg)$, i.e. it is possible to insert $\sg$ into $\cg$.
    Then
    \[
    \mathcal{I}(\sg,\cg)\cong \aut(\res(\sg)) \cong \aut(\skl(\cg)) =
    \vtwo{\prod_{k\ge 2} \aut(\V^{(k)}_\cg) 
    \times\prod_{v\in\V_\cg} \aut(v) }
    \]
    There are $k$ automorphisms for a $k$-valent vertex $v\in\V_\cg = \mathcal{C}(\sigma_\cg)$ which is a cycle of $\sigma_\cg$ with $k$ elements.
    There are $V^{(k)}_{\cg}!$ permutations of the $V^{(k)}_{\cg} = |\V^{(k)}_\cg|$ vertices in $\cg$ of degree $k$.
    In particular, 
    \[
    |\mathcal{I}(\sg,\cg)|
    %= |\aut(\res(\sg))| 
    %= |\aut(\skl(\cg))|
    = \prod_{k\ge 2} |\aut(\V^{(k)}_\cg)| \prod_{v\in\V_\cg} |\aut(v)|
    %= \prod_{k\ge 2} V^{(k)}_{\cg}! \prod_{v\in\V_\cg} k_v 
    = \prod_{k\ge 2} V^{(k)}_{\cg}! \, k^{V^{(k)}_{\cg}} \, .
\]
\begin{proof}
    The set of isomorphisms between $\res(\sg)$ and $\res(\cg)$ is isomorphic to the set of automorphisms on either of the two exactly because they %these $\res(\sg)$ and $\res(\cg)$ 
    are isomorphic (otherwise $\mathcal{I}(\sg,\cg)=\emptyset$).
    Since $\res(\cg)\in\Res^{*}$ is a union of graphs with single vertices, its automorphism group factorizes accordingly.
\end{proof}
\end{lemma}

The insertion of ribbon graphs $\sg$ with edge-structure, that is $\res(\sg)=\edgevertex$, %$k=2$ bivalent vertex residues, 
needs a special treatment.
Our definition of insertions demands a matching $\res(\sg) \cong \skl(\cg)$ and thus bivalent vertices in $\cg$ to allow for edge-diagram insertions. 
Such bivalent vertices are not present in the graphs of the series $X^e, X^v$.
Still, we have to include the possibility to insert 2-point graphs to obtain all relevant graphs by insertion.
To cover this, let us define the projection 
\[
\ep: \RG{}{} \to \RG{}{}%{k>2}{}
\]
which contracts (in the usual graph-theoretic sense) all bivalent vertices, $\edgevertex \mapsto \edge$. 
This yields an equivalence relation $\sim_2$ identifying all diagrams $\rg,\rg'$ for which $\ep(\rg)=\ep(\rg')$ and accordingly an equivalence class $[\rg]_{\sim_2} = [\rg']_{\sim_2}$.
Using this, we can define the properly weighted insertion operator which allows arbitrary many edge insertions on each edge:

\begin{definition}\label{def:graftingOP}
The grafting operator is the map 
\[
\graft: \hfd \times \hfd \to \hfd \, , \, (\sg,\Gamma) \mapsto \graft^{\Gamma}(\sg) \nonumber
\]
linear in both arguments 
and defined by  
\[\label{eq:grafting definition}
\graft^{\Gamma}(\sg) := \sum_{\cg\in [\Gamma]_{\sim_2}} \frac1{|\mathcal{I}(\sg,\cg)|} \sum_{\ib\in \mathcal{I}(\sg,\cg)} \frac{\cg\circ_\ib \sg}{\maxf(\cg\circ_\ib \sg)} \, ,
\]
wherein 
$\maxf(\rg)$ is the number of subgraphs $\sg\subsetneq\rg$ in $\hfd$
such that $\rg/\sg$ is primitive (corresponding to ``maximal $\rg$-forests'' in the language of Zimmermann's forest formula \cite{Zimmermann:1969up}). 
\end{definition}

Note that only  subgraphs $\sg\in \hfd$ are counted to obtain $\maxf$.
For example, the ribbon graph of Eq.~\eqref{eq:2-2-fish} is a ribbon subgraph $\fishtwob \sgr \sunrise$ but having two boundaries it is not in $\hfd$.
For this reason%
\footnote{In \cite{Tanasa:2009hb} this example is also given after Def.~4.3 therein. However, based on the Hopf algebra construction of \cite{Tanasa:0707}, the contraction of $\fishtwob \sgr \sunrise$ is allowed. The authors of \cite{Tanasa:2009hb} argue that the reason to exclude this case from $\maxf(\sunrise)$ is only that the resulting contracted graph would be non-planar (which is actually wrong, the point is rather that contraction of $\fishtwob$ leads to a multi-trace vertex, Eq.~\eqref{eq:multitrace residue}).

In the end, this wrong conception of allowed ribbon subgraphs and $\maxf$ in \cite{Tanasa:2009hb} seems to be the reason why the Dyson-Schwinger equations do not work out without including ``planar irregular" graphs such as $\fishtwob$ in the Hopf algebra.
Specifically, in Eq.~(6.7) of \cite{Tanasa:2009hb} there is the equation %the authors calculate 
$\maxf(\tadpoleud)=2$ though the fish-type ribbon subgraph is not $\fishup$ but $\fishtwobb\sgr\tadpoleud$ which should not be in $\hfd$ and thus the correct result is $\maxf(\tadpoleud)=1$.
As a consequence, \cite{Tanasa:2009hb} comes to the conclusion that $\fishtwobb$ and, more generally, all planar irregular ($g=0, b>1$) ribbon graphs should be included in the Hopf algebra.
However, this does not only lead to a Hopf algebra which is not the one renormalizing quartic MFT; it is also inconsistent because this extended Hopf algebra is not contraction closed. That is, contraction of planar irregular ribbon graphs, $b>1$, leads to $b$-fold multi-trace vertices which are still not part of that extended Hopf algebra.
}
$\maxf(\sunrise)=2$ in MFT in contrast to local scalar QFT where the sunrise diagram has three maximal forests. 

The grafting operator has an expansion in the ribbon graphs which result from the insertions summed over in the definition.
Thereby, some insertions turn out to be redundant and can be divided out. 
We follow the literature \cite{Kreimer:2005rw} for the notation of the resulting combinatorial factors:
\begin{proposition}\label{prop:grafting expansion}  
%    Let $\gamma$ be a primitive graph. 
\vtwo{Let $\sg, \Gamma\in\hfd$.}
The grafting operator $\graft$ has an expansion
\[
    \graft^{\Gamma}(\sg) = \sum_{\cg\in [\Gamma]_{\sim_2}} \frac{1}{(\cg|\sg) |\sg|_\vee }\sum_{ \rg\in\hfd %\substack{\rg\in\hfd \\ \res(\rg)=\res(\cg)}
    }\frac{\bij(\cg,\sg,\rg)}{\maxf(\rg)} \rg,
\]
where the weight factors are defined as follows \cite{Kreimer:2005rw}:
\begin{itemize}
    \item $(\cg|\sg)$ is the number of insertion places for $\sg$ in $\cg$.
    \item $|\sg|_\vee $ is the number of distinct graphs $\sg$ which are equal upon removal of the external edges.
    %obtained via a cyclic permutations of the external faces 
    \item  $\bij(\cg,\sg,\rg)$ is the number of \emph{``non-equivalent''} bijections between $\res(\sg)$ and $\skl(\cg)$ such that $\rg$ is obtained by insertion.
\end{itemize}
\begin{proof}
    Let $\sg, \cg$ be ribbon graphs with $\res(\sg)\cong\skl(\cg)$.
    One can partition the set of insertions according to the resulting ribbon graphs $\rg$,
    \[
    \mathcal{I}(\sg,\cg) = \bigoplus_{\rg\in\RG{}{}} \, \mathcal{I}_\rg(\sg,\cg)
    \, , \quad
    \mathcal{I}_\rg(\sg,\cg) = \{\ib\in \mathcal{I}(\sg,\cg) \vert \cg \circ_\ib \sg = \rg \} \, ,
    \]
    where, as always here, each $\rg$ is canonically labelled as described in the end of Sec.~\ref{sec:ribbon graphs} (otherwise one would have to identify graphs which are identical upon relabelling).
    %By Def.~\ref{def:insertion}, $\res(\rg)=\res(\cg)$.
    Now let $\sg, \cg\in\hfd$. Then, for any $\ib\in \mathcal{I}(\sg,\cg)$ also $\rg = \cg \circ_\ib \sg\in\hfd$ such that
    \[\label{eq:graft expansion 1}
    \graft^{\Gamma}(\sg) = \sum_{\cg\in [\Gamma]_{\sim_2}} \frac{1}{|\mathcal{I}(\sg,\cg)|}\sum_{\rg\in\hfd} \frac{|\mathcal{I}_\rg(\sg,\cg)|}{\maxf(\rg)} \rg \, .
    \]

    There are some parts in the factorization of $\mathcal{I}(\sg,\cg)$ according to Lemma \ref{lem:factorization} which are not sensitive to the insertion outcome $\rg=\cg\circ_\ib \sg$. 
    One type is permutations of single-vertex components of $\sg = \prod_i\sg_i$. If there are $j_k$ vertex components $\sg_1=...=\sg_{j_k}\in\Res$ of the same type $k$ (here degree $k=2$ or $k=4$), then the resulting graph $\rg$ does not depend on the $j_k!(\nv_\cg^{(k)}-j_k)!$ permutations how they are inserted.
    Thus one can factor them also in $|\mathcal{I}_\rg(\sg,\cg)|$ and define the number of relevant permutations in the denominator surviving after cancellation with the denominator in Eq.~\eqref{eq:graft expansion 1} as 
    \[
    (\cg|\sg) = \prod_{k\ge 2} \frac{|\aut(\V^{(k)}_\cg)|}{j_k !(\nv_\cg^{(k)}-j_k)!} = \binom{\nv^{(k)}_\cg}{j_k}
    \]
    For $k=2$, this definitions in fact differs from the literature, e.g. \cite{Kreimer:2005rw}, since the definitions here allow edges insertions only in bivalent vertices which are summed in the overall sum over $\cg\in\Gamma$ in $\graft^\Gamma$.
    
    Another type cancellations is redundant automorphisms in $\aut(v)$ for a given vertex $v\in\V_\cg$:
    If a component $\sg_i$ has automorphisms in $\aut(\res(\sg_i))$ which are equivalent upon insertion, these yield always the same graph $\rg$ and can thus be factored out in $|\mathcal{I}_\rg(\sg,\cg)|$.
    Again, for a single-vertex component $\sg_i=\res(H_i)\in\Res$ of degree $k$ all $k$ cyclic permutations give the same result.
    But also $\sg_i\notin\Res$ may have such symmetries (e.g.~the $k=2$ cyclic permutations of $\sunrise$ are equivalent, or two of the $k=4$ cyclic permutations of $\fishright$). 
    These symmetries are the automorphisms of the ribbon graphs themselves since there are no internal automorphisms.
    Thus we define for $\sg = \prod_i\sg_i$ the relevant factor
    \[
    |\sg|_\vee = \prod_i|\sg_i|_\vee = \frac{|\aut(\res(\sg))|}{|\aut(\sg)|} \, .
    \]
    The number of ``non-equivalent'' bijections $\bij(\cg,\sg,\rg)$ is then exactly the number of the subset of all insertions of $\sg$ into $\cg$ which yield $\rg$ with both types of trivial factors described factored out
    \[
    \bij(\cg,\sg,\rg) = \frac{|\mathcal{I}_\rg(\sg,\cg)|}{|\aut(H)|\prod_{k\ge2} j_k !(\nv_\cg^{(k)}-j_k)!}
    \]
    Then the statment follows by cancelling these factors in Eq.~\eqref{eq:graft expansion 1}.
\end{proof}
\end{proposition}

\begin{example}
    Let $\cg = \fishright$ and $\sg = \vertex \, \fishright$.
    Then $(\cg|\sg)=\binom{2}{1} = 2$, that is inserting $\fishright$ on the left ($l$) or the right ($r$) vertex of $\cg$. 
    Furthermore, $|\sg|_\vee = |\vertex|_\vee\cdot|\fishright|_\vee = 1\cdot2$, i.e.~inserting $\fishright$ either horizontally ($h$) as it is or rotated 90 degrees vertically ($v$).
    Thus,
    \begin{align}
    \graft^{\fishright}(\vertex \, \fishright) 
    &= 
    \frac{1}{2\cdot2}\sum_{i=l,r}\sum_{j=h,v} \begin{pmatrix}
        \frac{1}{2}\fishrr & \frac{1}{2}\fishrr \\
        \fishlu & \fishru
    \end{pmatrix}_{ij} \\
    &= \frac{1}{4}\left(\fishrr + \fishru + \fishlu \right) \, 
    \end{align}
    since $\maxf(\fishrr)=2$.
    Analogously, %(from now on dropping ${\sim_2}$ in the second argument),
    \[
    \graft^{\fishup}(\vertex \, \fishright)
    = \frac{1}{4}\left(\fishuu + \fishur + \fishdr \right)
    \]
    Since $\graft^\Gamma(\fishright)=\graft^\Gamma(\fishup)$ in general, we can combine this for $c^v_1=\fishright+\fishup$ to
    \[
    \graft^{c^v_1}(c^v_0 c^v_1) = \frac{1}{2}\left(\fishrr +\fishuu + \fishru + \fishur + \fishlu + \fishdr \right)
    = \frac{1}{2} c^v_2
    \]
\end{example}

\begin{lemma}\label{lem:number of vertices/edges}
Let $\rg\in\hfd$ with $k$ faces. Then the number of internal edges $\nei_\rg$ and vertices $\nv_\rg$ are     
    \begin{align*}
       \nei_\rg=
\begin{cases}
2k-1,& \text{for }\res(\rg)=\edgevertex\\
2k, & \text{for }\res(\rg)=\vertex
\end{cases}\, ,\qquad  \nv_\rg=
\begin{cases}
k,& \text{for }\res(\rg)=\edgevertex\\
k+1,& \text{for }\res(\rg)=\vertex
\end{cases}.
    \end{align*}
    \begin{proof}
        Use Euler's formula $\chi=2-2g-n=V-E+F$. 
        For $\rg\in\hfd$ we have $\chi_\rg=2-1=1$.
        The result follows since $E_\rg=\nei_\rg+4 =\frac{4\nv_\rg+4}{2}$ for $\res(\rg)=\vertex$, and $E_\rg=\nei_\rg+2 =\frac{4\nv_\rg+2}{2}$ for $\res(\rg)=\edgevertex$.
    \end{proof}
\end{lemma}

Combining the definitions of primitive graphs and the grafting operator yields \emph{the combinatorial DSE} in the renormalization Hopf algebra:

\begin{theorem}\label{thm:DSE} 
Let 
\vtwo{$\xe$ and $\xv$ be the formal power series in ribbon graphs with parameter $\alpha$ counting their faces, defined in \eqref{eq:seriesXe} and \eqref{eq:seriesXv} respectively, and} 
$Q=\frac{\xv}{(\xe)^2}$, where the inverse is symbolically understood as geometric series $\frac{1}{\xe}=1+(\edge-\xe)+(\edge-\xe)^2+...$ 
    Then, the system of combinatorial DSEs is
\begin{align}
    \xe\equiv X^{\edge}(\alpha) &= \edge %\mathbbm{1} 
    - \alpha \graft^{\edge}(Q\xe)
= \edge %\mathbbm{1} 
- \alpha (\graft^{\tadpoleup}+\graft^{\tadpoledown}) (Q\xe)
\label{2PcDSE}\\
% \end{align}
% and
% \begin{align}
\xv\equiv X^{\vertex}(\alpha) &= \vertex %\mathbbm{1} 
+ \!\! \sum_{\substack{\Gamma \text{ primitive} \\ \res(\Gamma)=\vertex}} \! \alpha^{\nf_\Gamma} \graft^{\Gamma}(Q^{\nf_\Gamma} \xv)
=\vertex %\mathbbm{1} 
+ \alpha (\graft^{\fishright} + \graft^{\fishup})(Q\xe) + ...
\nonumber %\label{4PcDSE}
\end{align}
where the set of primitive 4-point graphs is infinite.
\begin{proof}
    The proof follows from the definitions and the Lemmata above. 
    %From Lemma \ref{lem:primsub}, we know that e
    Each graph in $\xe$ and $\xv$ is either a primitive or can be constructed from insertions into primitive graphs. 
    The primitive graphs for the 2-point function are proven in Lemma \ref{lemma:primitives2}.
    Furthermore, it is shown in Lemma \ref{lemma:primitives4} that the set of primitive 4-point graphs is infinite.

    The weight $1/|\mathcal{I}(\sg,\cg)|$ in the definition of $\graft^\Gamma$, Def.~\ref{def:graftingOP}, exactly cancels multiplicities arising from the argument $Q^{\nf_\Gamma}X^\bullet$.
    According to Lemma~\ref{lem:number of vertices/edges} this argument has a factor $\xv$ for each 4-valent vertex in $\Gamma$ and a factor $1/\xe$ for each edge in $\cg\in\Gamma$ without bivalent vertices.
    Since these factors multiply commutatively, each monomial occurs with a multiplicity $\nv^{(2)}_\cg! \, \nv^{(4)}_\cg!$.
    Furthermore, each individual connected component $\sg$ occurs in $\xv$ and $\xe$ in $|\sg|_\vee$ rotated ways which however yield the same graph under $\graft^\Gamma$. 
    Together with Lemma~\ref{lem:factorization} and Prop.~\ref{prop:grafting expansion} this shows cancellation of multiplicities from the argument $Q^{\nf_\Gamma}X^\bullet$ with $|\mathcal{I}(\sg,\cg)|$ in the denominator of $\graft^\Gamma$.
    
    Additional multiplicities arise if a ribbon graph $\rg$ results from different insertions,  modulo the symmetries just discussed; since only insertions into primitive diagrams are allowed, these are exactly the number of subgraphs $\sg\sgr\rg$ with primitive cograph $\rg/\sg$ and this is cancelled by $1/\maxf$ in the definition of~$\graft^\Gamma$.    
\end{proof}
\end{theorem}

\subsection{Subalgebra of loop orders}
An important property in the renormalization Hopf algebra is the subalgebra structure. For the 4D MFT model, the subalgebra will be studied in details which takes an important part in revealing the complexity of renormalizing this model. The following Proposition provides the subalgebra structure on the perturbative series.
\begin{proposition}\label{prop:coprodX}
Let $\bullet\in \{\vtwo{e\equiv\edge, v\equiv\vertex} \}$. Then, the following holds
    \begin{align}\label{coproductX}
        \cop(\xr)=\sum_{k\ge0} \xr \bigg(\frac{\xv}{(\xe)^{2}}\bigg)^k\otimes c^\bullet_k
        %,\qquad \cop(\xv)=\sum_{k} \xv \bigg(\frac{\xv}{(\xe)^{2}}\bigg)^k\otimes c^\vertex_k.
    \end{align}
    \begin{proof}
        The Connes-Kreimer Hopf algebra considered here is a special case of the Hopf algebra of 2-graphs in \cite{Thurigen:2102}. %The main result there is 
        This implies the so-called central identity \cite[Thm. 1]{Thurigen:2102} which translates into our setting by the following. Note that we consider the $D=4$ MFT with 4-valent vertices. The renomalization Hopf algebra is 
        $
    \hfd=\langle {}^\opi\RG{0,1}{2}(\vertex)\cup {}^\opi\RG{0,1}{4}(\vertex) \rangle.
    $
    Then the central identity reads
    \begin{align}\label{coproductcentralidentity}
        \cop(\xr)=\sum_{\res(\rg)=\bullet}\frac{(\xv)^{\nv_\rg}}{(\xe)^{\nei_\rg}}\otimes \rg
    \end{align}
    where $\bullet\in\{\edge,\vertex\}\equiv\{e,v\}$ and $\nv_\rg,\nei_\rg$ %:\RG{}{}\to \mathbb{N}$ 
    is the number of vertices or internal edges of $\rg$, respectively. 
    {These are expressed in the loop order $k$ according to Lemma~\ref{lem:number of vertices/edges}.}
The central identity \eqref{coproductcentralidentity} implies further that at each loop order $k$ and for any graph in the perturbative series $\xr$ the same factor appears on the left of the coproduct. 
Upon factorization
this yields $\sum_{\res(\rg)=\bullet,\nf_\rg=k}\rg=c_k^\bullet$ on the right. This derives the desired coproduct formulas.
    \end{proof}
\end{proposition}

\begin{example}
Consider the coproduct of two-loop 2-point graphs. To save space we write out only the reduced part $\rcop(\rg)=\cop(\rg)-\rg\otimes\res(\rg)-\skl(\rg)\otimes\rg$:
\[
\rcop
\begingroup
\renewcommand*{\arraystretch}{2}
\begin{pmatrix}
\sunrise \\ \tadpoleuu \\ \tadpoleud \\ \tadpoledu \\ \tadpoledd
\end{pmatrix}
\endgroup
=
\begin{pmatrix}
\begin{matrix} \\\fishright\end{matrix} & 0 & \begin{matrix}\fishright\\ \end{matrix} & 0 \\
\fishup & \begin{matrix}\tadpoleup\\\vertex\end{matrix} & 0 & 0 \\
0 & \begin{matrix}\tadpoledown\\\vertex\end{matrix} & 0 & 0 \\
0 & 0 & 0 & \begin{matrix}\vertex\\\tadpoleup\end{matrix} \\
0 & 0 & \fishup & \begin{matrix}\vertex\\\tadpoledown\end{matrix}
\end{pmatrix}
\cdot_\otimes
\begingroup
\renewcommand*{\arraystretch}{1.5}
\begin{pmatrix} 
\tadpoleumv \\  \tadpoleume \\  \tadpoledmv \\ \tadpoledme
\end{pmatrix}
\endgroup \, ,
\]
Here we use a matrix notation with product $\cdot_\otimes$ referring to a matrix product in which the coefficients multiply as tensor product.
Furthermore we have distinguished the ribbon graphs on the right-hand side of the coproduct by indicating the new vertices resulting form non-trivial contraction by white vertices. 
Summing the rows, we find
\[
\rcop(\ce_2) = \cv_1 \otimes \ce_{1} + \cv_0 \ce_1 \otimes \ce_1 
= (\cv_1 + \cv_0 \ce_1)\otimes \ce_1 
= \pe_{2,1} \otimes \ce_1 
\]
with a degree-2 polynomial $\pe_{2,1} = \cv_1 + \cv_0 \ce_1$ on the left of the coproduct.
\end{example}

For fixed loop order, the left of the coproduct is always a polynomial in the loop coefficients $\ce_n$ or $\cv_n$ which we can give explicitly:
\begin{theorem}\label{thm:coproduct}
The coproduct of $c_{n}^{\bullet}$ for any $n\in\mathbb{N}$ and $\bullet\in \{\edge,\vertex\}\equiv\{e,v\}$ is given by
\begin{align}\label{coproductsc}
    \cop (c_n^{\bullet})=\sum_{k=0}^n P^\bullet_{n,k}\otimes c_{n-k}^{\bullet},
\end{align}
where $P^{\bullet}_{n,k}$ are polynomials of loop order $k$ in the generators $\cv_i, \ce_j$ of the form:\\
For $0<k<n$ we have %for $\bullet=\vertex\equiv v$ 
\begin{align}
    \pv_{n,k} &=\sum_{j=0}^{k} \qv_{n-k+1,j} \, \qe_{2(n-k),k-j} \\ 
% \end{align}
% and for $\bullet=\edge \equiv e$ 
% \begin{align}
    \pe_{n,k} &=\sum_{j=0}^{k} \qv_{n-k,j} \, \qe_{2(n-k)-1,k-j} \, ,
\end{align}
where 
\begin{align*}
    \qv_{k,j}:=& \sum_{l_1+...+l_k=j} \cv_{l_1}...\cv_{l_k}\\
    \qe_{k,j}:=&\sum_{1s_1+2s_2+...+js_j=j}\frac{(s_1+...+s_j+k-1)!}{s_1!...s_j!(k-1)!}\prod_{i=1}^j(\ce_i)^{s_i}.
\end{align*}
The trivial cases are $P^{\bullet}_{n,n}(c)=c_n^{\bullet}$ and $P^{\bullet}_{n,0}(c)=\skl(c_n^{\bullet})$. 
\begin{proof}
    Take the coproduct formula of Proposition \ref{prop:coprodX} and restrict to a fixed $n$ loop order. Thus, we directly get the form of \eqref{coproductsc} with $P^\bullet_{n,k}$ the $k$'th loop order. 

    We define the following restriction to the $k$-loop order of an arbitrary perturbative series $Y$ by $(Y)_k$. Then, from \eqref{coproductX} we conclude
    \begin{align}\label{cne}
        \cop (\ce_n) &= \sum_{k=0}^n\underbrace{\bigg(\frac{(\xv)^{n-k}}{(\xe)^{2(n-k)-1}}\bigg)_k}_{\pe_{n,k}} \otimes \, \ce_{n-k},\\\label{cnv}
        \cop (\cv_n) &= \sum_{k=0}^n\underbrace{\bigg(\frac{(\xv)^{n-k+1}}{(\xe)^{2(n-k)}}\bigg)_k}_{\pv_{n,k}} \otimes \, \cv_{n-k}.
    \end{align}
    This amounts to compute the following coefficients
    \begin{align*}
        \qv_{k,j}:=\big((\xv)^k\big)_j,\qquad \qe_{k,j}:=\bigg(\frac{1}{(\xe)^k}\bigg)_j,
    \end{align*}
    where again, the index $j$ restricts to the $j$th loop order. The coefficient $\qv_{k,j}$ is straightforwardly given by \cite{Bergbauer:2005fb} 
    \begin{align*}
        \qv_{k,j}=\sum_{l_1+...+l_k=j} \cv_{l_1} ... \cv_{l_k}.
    \end{align*}
    The coefficient $\qe_{k,j}$ is slightly more involved but generalizing \cite{Broadhurst:2000dq} we can still derive it explicitly via Faa di Bruno's formula:
    \begin{align*}
        \qe_{k,j}=&[\alpha^j]\frac{1}{(\xe)^k}=\frac{1}{j!}\frac{\partial^j}{\partial \alpha^j}\frac{1}{(\xe)^k}\\
        =&\frac{1}{j!}\sum_{1s_1+2s_2+...+js_j=j}\frac{j!}{s_1!...s_j!}\frac{\partial^{s_1+...+s_j}}{\partial x^{s_1+...+s_j}}\frac{1}{(1-x)^k}\bigg\vert_{x=0}\prod_{i=1}^j(\frac{1}{i!}\frac{\partial^i \xe}{\partial \alpha^i})^{s_i}\\
        =&\sum_{1s_1+2s_2+...+js_j=j}\frac{(s_1+...+s_j+k-1)!}{s_1!...s_j!(k-1)!}\prod_{i=1}^j(\ce)^{s_i}.
    \end{align*}
    Inserting the explicit formula of $\qv_{k,j}$ and $\qe_{k,j}$ into \eqref{cne} and \eqref{cnv} finishes the proof.
\end{proof}

\end{theorem}

\begin{example}
Let us consider $n=4$ loop coefficient of the vertex series $\xv$.
We get for $k=3$
\begin{align*}
    \pv_{4,3}=&\sum_{\substack{l_1+l_{2}=3\\ l_i\geq 0}}\cv_{l_1}\cv_{l_2} +\sum_{\substack{l_1+l_{2}=2\\ l_i\geq 0}}\cv_{l_1}\cv_{l_2}\sum_{1s_1=1} (\ce_1)^{s_1}\frac{(s_1+1)!}{s_1!1!}\\
    &+\sum_{\substack{l_1+l_{2}=1\\ l_i\geq 0}}\cv_{l_1}\cv_{l_2}\sum_{1s_1+2s_2= 2} (\ce_1)^{s_1}(\ce_2)^{s_2}\frac{(s_1+s_2+1)!}{s_1!s_2!1!}\\
    &+\sum_{\substack{l_1+l_{2}=0\\ l_i\geq 0}}\cv_{l_1}\cv_{l_2}\sum_{1s_1+2s_2+3s_3=3} (\ce_1)^{s_1}(\ce_2)^{s_2}(\ce_3)^{s_3}\frac{(s_1+s_2+s_3+1)!}{s_1!s_2!s_3!1!}\\
    =& \, 2\cv_0\cv_3 +2\cv_1\cv_2 +(2\cv_0\cv_2+(\cv_1)^2)2\ce_1 +2\cv_0 \cv_1(2 \ce_2+ 3(\ce_1)^2) \\
    &+ (\cv_0)^2\left(2\ce_3+6\ce_1\ce_2+4(\ce_1)^3\right).
\end{align*}

For $k=2$, we have
\begin{align*}
    \pv_{4,2}=&\sum_{\substack{l_1+l_{2}+l_3=2\\ l_i\geq 0}}\cv_{l_1}\cv_{l_2}\cv_{l_3} +\sum_{\substack{l_1+l_{2}+l_3=1\\ l_i\geq 0}}\cv_{l_1}\cv_{l_2}\cv_{l_3}\sum_{1s_1= 1} (\ce_1)^{s_1}\frac{(s_1+3)!}{s_1!3!}\\
    &+\sum_{\substack{l_1+l_{2}+l_3=0\\ l_i\geq 0}}\cv_{l_1}\cv_{l_2}\cv_{l_3}\sum_{1s_1+2s_2= 2} (\ce_1)^{s_1}(\ce_2)^{s_2}\frac{(s_1+s_2+3)!}{s_1!s_2!3!}\\
    =&3(\cv_0)^2\cv_2 +3\cv_0(\cv_1)^2+3(\cv_0)^2\cv_1 \cdot 4\ce_1 + (\cv_0)^2\left(4\ce_2+10(\ce_1)^2\right).
\end{align*}
For $k=1$, we find
\begin{align*}
    \pv_{4,1}=&\sum_{\substack{l_1+l_{2}+l_3+l_4=1\\ l_i\geq 0}}\cv_{l_1}\cv_{l_2}\cv_{l_3}\cv_{l_4} +\sum_{\substack{l_1+l_{2}+l_3+l_4=0\\ l_i\geq 0}}\cv_{l_1}\cv_{l_2}\cv_{l_3}\cv_{l_4}\sum_{1s_1= 1} (\ce_1)^{s_1}\frac{(s_1+5)!}{5!}\\
    =&4(\cv_0)^3\cv_1+6(\cv_0)^4 \ce_1.
\end{align*}
This gives the polynomials in the coproduct of $\cv_4$,
\begin{align*}
    \cop (\cv_4)=\sum_{k=0}^4 \pv_{4,k}(c)\otimes \cv_{4-k}.
\end{align*}
\end{example}

\begin{proposition}\label{prop:coproductmonomial}
    Let us denote by $(Y)_k$ the $k$th order of some perturbative series $Y$. Let $f(\xv,\xe)$ be some monomial of $\xv$ and $\frac{1}{\xe}$, then the more general Hopf subalgebra structure holds
    \begin{align*}
        \cop\left(f(\xv,\xe)\right)=\sum_kf(\xv,\xe)\left(\frac{\xv}{(\xe)^2}\right)^k\otimes f(\xv,\xe)_k \, .
    \end{align*}
    %where $f(\xv,\xe)_k=\mathbbm{1}$.
    \begin{proof}
        By assumption $f(\xv,\xe)$ has the form $\frac{(\xv)^{n_1}}{(\xe)^{n_2}}$ for some $n_i\in \mathbb{Z}_{\geq 0}$. Note that the trivial zero'th order is given by one, or more precisely $\left(\frac{(\xv)^{n_1}}{(\xe)^{n_2}}\right)_0 = \left(\vertex\right)^{n_1}$. % =\mathbbm{1}$.

        The proof essentially follows from Proposition \ref{prop:coprodX} and the fact that a tensor space behaves under a product in the algebra as 
        \begin{align*}
            (a\otimes b)(c\otimes d)=(ac\otimes bd).
        \end{align*}
        This means that the coproduct of a product of graphs $\rg=\rg_1\rg_2$ is
        \begin{align*}
            &\cop(\rg)=\cop(\rg_1\rg_2)=\cop(\rg_1)\cop(\rg_2)=\sum_{k_1}(a_{k_1}\otimes \rg_1/a_{k_1})\sum_{k_2}(b_{k_2}\otimes \rg_2/b_{k_2})\\
            =&\sum_{k_1,k_2}a_{k_1}b_{k_2}\otimes \rg_1/a_{k_1}\rg_2/b_{k_2}.
        \end{align*}
        Extrapolating this to any $m$-folded product of graphs $\rg=\rg_1...\rg_m$ together with Proposition \ref{prop:coprodX} proves the assertion for any monomial of the form $f(\xv,\xe)=\frac{(\xv)^{n_1}}{(\xe)^{n_2}}$, where $\frac{1}{\xe}$ is understood as the geometric series 
        $\frac{1}{\xe}=1+(\edge-\xe)+(\edge-\xe)^2+...$ . 
        Performing the computation on the level of coefficients can be very ugly, but it is for instance performed for some specific Hopf algebras in \cite{Bergbauer:2005fb}.
    \end{proof}
\end{proposition}

One specific example of this proposition is of primary interest, which the coproduct of the argument of the grafting operator in the combinatorial DSE. This reads
\begin{lemma}\label{lem:coprodPnk}
Let $\bullet\in \{\vertex,\edge\}$.
    The coproduct of $\xr \left(\frac{\xv}{(\xe)^2}\right)^n$ reads
    \begin{align*}
        \cop\bigg(\xr \left(\frac{\xv}{(\xe)^2}\right)^n\bigg)=\sum_k\xr \left(\frac{\xv}{(\xe)^2}\right)^{n+k}\otimes \bigg[\xr \left(\frac{\xv}{(\xe)^2}\right)^n\bigg]_k
    \end{align*}
    \begin{proof}
        Insert $\xr \left(\frac{\xv}{(\xe)^2}\right)^n$ into Proposition \ref{prop:coproductmonomial} and the left hand side of the coproduct summarises as $\xr\left(\frac{\xv}{(\xe)^2}\right)^n \left(\frac{\xv}{(\xe)^2}\right)^k=\xr\left(\frac{\xv}{(\xe)^2}\right)^{n+k}$.
    \end{proof}
\end{lemma}

\subsection{Hochschild 1-cocycle}
The Hochschild 1-cocycle property is a direct consequence of Theorem \ref{thm:coproduct} by acting with the coproduct on the Dyson-Schwinger equation
\begin{theorem}\label{Thm:Hochschild}
The grafting operator is \vtwo{a} Hochschild 1-cocycle, i.e.
\begin{align}\label{Hochschild}
    \cop\graft^\vtwo{\bullet}=\graft^\vtwo{\bullet}\otimes\bullet + (\id\otimes\graft^\vtwo{\bullet})\cop
\end{align}
where $\displaystyle \graft^\vtwo{\bullet}=\sum_{\substack{\Gamma\, primitive\\ \res(\Gamma)=\bullet}}\graft^\Gamma$ with fixed $\bullet\in \{\edge,\vertex\}$ and $\graft^\Gamma$ acts on $Q^{F_\Gamma}\xr$ with $F_\Gamma$ the number of faces of the primitive graph $\Gamma$.
\begin{proof}
From the combinatorial DSE \eqref{2PcDSE} %and \eqref{4PcDSE} 
we have 
    \begin{align*}
        \xr= \bullet \pm \sum_{\substack{\Gamma\, primitive\\ \res(\Gamma)=\bullet}}B_+^\Gamma\bigg(\xr \left(\frac{\xv}{(\xe)^2}\right)^{F_\Gamma}\bigg).
    \end{align*}
    Let now the identity \eqref{Hochschild} act on $\xr \left(\frac{\xv}{(\xe)^2}\right)^n$. The lhs reads
    \begin{align*}
        \cop \sum_{\substack{\Gamma\, primitive\\ \res(\Gamma)=\bullet}}B_+^\Gamma\bigg(\xr \left(\frac{\xv}{(\xe)^2}\right)^{F_\Gamma}\bigg)
        =\cop \xr
        =\sum_k\xr \left(\frac{\xv}{(\xe)^2}\right)^k\otimes c_k^\bullet.
    \end{align*}
    The first term on the rhs of \eqref{Hochschild} reads 
    \begin{align*}
        \sum_{\substack{\Gamma\, primitive\\ \res(\Gamma)=\bullet}}B_+^\Gamma\bigg(\xr \left(\frac{\xv}{(\xe)^2}\right)^{F_\Gamma}\bigg)\otimes \bullet=\xr\otimes \bullet.
    \end{align*}
    The second term reads by Lemma \ref{lem:coprodPnk}
    \begin{align*}
        &(\id\otimes \sum_{\substack{\Gamma\, primitive\\ \res(\Gamma)=\bullet}}B_+^\Gamma)\cop\bigg( \xr \left(\frac{\xv}{(\xe)^2}\right)^{F_\Gamma}\bigg)\\
        =&(\id\otimes \sum_{\substack{\Gamma\, primitive\\ \res(\Gamma)=\bullet}}B_+^\Gamma)\bigg(\sum_k\xr \left(\frac{\xv}{(\xe)^2}\right)^{F_\Gamma+k}\otimes \bigg[\xr \left(\frac{\xv}{(\xe)^2}\right)^{F_\Gamma}\bigg]_k\bigg)\\
        =&\sum_{k,n} \xr \left(\frac{\xv}{(\xe)^2}\right)^{n+k}\otimes \sum_{\substack{\Gamma\, primitive\\ \res(\Gamma)=\bullet\\ F_\Gamma=n}}B_+^\Gamma\bigg(\bigg[\xr \left(\frac{\xv}{(\xe)^2}\right)^n\bigg]_k\bigg)\\
        =&\sum_{l\geq 1}\xr \left(\frac{\xv}{(\xe)^2}\right)^{l}\otimes c_l^\bullet,
    \end{align*}
    where in the last step we have used 
    \[
    %\displaystyle 
    \sum_{n}\sum_{\substack{\Gamma\, primitive\\ \res(\Gamma)=\bullet\\ F_\Gamma=n}}B_+^\Gamma\bigg(\bigg[\xr \left(\frac{\xv}{(\xe)^2}\right)^n\bigg]_k\bigg)=c_{n+k}^\bullet
    \]
    since the primitive graph is of loop order $n$ and the inserted graph of loop order $k$ we have to sum over all $n,k$ to generate $c_{n+k}$ due to the combinatorial DSE \eqref{2PcDSE}. % and \eqref{4PcDSE}.
\end{proof}
\end{theorem}

Theorem \ref{Thm:Hochschild} is stated in a slightly weaker way than it is generally stated in the literature, see for instance \cite{Kreimer:2005rw,Kreimer:2009iy}. 
The stronger statement in literature is that already the operator $\graft^{\bullet,n}$ satisfies the Hochschild 1-cocycle property, which is defined by 
\begin{align}\label{graftingn}
    \graft^{\bullet,n} := \sum_{\substack{\Gamma\, primitive\\ \res(\Gamma)=\bullet\\ F_\Gamma=n}}\graft^\Gamma.
\end{align}
The difference to Theorem \ref{Thm:Hochschild} is that $\graft^{\bullet,n}$ of \eqref{graftingn} has fixed loop order $F_\Gamma=n$ of all primitive graphs $\Gamma$, whereas Theorem \ref{Thm:Hochschild} has an additional sum over $n$ on top. 
However, to the best of our knowledge, there is no satisfactory proof in the literature that already \eqref{graftingn} satisfies the Hochschild 1-cocycle property. 
%It is even worse, in the literature where it is claimed to be proved a further extension is claimed to hold, namely that already $B_+^\Gamma$ satisfies the Hochschild 1-cocycle property for an individual primitive graph $\Gamma$. This statement is definitely not true. 
%We emphasize that this not only fails in our $\phi^4$ Connes-Kreimer Hopf algebra, but it also fails in several other Connes-Kreimer Hopf algebras, see for instance \cite[(101) and (105)]{Kreimer:2005rw} 
%different primitive generate the same graphs.

The Hochschild 1-cocycle property of \eqref{graftingn} is a desirable property due to further understanding of the structure of the Connes-Kreimer Hopf algebra. 
For this reason, \vtwo{we provide the following necessary and sufficient property for} 
%an equivalent statement to
the statement that \eqref{graftingn} is Hochschild:
\begin{proposition}\label{prop:Hochschildn}
    Let $\graft^{\bullet,n}$ be as in \eqref{graftingn}. Then, $\graft^{\bullet,n}$ satisfies the Hochschild 1-cocycle property, that is 
    \begin{align}\label{Hochschildn}
    \cop \graft^{\bullet,n}=\graft^{\bullet,n}\otimes \bullet+(\id\otimes \graft^{\bullet,n})\cop 
\end{align} 
acting on $\left(\frac{\xv}{(\xe)^2}\right)^n \xr$, if and only if for each single graph $\rg\in \hfd$ all %primitive 
cographs \vtwo{$\rg/\sg$ which are primitive} 
have the same loop number. %order. 
\begin{proof}
    Let us construct the following sub-series generated by the insertion into primitives of a fixed loop order $n$, that is
    \begin{align}\label{defXbn}
        X^{\bullet,n}:=\graft^{\bullet,n}\left(\left(\frac{\xv}{(\xe)^2}\right)^n\xr\right)=\sum_{\substack{\Gamma\, primitive\\ \res(\Gamma)=\bullet\\ F_\Gamma=n}}\graft^\Gamma\left(\left(\frac{\xv}{(\xe)^2}\right)^n\xr\right)
    \end{align}
    where we have by linearity $\xr=\sum_n X^{\bullet,n}$. Assuming that the Hochschild 1-cocycle property holds for $\graft^{\bullet,n}$, we can derive the coproduct of $X^{\bullet,n}$
    \begin{align*}
        \cop X^{\bullet,n}=&\cop\graft^{\bullet,n}\left(\left(\frac{\xv}{(\xe)^2}\right)^n\xr\right)\\
        =&\graft^{\bullet,n}\left(\left(\frac{\xv}{(\xe)^2}\right)^n\xr\right)\otimes \bullet+(\id\otimes \graft^{\bullet,n})\cop\left(\left(\frac{\xv}{(\xe)^2}\right)^n\xr\right)\\
        =&X^{\bullet,n}\otimes \bullet+(\id\otimes \graft^{\bullet,n})\bigg(\sum_k\xr \left(\frac{\xv}{(\xe)^2}\right)^{n+k}\otimes \bigg[\xr \left(\frac{\xv}{(\xe)^2}\right)^{n}\bigg]_k\bigg)\\
        =&X^{\bullet,n}\otimes \bullet+\sum_{k}\xr\left(\frac{\xv}{(\xe)^2}\right)^{n+k}\otimes (X^{\bullet,n})_{n+k}.
    \end{align*}
    Note that $k$ starts from 0 and there is no contribution for negative $k$ since $(X^{\bullet,n})_{l}$ is the $l$ loop order of $X^{\bullet,n}$ which is zero for $l\leq n$. Summing over $n$ yields the coproduct of $\xr$ (Proposition \ref{prop:coprodX}). Now, the only primitives on the rhs of the tensor product need to be of loop order $n$, since by definition \eqref{defXbn} all primitives in $X^{\bullet,n}$ are at loop $n$. Any graph with higher loop order than $n$ has already a nontrivial inserted graph. Thus assuming Hochschild for $\graft^{\bullet,n}$, the reduced coproduct of each single graph $\rg\in \hfd$ (let say it is contained in $X^{\bullet,n}$) has just primitives at loop order $n$, and all primitive cographs have the same loop order~$n$.

    Now, we want to prove the converse. Let us assume that for each graph in $\hfd$ all primitive cographs have the same loop order. Let $\RG{\text{co-}n}{}$ be the set of single graphs where all primitive cographs have just loop order $n$. We define
    \begin{align*}
        \tilde{X}^{\bullet,n}=\sum_{\substack{\rg\in\RG{\text{co-}n}{}\\ \res(\rg)=\bullet}} \rg.
    \end{align*}
    By linearity, we have $\xr=\sum_n\tilde{X}^{\bullet,n}$. Note also that a graph contained in $\tilde{X}^{\bullet,n}$ must be different of any graph in $\tilde{X}^{\bullet,m}$ for $n\neq m$ since they have different coproducts. By the central identity  \cite[Thm. 4.6]{Thurigen:2102} we have
    \begin{align*}
        \cop\tilde{X}^{\bullet,n}=\tilde{X}^{\bullet,n}\otimes \bullet+\sum_{k}\xr\left(\frac{\xv}{(\xe)^2}\right)^{k}\otimes (\tilde{X}^{\bullet,n})_{k},
    \end{align*}
    where $k\geq n$. Next note that $(X^{\bullet,n})_n$ and $(\tilde{X}^{\bullet,n})_n$ just consists of all primitives of loop order $n$, i.e.
    \begin{align*}
        (X^{\bullet,n})_n=(\tilde{X}^{\bullet,n})_n=\sum_{\substack{\Gamma\, primitive\\ \res(\Gamma)=\bullet\\ F_\Gamma=n}}\gamma,
    \end{align*}
    from which one conclude inductively by the coproduct formula that $X^{\bullet,n}=\tilde{X}^{\bullet,n}$. The central identity of $\tilde{X}^{\bullet,n}$ becomes therefore an equivalent statement to the Hochschild 1-cocycle property of $X^{\bullet,n}$.

\end{proof}
\end{proposition}

Proposition \ref{prop:Hochschildn} and its proof discuss the property of $\graft^{\bullet,n}$ being Hochschild 1-cocycle which comes together with an additional explicit decomposition of the combinatorial perturbation series into sub-series $X^{\bullet,n}$
\begin{align*}
    \xr=\sum_nX^{\bullet,n},\qquad X^{\bullet,n}:=\sum_{\substack{\gamma\, primitive\\ \res(\gamma)=\bullet\\ F_\gamma=n}}\graft^\gamma\left(\left(\frac{\xv}{(\xe)^2}\right)^n\xr\right).
\end{align*}
Thus if $\graft^{\bullet,n}$ is Hochschild, the perturbation series $\xr$ splits into a direct sum of $X^{\bullet,n}$. The perturbation series $X^{\bullet,n}$ is built by insertions of graphs into primitive graphs just of loop number $n$. 
Since there are infinitely many loop orders where primitive graphs occur for $\xv$ (Lemma \ref{lemma:primitives4}), we find infinitely many sub-series  $X^{v,n}$ where each consists of infinitely many graphs. 
It seems quite remarkable that any graph of a renormalizable QFT should have primitive cographs just with constant loop number. However, it seems to be believed for the ordinary $\phi^4$ QFT,  QED and gauge theories in general \cite{Kreimer:2005rw}. 
In the core Hopf algebra \cite{Kreimer:2009iy} this statement is certainly true since all primitive graphs have the same loop number, that is one.

However, to the best of our knowledge this question has never been studied in details. Prop.~\ref{prop:Hochschildn} provides therefore a new point of view which might help to give rigorous proofs that $\graft^{\bullet,n}$ is Hochschild plus a new conjectured algebraic structure for a renormalizable QFT. 

Similar questions have been studied considering the poset and lattice structure of subdivergencies in QFT \cite{Figueroa:0408, Berghoff:1411, Borinsky:2015mga}. 
Local $\phi^3$ and $\phi^4$ scalar QFT are known to have a lattice structure, whereas the $\phi^6$ theory has not. 
Furthermore, \cite{Borinsky:2015mga} shows further for $\phi^3$ and the 4-point function in $\phi^4$-theory this lattice is semimodular which yields a grading both in loop number and complete forests. 
Despite this, it is not possible to conclude that for each graph there are only primitive cographs with the same loop number. 
This structure (if it is true) would be a further algebraic structure compatible with the lattice structure. 
Hopefully, similar methods might allow to tackle also the question if (and which) renormalizable QFT consists of graphs which have primitive cographs with a constant loop number.

% \newpage

\appendix

\section{Feynman diagrams, ribbon graphs and maps}
\label{sec:ribbon graphs}

For precise statements we have introduced in Def.~\ref{def:ribbon graph} a combinatorial definition of ribbon graphs.
This definition is closely related to the notion of a combinatorial map \cite{Gurau:2014, Eynard:2016yaa}.
In this appendix we give the details how a combinatorial ribbon graph, upon completion, is dual to a combinatorial map, with vertices taken as faces of the map and vice versa.

\begin{definition}%[combinatorial map]
\label{def:map}
    A \emph{combinatorial map} $M=(\H, \H^*,\sigma,\ei)$ consists of finite sets $\H$ and $\H^*$ of unmarked and marked \emph{half edges}, \vtwo{respectively}, and two permutations $\sigma$ and $\ei$ on $\H\sqcup\H^*$ where $\ei$ is an involution \emph{without} fixed points. 
    \vtwo{Cycles of $\ei$ define edges} 
    and cycles of $\sigma$ define vertices.
    Cycles of $\phi=\sigma^{-1}\circ\ei$ are called \emph{faces} 
    \vtwo{and they contain at most one marked half edge $h^*\in\H^*$ each, also called a root.  
    Faces containing a root are \emph{boundary faces}. 
    }

    A map is called \emph{fully simple} if 
%2    no more than two edges belonging to any boundary are incident to a vertex of $M$.
    %3 any vertex of $M$ is incident to at most two edges incident to any boundary.
    %1 any vertex is incident to at most two edges which belong to some boundary face.
    any vertex \vtwo{is incident} to at most one boundary and is incident to at most two edges which belong to the boundary.%
\footnote{This is a slight modification of the definition given in \cite{Borot:2017agy} which excludes also peculiar cases as \inftymap which seem to be permitted in their definition, but should actually not be included.}
\end{definition}

%(K) : add example of a fully simple map?? - too much, I think

In this definition we have used the convention that permutation cycles define counter-clockwise orientation to vertices and faces and that half edges considered as darts pointing outward from vertices belong to the respective left face (as e.g. in \cite{Garcia-Failde:2018ylj}; alternatively,\cite{Eynard:2016yaa} for example defines faces on the right of their darts with clockwise orientation but vertices as counter-clockwise). 

\begin{definition}
\label{def:dual}
Given a combinatorial map $M=(\H, \H^*,\sigma,\ei)$, there is a dual map $M_\star = (\H, \H^*,\sigma_\star=\ei\circ\sigma,\ei_\star=\ei)$.
\end{definition}
Duality is in fact an involution, $(M_\star)_\star = M$, due to $\ei$ being an involution.
However, since duality interchanges vertices and faces,
for each boundary face with $n$ vertices in $M$ there is a boundary vertex incident to $n$ boundary faces in the dual map $M_\star$. 
Thus, there is a different definition of boundary for dual maps: each vertex incident to a root defines a boundary in the sense that all faces incident to it belong to this boundary.
Furthermore, the dual map corresponds to opposite orientation: For half-edges pointing outward from vertices one has to consider faces on the right to the half-edge they contain and vertices and faces have clockwise orientation since $\sigma_\star=\phi^{-1}$ and $\phi_\star\equiv\sigma_\star^{-1}\circ\alpha=\sigma^{-1}$. 
Thus one calls a map explicitly a \emph{dual map} to account for these differences in interpretation.

\

%\begin{example}
To illustrate this duality of maps, consider the combinatorial map corresponding to the fish graph in MFT:
\[
M \cong \fishmap  \, .
\]
One can explicitly define $M=(\H = \{1,2,3,4,5,6,7,8,2',6',7'\}, \, \H^* = \{1'\}, \sigma, \ei)$ with the cycles of $\sigma$ and $\ei$
\begin{align}
    \mathcal{C}(\sigma) &= (1'2)(2'63)(6'7)(187')(45)  \\
    \mathcal{C}(\ei) &= (1'1),(2'2)(35)(48)(6'6)(7'7) \, ,
\end{align}
such that the faces are
\[
\mathcal{C}(\phi) = \mathcal{C}(\sigma^{-1}\circ\ei) = (1234)(5678)(7'6'2'1')
\]
and $(7'6'2'1')$ is the boundary since $1'$ is the single root of $M$.
Then the dual $M_\star=(\H,\H^*,\mathcal{C}(\sigma_\star) = (4321)(8765)(1'2'6'7'),\ei)$ can be pictured as
\[\label{eq:fishmapdual}
\fishmapdual \rightarrow M_\star \cong \fishdual \,.
\]
While labels are on the left of half edges in $M$, they have to be on the right in the dual $M_\star$ so that the dual edges carry the same labels, $\ei_\star = \ei$.
In $M_\star$, the vertex $(1'2'6'7')$ corresponds to the boundary.
%\vtwo{faces not yet defined!}
%and all faces, i.e.~cycles of $\phi_\star= \sigma^{-1}$, except for $(45)$ are thus external faces.

\

That boundary faces of maps become vertices in the dual map %And each vertex at the boundary becomes an external face. 
is different to physics where QFT Feynman diagrams have open external edges.
The usual MFT ribbon graph is therefore the \emph{decompletion} of a dual combinatorial map with respect to its boundary vertices, that is the object obtained by deleting all boundary vertices and incident half-edges.
In the example \eqref{eq:fishmapdual}, decompletion of~$M_\star$ yields the planar fish ribbon graph
\[%\label{eq:fish graph}
\fishl \equiv \left(\{1,2,3,4,5,6,7,8\},(4321)(8765),(1)(2)(35)(48)(6)(7))\right)
\]
From this perspective, one can even give meaning to the ribbon graph $G=(\emptyset, \sigma, \alpha)$ as an open edge ``$\edge$'' being the unique decompletion of the dual map
\[\label{eq:edgecompletion}
\overline{G}= M_\star = \left(\H=\{2'\},\H^*=\{1'\},(1'2'),(1'2') \right) \cong \edgecompletion %\rightarrow G = (\emptyset, \sigma, \alpha) \cong \edge \, . 
\]
which is the dual to $M \cong \edgemap$.

Thus, a ribbon graph in the combinatorial sense of Def.~\ref{def:ribbon graph} is a combinatorial map upon \emph{completion} of external half-edges and vertices:

\vtwo{
\begin{definition}\label{def:completion}
Given a ribbon graph $G=(\H,\sigma,\ei)$ with $\H=\Hint\sqcup\Hext \ne \emptyset$, 
we define the \emph{completion} $\overline{G}=(\Hc,\sigmac,\eic)$ in the following way:
\begin{itemize}
    \item $\Hc%\sqcup\Hc^* 
    := \H\sqcup\Hext'$ where $\Hext'\cong\Hext$ is defined by doubling the external half edges, i.e.
    for every $h\in\Hext$ there is an $h'\in\Hext'$ with $\eic(h):= h'$.
    \item For all internal half-edges $h\in\Hint$ define $\eic(h):=\ei(h)$.
    \item Internal vertices of $\overline{G}$ are the vertices of $G$, that is $\sigmac\vert_{\H}=\sigma$.
    \item Boundary vertices of $\overline{G}$ are the cycles $\mathcal{C}(\sigmac\vert_{\Hext'}):=f_{\eic} \left(\mathcal{C}(\ei\circ\sigma)\vert_{\Hext}\right)$,
     where the restriction to $\Hext$ on cycles means deleting all $h\in\H\setminus\Hext=\Hint$ in the cycles of $\mathcal{C}(\ei\circ\sigma)$ and $f_{\eic}$ operates on the cycle set by mapping each $\Hext\ni h \mapsto h'\in\Hext'$ in each cycle.
%    \item In each cycle of a boundary vertex, mark one $h'$. These marked half edges constitute $\Hc^*$ and in this way define the partition $\Hc\sqcup\Hc^*$
\end{itemize}
\end{definition}
The rationale for the construction is to mirror each boundary. Therefore, for each boundary given by a cycle of $\sigma^{-1}\circ \alpha$ which contains an $h\in\Hext$, one flips the orientation to $(\sigma^{-1}\circ \alpha)^{-1}=\alpha\circ\sigma$;
then one eliminates internal half edges;
and finally maps the external half edges $h\in \Hext$ to their partners $h'=\eic(h)\in\Hext'$.
These resulting cycles on $\Hext'$ define the boundary vertices $\mathcal{C}(\sigmac)$, and thereby $\sigmac$.
}

\begin{proposition}\label{prop:completion as fully simple map}
\vtwo{Up to choosing the marked edges, 
the completion $\overline{G}$ of} a ribbon graph $G$ 
%has a unique completion~$\overline{G}$ which 
is a combinatorial map dual to a fully simply map.
\end{proposition}

\begin{proof}
Let $G=(\H=\Hint\sqcup\Hext,\sigma,\ei)$ be a ribbon graph 
\vtwo{and $\overline{G} = (\Hc=\H\sqcup\Hext',\sigmac,\eic)$ its completion according to Def.~\ref{def:completion}.
The cycles of $\sigmac$ on $\Hext'$ define boundary vertices but they do not yet contain a marked half edge. 
Thus, to obtain a dual combinatorial map, choose one $h\in\Hext'$ for each such cycle as the root, defining a partition of $\Hc$ into unmarked and marked half edges.
Now this} defines a dual combinatorial map since $\eic$ is an involution without fixed points, $\sigmac$ a permutation, and roots designate boundary vertices.
By construction of the completion Def. \ref{def:completion}, 
every external face is incident to exactly two external edges 
(since the boundary vertex cycles are constructed to mirror the boundary external faces of $G$ thereby closing them along two new edges). 
This means that in the dual map $\overline{G}_\star$ a vertex at a boundary face is incident to two boundary edges part of only that boundary. 
This is the defining property of a fully simple map, Def.~\ref{def:map}.
Thus, $\overline{G}$ with arbitrarily chosen roots on the boundary vertices is dual to a fully simple map.
\end{proof}

Full understanding of the relation of ribbon graphs as occurring in MFT to combinatorial maps finally allows  for a straightforward definition of (external) \emph{faces} of a ribbon graph $G$: 

\begin{definition}
    Let $G=(\Hint\sqcup\Hext,\sigma,\ei)$ be a ribbon graph and $\overline{G}=(\Hc,%\vtwo{=\H\sqcup\Hext'},
    \sigmac,\eic)$ its completion.
    A \emph{face} of $G$ is a cycle of $\sigmac^{-1}\circ\eic$. % restricted to $G$.
    A face is external iff it contains an $h\in\Hext$, or equivalently iff it is incident to a boundary vertex in $\overline{G}$.
    Otherwise it is an internal face.
    We denote the number of internal faces $F_G$.
\end{definition}
It follows that each connected component $j=1,...,b$ of the boundary of $\overline{G}$ has $n_j$ external faces which is the degree of the respective boundary vertex.
In contrast, in the corresponding ribbon graph $G$ the external edges are open and thus there is only a single cycle of $\sigma^{-1}\circ\alpha$ for each boundary component.

\bibliographystyle{halpha-abbrv}
\bibliography{main}

\end{document}